\documentclass[a4paper,11pt]{report}
\usepackage{amsthm,amssymb,amsmath,mathrsfs}
\usepackage[top=35mm, bottom=25mm, left=35mm, right=25mm]{geometry}
\usepackage{graphicx}
\usepackage[all]{xy}
\usepackage{fancybox}
\usepackage{fancyhdr}
\usepackage{tocbibind}
\usepackage[onehalfspacing]{setspace}
\usepackage{amssymb}
\usepackage{supertabular}
\usepackage{dsfont}
\usepackage{tikz}
\usepackage{amsthm}                % better theorem environments

\usepackage{indentfirst}

\usepackage{makeidx}
\makeindex
\usepackage{color}
\usepackage{booktabs}
\usepackage{bm}
\usepackage{float,lscape}
\usepackage{etoolbox}

\theoremstyle{plain}
\newtheorem{thm}{Theorem}[section]
\newtheorem{cor}[thm]{Corollary}
\newtheorem{lem}[thm]{Lemma}
\newtheorem{prop}[thm]{Proposition}
\newtheorem{example}[thm]{Example}
\theoremstyle{definition}
\newtheorem{ch}{Church-Turing thesis}[section]
\newtheorem{defn}{Definition}[section]
\theoremstyle{remark}
\newtheorem{rem}{Remark}[section]
\numberwithin{equation}{section}

\newcommand{\z}{\mathbb{Z}}
\newcommand{\com}{\mathbb{C}}
\newcommand{\q}{\mathbb{Q}}
\newcommand{\rr}{\mathbb{R}}
\newcommand{\n}{\mathbb{N}}
\newcommand{\h}{\mathbb{H}}

\newcommand{\bi}{\mathbf{i}}
\newcommand{\0}{\mathbf{0}}
\newcommand{\bj}{\mathbf{j}}
\newcommand{\bo}{\mathbf{\Omega}}

\newcommand{\ion}{\textbf{i}^{(n)}\in\Omega_2^{(n)}}
\newcommand{\p}{\mathcal{P}_{\ii}}
\newcommand{\ii}{\textbf{i}^{(n)}}

\newcommand{\jj}{\textbf{j}^{(n)}}
\newcommand{\cc}{\mathds{C}}
\newcommand{\zz}{\mathds{Z}}

\newcommand{\A}{\mathcal{A}}

\newcommand{\one}{\mathds{1}}
\newcommand{\ket}[1]{\vert#1\rangle}

\newcommand{\oh}{\overline{H}}
\newcommand{\lh}{\underline{H}}

\begin{document}
%\newpage
\pagenumbering{gobble}
\thispagestyle{plain}
\begin{center}
\vspace*{-2.5cm}\includegraphics{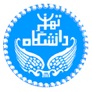} \\
\Large{University of Tehran} \\[0.03cm]
%\Large{College of Science}\\[0.02cm]
\Large{School of Mathematics, Statistics and Computer Science}
\vskip 2.5cm
{\LARGE\textbf {  Gacs Algorithmic Complexity on Hilbert Spaces \\ and\\ Some of its Applications }}
\vskip 1cm
{\Large By}\\
{\Large \textbf{Samad Khabbazi Oskouei }}\\
\vskip 0.5cm
{\Large Under}\\
\vskip 0.5cm
{\Large Supervisions of}\\
{\Large\textbf{Ahmad Shafiei Deh Abad}}\\
\vskip 0.5cm
and\\
{\Large\textbf{ Fabio Benatti}}\\
\vskip 2cm
{\Large A thesis submitted to the \\School of Mathematics, Statistics and Computer Science\\
In partial fulfillment of the requirements for\\
The degree of Doctor of Philosophy in}\\
\vskip 0.02cm
{\Large \textbf{Applied Mathematics}}\\
\vskip 1.2cm
{\Large \textbf{April 2015}}\\
\end{center}

\begin{center}
\includegraphics[width=16 cm]{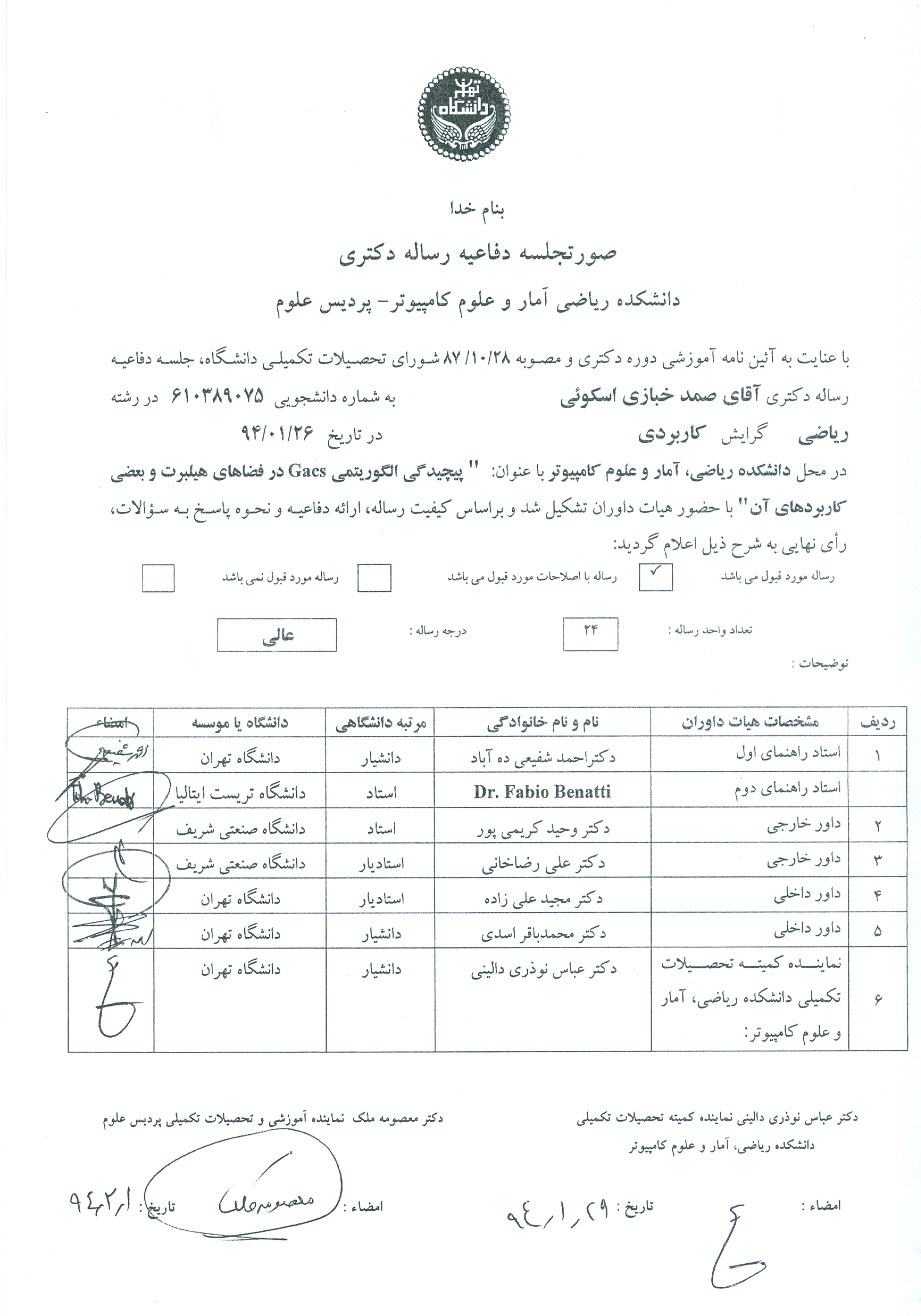} 
\end{center}
\newpage
\setcounter{tocdepth}{3}
\baselineskip=1 cm
\textbf{\large{Abstract}}
\smallskip

We extend the notion of Gacs quantum algorithmic
entropy, originally formulated for finitely many qubits,  to infinite dimensional quantum spin chains and investigate the
relation of this extension with two quantum dynamical entropies
that have been proposed in recent years. Further, we  prove an extension of Brudno's theorem in quantum spin chains with shift dynamics.

\smallskip

% ----------------------------------------------------------------------

\clearpage
%\newpage
%\include{preface_en}
%\prefacesection{\textbf{Acknowledgements}}
\baselineskip=1cm
\textbf{\large{Acknowledgements}}

\smallskip
I would like to express my special appreciation and thanks to my advisors Professor Ahmad Shafiei Deh Abad and Professor Fabio Benatti, you have been a tremendous mentor for me. I would like to thank you for encouraging my research and for allowing me to grow as a research scientist. Your advice on both research as well as on my career have been priceless. I would also like to thank my committee members, professor Vahid Karimipour, Dr. Ali Rezakhani, Dr. Majid Alizade and Dr. Mohammad Bager Asadi for serving as my committee members even at hardship. I also want to thank you for letting my defense be an enjoyable moment, and for your brilliant comments and suggestions, thanks to you.

A special thanks to my family. Words cannot express how grateful I am to  my mother for all of the sacrifices that you’ve made on my behalf. Your prayer for me was what sustained me thus far.  I would like also express appreciation to my beloved wife Elham who spent sleepless nights with and was always my support in the moments when there was no one to answer my queries.

At the end, I would like the support  this thesis by the STEP programme of the Abdus Salam ICTP of Trieste. 
\smallskip
\pagenumbering{roman}
\tableofcontents
%\renewcommand{\contentsname}{}
%roman Roman alph Alph
\cleardoublepage
\doublespacing
%\pagenumbering{roman}
%\addcontentsline{toc}{chapter}{List of Abbreviations and Acronyms}
%\include{acronyms}
%\listoftables \listoffigures %\listofsymbols
\cleardoublepage
 %\doublespacing \pagenumbering{arabic}
\addcontentsline{toc}{section}{Introduction} % what looks better section or chapter?$
%%% Thesis Introduction --------------------------------------------------

%\nonumchapter{\textbf{Introduction}}
\baselineskip=1cm\textbf{\large{Introduction}}
\goodbreak
\paragraph{Computability theory:}
Algorithms and computational techniques started to be studied at
least since  the  Babylonians and later by Euclid (c. 330 B.C).
But, only around the 20ies, mathematicians could
successfully formalize these concepts which then applied to
modern computers. The subject ripened with Godel's  Incompleteness Theorem 1931 which solved with NO the first
Hilbert  problem:
\begin{center}
  \emph{Can a formal system prove its own consistency?}
\end{center}
 In addition,  the result opened new ways to solve the second
Hilbert problem, namely the decision problem. Indeed, Godel  invented the theory of primitive recursive functions and later  extended it to general recursive functions in 1934.
On the other hand, formal definitions of computability  were given in the
mid-1930's by Kleene and  Turing. However, a major breakthrough in computability theory occurred in 1936 with Turing's work on recursion theory \cite{turing}. Since then, computability had  not only a fundamental role in computer
science but  also  many applications to logic, algebra,
analysis, and   , nowadays, it plays a role in as different fields as   physics and economy.

\paragraph{Kolmogorov complexity theorey:}
it was born as an attempt to answering questions of the following kind
\begin{center}
 \emph{ When  can mathematical objects
such as finite or infinite binary strings be termed random? }
\end{center}
and
\begin{center}
\emph{Given two binary
strings, how can one decide which one of them is more random than another?}
\end{center}
  Clearly, the issue at stake here is how  to measure the randomness of  strings? The measures which are used in computability
theory and algorithmic information theory should  explore the
relationships among three fundamental aspects: 1) relative
computability, as measured by notions such as Turing reducibility;
2) information content, as measured by notions such as Kolmogorov
complexity; 3) randomness of individual objects, as first
successfully defined by Martin-L\"{o}f \cite{martinlof} (but
prefigured by others, dating back at least to the work of von Mises
\cite{vonmises}). In this thesis, we will focus on the second aspect: informational string content and descriptional complexity..
Let us consider two sequences such as
$$101010101010101010101010101010101010 \ldots $$
$$101101011101010111100001010100010111 \ldots $$
 Most people  may  agree that, intuitively, the second binary string is more random than the first one.
But, from a mathematical point of view, how can we give solid ground to such an intuition?
Or why should some sequences count count as �random� and others as
�regular�, and how can we translate our intuitions about these
concepts into meaningful mathematical notions?
Algorithmic complexity, or descriptive complexity, or, in the
following, Kolmogorov complexity, developed by Solomonoff
\cite{Solomonoff}, Kolmogorov \cite{Kolmogorov} and Chaitin
\cite{Chaitin}, has been an important tool in many different fields
\cite{Downey,Nei,Vitany,Calude} where it shed light on subtle
concepts such as information content, randomness, inductive
inference and also had applications in thermodynamics. In a
nutshell, the complexity of a target object is measured by the
difficulty to describe it; in the case of targets describable by
binary strings, they are algorithmically complex when their shortest
binary descriptions are essentially of the same length in terms of
necessary bits, the descriptions being binary programs such that any
universal Turing machine that runs them outputs the target string.

\paragraph{Ergodic theory:}
Ergodic theory  goes back to  Boltzmann and Gibbs. It provides a successful mathematical framework for the description of dynamical
systems. It gives a probabilistic approach to dynamics  that is
useful to investigate statical properties of the  evolution of a
mechanical system over long  time scales. Further, ergodic theory  \cite{khinchin} explains why in
thermodynamical systems mean values of observables  coincide   with
time averages  and why trajectories in ergodic
systems  fill the phase-space  densely.
The KS-entropy, introduced by Kolmogorov and developed by
 Sinai \cite{billingsley}, defined on classical dynamical systems has  provided a link among
different fields of mathematics and physics. In fact, in the light
of the first theorem of Shannon \cite{cover}, the KS entropy gives the maximal
compression rate of the information emitted by ergodic information
sources. A theorem of Pesin \cite{mane} relates it to the positive Lyapounov
exponents and thus to the exponential amplification of initial small
errors, in other words, to classical chaos. Finally, a theorem of Brudno \cite{Brudno}
links the KS entropy to the compressibility of classical
trajectories by means of computer programs, namely to their
Kolmogorov complexity. In fact,  Brudno's theorem establishes relations among all the above mentioned issues.

\paragraph{}
Since Quantum Mechanics teaches us that the basic structure of the world is non-commutative and because of the fast development of quantum computation and information theory, it has become important to extend such classical notions to quantum dynamical systems.

\paragraph{Quantum mechanics:}
Quantum mechanics developed by Heisenberg, Schrodinger,
Dirac and others in the 20ies  describes the behavior of elementary, particles, atoms and molecules.  In quantum mechanics, dynamical systems are
described by  a  Hilbert space, whose vectors provide their physical states, and Hamiltonian self-adjoint operators that generate their dynamics.
By  the Stone$-$von Neumann uniqueness theorem \cite{beratteli}, this
Hilbert space description is suitable for the systems with finite
number of freedoms. Since with infinitely many degrees  of freedom, for a example one dimensional lattice  $\zz$ the sites carrying each $d$-level spins,  we do not have  this properties then considering the
$C^*$-algebra of observables  is more convenient.

\paragraph{Quantum dynamical entropies:}
Since there
 are different approaches to the information
content in quantum systems ,there are as well several extensions of the
KS-entropy to quantum dynamical systems. One  aim of these
extensions is to classify quantum dynamical systems as  done in classical dynamical systems by the KS-entropy. Recently,
they have been used in quantum information theory in relations to the capacity of quantum channels and quantum algorithmic complexity. Two quantum dynamical
entropies, one proposed by Connes, Narnhofer and Thirring  (CNT entropy) \cite{CNT}
and the other one introduced by Alicki and Fannes  (AF entropy) \cite{Alicki},   have been more used than the others. The
two entropies are defined differently from each other, and they may
exhibit a different behavior on a same quantum dynamical systems.

\paragraph{Quantum Turing machines}
A fundamental goal in computer technology is to construct computer devices with
high speed and  low prices, what has implied a steady increase in miniaturization. In view of this fact, information processing under quantum mechanical rules is becoming a concrete and substantial issue \cite{braunstein}. The first suggestion of quantum computers
was given by  Feynman who predicted that quantum computers
 might provide   more efficient computing devices  than classical (probabilistic) computers. Once these advantages have been demonstrated by the first quantum algorithms, quantum computation and quantum information theory started blossoming \cite{ gruska, kaye_Laflamme_Mosca, ohya}.

In view of the importance of classical Turing machines for the development of classical computability theory, it soon became important to extend these notions to the quantum realm:  quantum Turing machines (QTM) and universal quantum
Turing machines (UQTM) were thus introduced in \cite{e_bernstein_vazirani}.

The subsequent question was how to reformulate  the notion of algorithmic complexity in a way that it could be used for quantum systems, too. Several proposals have been put forward that reflect different points of view. However, all of them have the same basic
intuitive idea that complexity should characterize properties of
systems that are difficult to describe. They can roughly be
summarized as follows:
\begin{enumerate}
    \item Qbit quantum complexity: one may decide to describe quantum states by means of other quantum states that are processed by UQTMs \cite{Berthiaume}: the corresponding complexity will be denoted by  $\text{QC}_\text{q}$.
    \item  Bit quantum complexity: one may decide to define
    the complexity of quantum states using classical
         \cite{Vitanyi} programs run by UQTMs which is denoted by
         $QC_c$.
    \item  Quantum circuit complexity: one may choose to relate the complexity of a quantum state to the
         complexity of the (classical) description of the quantum circuits
         that operatively construct the state \cite{Mora, Mora2}.
          The
         corresponding complexity will be denoted by $\text{QC}_\text{net}$.
    \item Gacs complexity: one may  extend the notion of universal probability and define
     a quantum universal semi-density matrix \cite{Gacsl, FSA} and then, mimicking the classical approach, introduce a quantum complexity as minus its logarithm.
This thesis is based on exactly this latter train of ideas
        \end{enumerate}

\paragraph{Thesis subject:}
The recent developments in quantum mechanics that, together with the
birth of the so-called quantum computation theory, have also led to
the development of a broad quantum information theory, have spurred
the attempt to extend the concept of algorithmic complexity to the
quantum realm.
As we have seen, there exist different proposals of quantum
algorithmic complexity that, while agreeing on quantum states as
description targets, differ on how their description should be
achieved.

In all cases a useful guide to sort out the various quantum
extensions of algorithmic complexity is provided by the relations
between the classical algorithmic complexity and the Shannon
entropy. Even when not pretending to exactly reproducing them in a
non-commutative context, it is nevertheless important to clarify the
connections, if any, between the quantum algorithmic complexities
and the von Neumann entropy or related concepts. In particular, in
the classical setting a theorem of Brudno \cite{Brudno} states that
almost every trajectory of an ergodic classical system has an
algorithmic complexity rate which equals the Shannon entropy rate, the
latter being also known as Kolmogorov-Sinai, or dynamical entropy.
Inequivalent quantum extensions of the Kolmogorov-Sinai entropy have
also been proposed \cite{CNT, Alicki, Sl, Do}.

In \cite{Benatti}, a relation was established between the quantum
algorithmic complexity  as formulated in \cite{Gacs}, that we shall
refer to as Gacs complexity (entropy) in the following, and the quantum
dynamical entropies of the shift automorphism on quantum spin chains
as formulated by Connes, Narnhofer and Thirring (CNT-entropy)
\cite{CNT} and by Alicki and Fannes (AF-entropy) \cite{Alicki}. A
quantum spin chain is a one-dimensional lattice with $d$-level
quantum systems at each site and the lattice translations or
shift-automorphisms are the simplest possible dynamics. For ergodic translation invariant states $\omega$ on quantum
spin chains, the CNT-entropy equals the von Neumann entropy density
$s(\omega)$, while the AF-entropy equals $s(\omega)+\log d$.

 In \cite{Benatti}, the extra term $\log d$ is given an
informational interpretation in terms of the Gacs complexity per
spin in the Alicki-Fannes construction. There, the limit rate is
obtained starting from increasingly large, but finite-dimensional
sub-chains and using the formulation in \cite{Gacs} that concerns
arbitrary, but finite number of spins. As a consequence of the
construction of the complexity rate from below, that is from finite
dimensional sub-algebras to the infinite dimensional spin-chain, a
constraint had to be imposed in \cite{Benatti} on the growth of the
classical complexity of finite-size density matrices; namely, that
it be slower than the size of the sub-chain. Instead, in this thesis,
we construct a Gacs complexity quantity starting directly from the
infinite dimensional quantum spin chain. The resulting complexity is equivalent to the finite
dimensional one when restricted to finite portions of the chain, but
allows us to remove the unnecessary limitation mentioned above. As a
result, we report an instance of quantum spin-chain with finite
Alicki-Fannes entropy equalling the Gacs complexity rate, while
finite-size density matrices have Kolmogorov complexities diverging
faster than $n$.

Further, an extension of Brudno's theorem using the Gacs complexities are mentioned in this thesis. One way to extend it is to reformulate the lower Gacs complexity in classical dynamical systems and then  reformulate Brudno's theorem using it in quantum spin chains. Another way that we extend it is by the help of a generalization of the classical  Shannon-MacMillan theorem, or qauntum Shannon-MacMillan theorem, \cite{Bjelakovictheshannon-mcmillan} in ergodic quantum spin chains with shift dynamics. The two proposals  are mentioned in this thesis, the first one is just formulated in the classical case, where it reduces to a short proof of the Brudno's theorem.     While  the proof of Brudno's theorem in quantum spin chain results from the second one. Namely, the rate of lower Gacs complexity of minimal projections which are dominated by a sequence of projections  with high probability is equal to the von-Neumann entropy rate of the state.

\paragraph{}
The organization of the thesis is as follows:

\textbf{Chapter 1:}
We shortly introduce  computability theory based on a specific programming language $\cite{Da}$. We also briefly describe   Turing machines   and  define the Kolmogorov complexity  with an attached thermodynamical interpretation.

\textbf{Chapter 2:}
We explain how classical and quantum dynamical systems can be given a unifying algebraic description as commutative and non-commutative $C^*$-algebras, respectively. Then, we introduce the two quantum dynamical entropies, CNT and AF, which are extensions of the classical KS-entropy, their   relations and properties with applications to quantum spin chains. Most of the material in this chapter is taken from \cite{Benatti}.

\textbf{Chapter 3:}
We first extend the concepts of semi-computable semi-density matrices, Gacs entropy to  infinite dimensional separable Hilbert space, and  the apply them to quantum spin chains.

\textbf{Chapter 4:}
We introduce  classical version of Gacs complexity; then, using the semi-computability concept, we will give a short proof a restricted version of Brudno's theorem.

\textbf{Chapter 5:}
This final chapter is entirely devoted to the extension of Brudno's theorem to the case of the shift dynamics on quantum spin chains.

%%% ----------------------------------------------------------------------

\pagenumbering{arabic}

\chapter{Programs and Computable Functions }
\goodbreak
 In this chapter, we first introduce computability theory which plays an important role in   computer science: it will be done  by means of a specific programming language. So, we introduce  necessary concepts and tools
such as computable functions and partial computable. Further, we review the notion of  Turing machine which is the simplest mathematical model of computing device.

 Complexity in  computer science is usually either
computational or descriptional, where the first one refers to  the
number of needed computational steps in a given program and
the second one measures the amount of information  in a program.
In the following, we shall concentrate on the latter case. Finally, we will consider a thermodynamical application of Kolmogorov complexity to an oversimplified model of computing device which shows the relations between data compression and energy cost.
\goodbreak

\section{A Programming Language}
We are going to   introduce  computability theory  based on a specific
programming language $P$.

 This consists of the letters:
$$X_1, X_2, \ldots X_n,$$
which will be called input variables with values in $\n\cup\{0\}$.
The output variable will be denoted by  $Y$. In most
programs, we also need  local variables  $Z_1, Z_2, \ldots, Z_k$. Moreover, $P$
contains the following instructions.

\begin{enumerate}
    \item $V\rightarrow V+1$: Increase by 1 the value of the variable $V$.
    \item $ V\rightarrow V-1$: If the value of $V$ is zero leave it unchanged; otherwise decrease it by 1.
    \item IF $V\neq 0$ GOTO L: If the value of V is nonzero, perform the instruction with label L;  otherwise proceed to the next instruction in the list.
\end{enumerate}
The  labels
\begin{equation}\label{label}
A_1, B_1, C_1, E_1, A_2, B_2, C_2, E_2, A_3, \ldots,
\end{equation}
are used to indicate a specific instruction of a program,  a program $P$ being a finite list of above instructions.

A
program  can halt in ways: in the first one, there are no more instructions after the last one in the list which constitutes the program. In the
second case,  an  instruction labeled $L$ is to be executed, while,
there is no instruction with that label $L$ in the program; we
usually denote the label $L$ with the letter $E$.

\begin{example}
\label{pro} The following program computes the function $f( x, y ) =
x + y$.
\begin{align*}
&  & & Y\rightarrow X_1& \\
& &  &Z_1\rightarrow X_2&\\
&\text{[ B ]}&&\text{If}\,\, Z_1 \neq 0\,\,\text{GOTO}\,\,
A&\\
&&&\text{GOTO E}&\\
&\text{[ A ]}&& Z_1 \rightarrow Z_1 - 1&\\
  &  &   &    Y\rightarrow Y + 1&\\
&&&\text{GOTO B}&
\end{align*}
where GOTO  $E$  is an abbreviation for
\begin{align*}
        &Z_2 \rightarrow Z_2+1&\\
        &\text{IF}\,\, Z_2 \neq 0 \,\,\text{GOTO E}.
\end{align*}
Moreover, since there is no label $E$, the command $GOTO$ $E$ forces
the program to halt. Of course, the symbols $X_1, X_2$ denote input
variables, $Z_1$ a local variable, $Y$ the output variable, while $A,
B, E, L$ are labels.
\end{example}

We will show that programs can be assigned natural numbers in a specific way called Godel numbering. Namely, we will show that there exists a one-to-one correspondence  between
$\n$ and the set of all programs in a programming language $P$. The corresponding number of each  program $p$ is denoted by
$\#(p)$. In such a way, the program can be retrieved from its
number:

Let  the variables be listed as follows:
$$Y, X_1, Z_1, X_2, Z_2, X_3, Z_3, \ldots \,,   $$
and the labels be listed as in Example \ref{label}.
The number assigned to  a given variable is  its position
number in the above list.
 For example:
$\#(Z_1)=3$. The number assigned to a given label  is also  its position number.
Now, let $I$ be an instruction of the  program
$p$. Let's define for any $x, y\in\z$, $<x,y>=2^x(2y+1)-1$. Then,
the number assigned to $I$ is defined by
$$\#(I)=<a, <b, c>>,$$
where
\begin{itemize}
    \item if $I$ is labeled $L$ then $a=\#(L)$; otherwise $0$.
    \item if the variable $V$ is used in $I$, then $c=\#(V)-1$.
    \item if $I$ contains one of the following  statements
    $$v\rightarrow v,\quad v\rightarrow v+1,\quad v\rightarrow v-1,$$
    then $b=0$, $1$ or $2$, respectively.
    \item if the statement   $$\text{If}\,\, V \neq
    0\,\,\text{GOTO}\,\, L,$$ is used in $I$
    then $b=\#(L)+2$.
\end{itemize}
The number of the program $p$ consisting of the instructions $I_1, I_2,
\ldots I_k$ is defined by
$$\#(p)=[\#(I_1), \#(I_2), \ldots , \#(I_k)]:=2^{\#(I_1)}\cdot3^{\#(I_2)}\cdot \ldots\cdot p_k^{\#(I_k)}-1,$$
where $p_k$ is the $k$-th prime number. A program with the number $n$
will be denoted by $p_n$.
\begin{example}
\begin{eqnarray*}
[A] &&  X\rightarrow X+1\\
 &&\text{IF}\, X\neq 0\,\, \text{GOTO}\,\,A
\end{eqnarray*}
 The program contains two instructions, which will be called  $I_1$ and $I_2$,
respectively. Instruction $I_1$ is labeled by $A$ thus $a=\#(A)=1$,
$b=1$ and $c=\#(X)-1=1$; therefore, $\#(I_1)=21$. Since $I_2$ is
unlabeled,
$$\#(I_2)=<0,<3,1>>=46.$$
Finally, $\#(p)=[21, 46]=2^{21}\cdot3^{46}-1$.
\end{example}
\begin{defn}
A function $f:\n^n\to\n$ is called partially computable if there
exists a program $p$, which for each  $(x_1, x_2, \ldots,
x_n)\in\n^n$, halts on input $(X_1=x_1, X_2=x_2, \ldots, X_n=x_n)$
if and only if $f(x_1, x_2, \ldots, x_n)$ is defined and  its output
$Y$ is equal to $f(x_1, x_2, \ldots, x_n)$.
\end{defn}
 By the Godel numbering, the set of all
programs of a programming language is enumerated and hence the set of all
partially computable functions is also enumerated.

The function
$f:\n^n\to\n$ computed by a program $p_k$ is denoted by $\phi_k$.
Since one program may halt or not  on an input value,
 partially computable functions may be not defined  on
certain values.
\begin{defn}
 A program $p$ is
called computable if it halts on each input value.
\end{defn}
 It is important
that there is no enumeration for computable functions. Indeed, the
set $\{n\,:\,\, \phi_n \,\,\text{is a computable function} \}$ is not computable.

There are different mathematical models for computability theory. One of them is the so-called  Turing machine ~\cite{turing}. These models
 are all equivalent. Indeed, each program in any programming
language can be simulated in other programming languages. More
precisely,

\begin{ch}
A function $f:\n\to\n$ is
effectively computable if and only if  $f$ is partial recursive if and
only if it is Turing computable,
where effectively computable means that for a given function there exists a brief way or an algorithm  to compute it for input numbers.
\end{ch}

 Different models  reinforce our intuition regarding what is computable. Alan Turing's in 1936   introduced
a mathematical model of a computing device  that mechanically works on a tape
which is specially used to  operate as a CPU
inside a computer. More precisely,

 A Turing
machine $T$ consists of a  two-sided infinite tape, subdivided into square cells, and a reading/writing head. To describe a Turing program for  $T$, one needs programming symbols as follows:
\begin{enumerate}
    \item Only one of the tape symbols $S_0, S_1, \ldots, S_n$ can be written on each tape cell. We usually assume that $S_0=0$  called "blank", and $S_1=1$. The set of tape symbols $S_1, S_2, \ldots, S_n$ is called an alphabet set, where only a finite number of them is allowed to be written on the tape while the remaining cells are "blank". In the following the alphabet consists of  "$1$".
    \item  $T$  consists of a finite list of internal states $q_0, q_1, \ldots, q_s
    $.  These states specify the state of the reading head before any given program step.
    \item  The action symbols are  $L$ ( move left one cell ), $R$ (move  right one
    cell), $1$ (print $1$) and $0$ (erase the current cell) which are used by a program to tell the reading head what to do in relation to its current cell.
\end{enumerate}
A program for $T$ is a finite list of instructions, called
quadruples, $q_i S A q_j$. The meaning of this symbol is as follows: "if $T$ is in  state
$q_i$ reading tape symbol $S$, then perform action $A$ and pass into
new internal state $q_j$". To input $n\in\n$ on the tape, we
write $n+1$  $1$'s on the tape and set the reading head  in
starting state $q_0$ reading the leftmost $1$. If there is no
applicable quadruple in $T$ then $T$ halts and the output of the
program is the remaining number of $1$'s on the tape.

A function is
called Turing computable if there exists a Turing program that computes it. Now using Godel numbering, let $\phi_0, \phi_1, \phi_2,
\ldots$ be the enumeration for partially computable functions from $\n$
into $\n$.
\begin{defn}
The partially computable function $f(x,y)=\phi_x(y)$ from
$\n\times\n$ into $\n$  is called universal: it  which simulates any partially computable function $\phi_n$ from $\n$
to $\n$ for a given number $n\in\n$.
\end{defn}

In the real world,  a standard (classical) computer or quantum computer  cannot execute  a
program for infinite long time. For this reason, we introduce  the following useful notion of  decidable
predicate:
\begin{defn}
The predicate $STP^{(n)}(x_1, x_2, \ldots, x_n, y, t)$ is defined as follows:
\begin{eqnarray*}
STP^{(n)}(x_1, x_2, \ldots, x_n, y, t)
 &\Leftrightarrow&
\text{Program number}\,\, y\,\, \text{halts after}\,\, t \,\,\
\text{or fewer steps }
\\&&
\text{on inputs}\,\, x_1, x_2, \ldots, x_n
\\&\Leftrightarrow&
\text{There is a computation of program number}\,\, y\,\, \text{of length}
\\&&
 \leq t+1, \text{beginning with inputs}\,\, x_1, x_2, \ldots, x_n
\end{eqnarray*}
where   $x_1, x_2, \ldots,
x_n$ are input variables of the program.
\end{defn}
% It has been also shown in ~\cite{Da} that the predicate
%$STP^{(n)}(x_1, x_2, \ldots, x_n, y, t)$ is  primitive recursive and hence it is decidable.

 The set $\bigcup_{n\geq 0}\{0,1\}^n$ of all binary strings of finite
length will be denoted by $\{0,1\}^*$ or $\Omega_2$. The map
$str= \{0, 1  \}^*\to \n$ where
$str(a_0 a_1 \ldots a_n)=2^{n+1}-1+ \sum_{k=0}^n a_k 2^{k} $, $a_i=0, 1,$ for each $0\leq i\leq n$,  defines a
one-to-one correspondence between $\{0,1\}^*$ and $\n$.
\begin{defn}
A function
$f:\{0,1\}^* \to \{0,1\}^*$ is called (partially computable)
computable if the function $\textbf{x}\circ f \circ
\textbf{x}^{-1}:\n\to\n$ is (partially computable) computable.
\end{defn}

Let $x$ and $y$ be two elements of $\{0,1\}^*$. we say that $x$ is a
prefix of $y$ if there exists an element $z\in\{0,1\}^*$  such that
$xz=y$ where  $xz$ means a concatenation of $x$ and $z$.
\begin{defn}
A
subset $S\subseteq\{0,1\}^*$ is called prefix-free if no element of
$S$ is a prefix of another elements.
\end{defn}
\begin{defn}
 A partially computable function
is called prefix-free if its domain is a prefix-free subset of
$\{0,1\}^*$.
\end{defn}
 It has been shown \cite{Downey} that there exist
prefix-free universal functions capable  to simulate  all other
prefix-free functions.
\begin{example}
 $$A^*= \{ \underbrace{11\cdots 1}_{n\,\, \text{times}} 0 i_1 i_2 \cdots i_n| \,\, i_1 i_2 \cdots i_n\in\Omega_2 \},$$  is clearly a prefix-free set.  The function $f:\Omega_2 \to \Omega_2$, $f(x)=x$, if $x\in\A^*$, otherwise undefined,  is a prefix-free function with domain $A^*$.
\end{example}
\section{Kolmogorov Complexity and Semi-Computable Functions}
Algorithmic complexity theory was developed by Kolmogorov, Solomonoff
and Chaitin in order to measure  the information content of a binary string. It is based on the fact that regular strings,
such as  a piece of text, have short descriptions. Consider for
example the two strings
$$s:=1111111111111111,$$
$$t:=1001101111000010.$$
One way to  describe the string $s$ which is a repetition of the bit
$1$, is \textbf{print 1 n times}. But, there is no pattern  underlying the  string $t$. Therefore,  the length of a program that describes it is longer  than the number of its bits, and the  length of the
description of the string $s$ is clearly shorter than the length of the description of the string $t$.

We are going to define  the Kolmogorov complexity.

One attributes to
a binary string $\bi^{(n)}=i_1i_2\cdots i_n\in\{0,1\}^n$ of length
$n$ a complexity $C(\bi^{(n)})$ measured by the length of any
shortest program $p^*$ (another binary string of length $\ell(p^*)$)
in the domain of a binary universal partially computable function or
equivalently, the shortest program for a  Universal Turing Machine (UTM)
$\mathcal{U}$, with output $\bi^{(n)}$,
$\mathcal{U}[p^*]=\bi^{(n)}$,
\begin{equation}
\label{algcomp} C(\bi^{(n)})=\min\Big\{\ell(p)\ :\
\mathcal{U}[p]=\bi^{(n)}\Big\}\ .
\end{equation}

The Kolmogorov complexity can be defined based on the prefix-free universal Turing machines. Let a $\mathcal{U}$ be prefix-free universal Turing
Machine (UTM) $\mathcal{U}$, then
\begin{equation}
\label{algcomp} K(\bi^{(n)})=\min\Big\{\ell(p)\ :\
\mathcal{U}[p]=\bi^{(n)}\Big\}\ .
\end{equation}
 The prefix property means that if $\mathcal{U}$ halts on a program
$p$ it does not continue to read on if another program $q$ is
appended to $p$; in other words, no halting program can be used as
prefix to a halting program.
\begin{rem}
The main properties of a prefix-free set is the Kraft inequality, the important inequality in coding theory, with many relevant consequences.
Furthermore, by relations (\ref{kolpre}) and (\ref{re:kol}), the rate of the prefix-free Kolmogorov complexity is equal to the rate
of the Kolmogorov complexity. Therefore, we will consider prefix-free Kolmogorov complexity in this thesis.
\end{rem}
\smallbreak
\textbf{Properties of the classical Kolmogorov complexity}
\begin{enumerate}
      \item   If $\mathcal{U}$ is a universal computer(prefix-free computer) then for every computer(prefix-free computer $Q$) $P$ there exists  constants $c_p>0$ ($k_q>0$) such that for all  string $\bi\in \Omega_2$,
          \begin{equation}
          C_\mathcal{U}(\bi)\leq C_p(\bi)+c_p,
          \end{equation}
          and
          \begin{equation}
          K_\mathcal{U}(\bi)\leq K_q(\bi)+k_q,
          \end{equation}
          where the constants $c_p$ and $k_q$ do not depend on $\bi$.
      \item   The number of all strings $\bi$ with complexity $C(\bi)<c$
         satisfies the following inequality
        \begin{equation}\label{re:compres}
        \#\{\bi| \bi\in\Omega: C_\mathcal{U}(\bi)<c\}<2^c.
        \end{equation}
         Thus,  there are no more than $2^c$ string $\bi$ with complexity $C(\bi)<c$.
      \item   Universal probability of a binary string $\bi$  is
         defined by
          \begin{equation}
          P_\mathcal{U}(\bi)=\sum_{(p:\mathcal{U}(p)=\bi)} 2^{-l(p)},
           \end{equation}
          where $\mathcal{U}$ is a universal prefix-free Turing machine. It is shown that for every
          program $P$ there exists a constant numbers $c_P>0$  such that for all string $\bi\in\Omega$
          \begin{equation}
           P_U(\bi)\leq c_p \cdot P_p (\bi),
          \end{equation}
          where the constant $c_p$ does not depend on the string $\bi$. In
          addition, it is proved that for a constant
           $c>0$,
          \begin{equation}
           K(\bi)\overset{+}{=}- \log{P_U (\bi)},
          \end{equation}
          where $c$ dose not depend on $\bi$. The symbol $ \overset{+}{=}$ means that there exist constants $c_1>0$ and $c_2>0 $ such that
          $$ K(\bi)\leq - \log{P_U (\bi)} +c_1,\quad  - \log{P_U (\bi)} \leq K(\bi)+c_2,$$
          for each any binary string $\bi$.
         \item
           For any string $\ii$,
         \begin{equation} \label{kolpre}
         C(\ii)\leq K(\ii),
         \end{equation}
         and if $p$ is a program such that $C(\ii)=\ell(p)$, then it
         follows that \cite{Vitany}:
         \begin{equation}\label{re:kol}
       K(\ii)\leq C(\ii)+2 \log{\ell(p)} + c_p\leq C(\ii)+2 \log{n} +
       c_p.
       \end{equation}
       \item   Unfortunately, algorithmic complexity or Kolmogorov complexity is not computable; therefore, there is no effective way to compute
        $P_U$. But, it  can be approximated within arbitrary precision.
\end{enumerate}
Let $h:\n\times \n\to \rr$ be a function. Then, for each $n$,
$h_n:\n\to\rr$ is defined as follows $h_n(x)=h(n,x)$.
\begin{defn}
 A function $g:\n\to\rr$ is called lower semi-computable if there exists a
computable function $f:\n\times\n\to\q$ such that the sequence $f_n$
is an increasing sequence and $\lim_{n\rightarrow\infty}f_n=g$.
\end{defn}
\begin{defn}
A function $\mu:\n\to\rr$ is called a (semi-computable) semi-measure
if it is a positive semi-computable function such that
$\Sigma_x\mu(x)\leq 1$.
\end{defn}
\begin{defn}
A function $h:\n\to\rr$ is called upper semi-computable if $-h$ is
lower semi-computable and it will be called computable if it is
lower and upper semi-computable.
\end{defn}
\begin{defn}
A semi-computable semi-measure $\mu$ is called universal if for any
semi-computable semi-measure $\nu$ there exists a constant $c_\nu>0$
such that for each $x\in\n$, $c_\nu\nu(x)\leq\mu(x)$.
\end{defn}
The existence
of a universal semi-measure is proved by Levin
\cite{Zvonkin70thecomplexity}:
\begin{thm}\label{thm:levin}
There is a semi-computable semi-measure $\mu$ with the property that
for any other semi-computable semi-measure $\nu$ there is a constant
$c_\nu>0$ such that for all $x\in\n$ we have $c_\nu \nu(x)\leq
\mu(x)$.
\end{thm}
Levin \cite{Zvonkin70thecomplexity} is also proved a relation
between prefix Kolmogorov complexity and universal semi-measure as
follows
\begin{thm}(Levin's Coding Theorem)\label{thm:levin2}
We have $K(x)\overset{+}{=} -\log \mu(x)$, for all $x\in\Omega_2$.
\end{thm}
\begin{thm}
Any semi-computable function $\psi:\n\to\q$ can be represented by a computable function from $\n$ into $\n$.
\end{thm}
\begin{proof}
Let $f:\n\times\n\rightarrow\q$ be a computable function and  increasing with respect to the first argument $n$, and such that
                         $$\psi(x) = \lim_{n\rightarrow\infty}f( n, x )\,\,\,\,\,\,\text{for all}\,\,\,\,\,\, x\in \n.$$
Let $\phi:\n\times\n\rightarrow\q$ be defined by $\phi(n,x)=f(n ,x)$
if $x\leq n$, otherwise $0$. Then, $\phi_n$ is an increasing
sequence of computable functions and $\lim_{n\to\infty}
\phi_n=\psi$. Let
$\Phi:\n\to\n$ be defined as follows: $\Phi(n)=2^{\alpha(\phi(n,0))}
\times 3^{\alpha(\phi(n, 1))} \times\ldots\times
p_n^{\alpha(\phi(n,n))},$ where $p_n$ is the $n$-th prime number and $\alpha:\q\rightarrow\n$ is an
injection~\footnote{ An instance of such an injection is the map
$\iota'\circ\iota$ defined in section \ref{sec:unversal semi-measure}.}. Then, $\Phi$ is a computable function.

 Now, $\psi$ can be defined by $\Phi$ as follows:
  $$\psi(n)=\alpha^{-1}(\Phi(n)_n),\,\,\, \text{where}\,\,\, \Phi(n)_n=\alpha(\phi(n, n)).
  $$
 \end{proof}

 In this way, we will represent all necessary
semi-computable  quantities that appear in the following like
semi-computable semi-measures, semi-computable Hilbert space vectors
and semi-density matrices by computable functions on $\n$.

\section{Relation between Algorithmic Complexity and Thermodynamics}
One nice application of algorithmic complexity concerns the relations
between computation and thermodynamics. Since computation is a
physical process, its  thermodynamical cost is certainly important.
 A usual question in computation  theory  is about which processes can be performed reversibly and which ones are necessarily irreversible.
Rolf Landauer and Charles Bennett ~\cite{land, bennett} have been shown that  any  thermodynamically irreversible computer operation should be logically irreversible.
For instance, data erasure is an example of irreversible process as it eliminates irretrievable information.

 Let  us consider a cubic box of volume $V$ containing a gas molecule with
a freely moving piston which can be used to locate the molecule on
the left side of the box; this molecule position can be identified as a bit 1.
\begin{center}
\begin{tikzpicture}\label{fig:box1}
     \draw (0,0) -- (0,2) -- (3,2)--(3,0)-- (0,0);
     \draw (1.5,0)--(1.5,2);
     \draw (1.5, 1)--(5, 1);
     \draw (0.40,1) node[rectangle,fill=white]{};
     \draw [fill=gray](.75,1)circle (1ex);
     %\caption{Box}
     \end{tikzpicture}
           %\draw[->] (description) to [out = 180, in = 0, looseness = 2] (text);
        \end{center}
    %\begin{figure}
%\includegraphics[width=2.5cm]{image.jpg}
%\end{figure}
The flip operation which transforms bit 1 (molecule confined to the left half of the box) into bit 0 (molecule confined in the right half of the box)   can be performed reversibly by slowly rotating the box around its vertical axis and thus exchanging the two halves of the box.
\begin{center}
\begin{tikzpicture}
     \draw (0,0) -- (0,2) -- (3,2)--(3,0)-- (0,0);
     \draw (-2, 1)--(1.5, 1);
     \draw (1.5,0)--(1.5,2);
 \draw [fill=gray](2.25,1)
circle (1ex);
       \end{tikzpicture}
    \end{center}
  The compression of the piston that confines the molecule to one half of the box, can be performed isothermically, without changing the temperature of the box. After that, one allows the piston slowly return to the initial state. Correspondingly, there occurs a loss of information due to the doubling of the volume,  that the molecule can occupy: this amounts to erasing one bit of information.

  Now, we can construct a computer working with only two bits per tape cell which can be operated in analogy with the thermodynamical box depicted above.
    Let us consider the
simple program to add the two bits and save its result in
the memory.
\begin{center}
\begin{tikzpicture}
         \draw (0,0) node[rectangle,draw]{0};
          \draw (0.5,0) node[rectangle,draw]{0};
           \draw (1.5,0) node[rectangle]{$\longrightarrow$};
           \draw (2.5,0) node[rectangle,draw]{0};
          \draw (3,0) node[rectangle,draw]{0};
          %line2
 \draw (0,-1) node[rectangle,draw]{0};
          \draw (0.5,-1) node[rectangle,draw]{1};
           \draw (1.5,-1) node[rectangle]{$\longrightarrow$};
           \draw (2.5,-1) node[rectangle,draw]{0};
          \draw (3,-1) node[rectangle,draw]{1};
          %line3
           \draw (0,-2) node[rectangle,draw]{1};
          \draw (0.5,-2) node[rectangle,draw]{0};
           \draw (1.5,-2) node[rectangle]{$\longrightarrow$};
           \draw (2.5,-2) node[rectangle,draw]{0};
          \draw (3,-2) node[rectangle,draw]{1};
          %line 4
           \draw (0,-3) node[rectangle,draw]{1};
          \draw (0.5,-3) node[rectangle,draw]{1};
           \draw (1.5,-3) node[rectangle]{$\longrightarrow$};
           \draw (2.5,-3) node[rectangle,draw]{1};
          \draw (3,-3) node[rectangle,draw]{0};
 %\draw[->] (description) to [out = 180, in = 0, looseness = 2] (text);
    \end{tikzpicture}
    \end{center}
The binary addition $\oplus$ in the above is defined as follows
$$0\oplus0=0,\quad 0\oplus1=1,\quad 1\oplus0=1,\quad 1\oplus1=0. $$
Since the results of two operations
\begin{tikzpicture}
         \draw (0,0) node[rectangle,draw]{0};
          \draw (0.5,0) node[rectangle,draw]{1};
\end{tikzpicture}
 and
 \begin{tikzpicture}
         \draw (0,0) node[rectangle,draw]{1};
          \draw (0.5,0) node[rectangle,draw]{0};
\end{tikzpicture}
 are the same
 \begin{tikzpicture}
         \draw (0,0) node[rectangle,draw]{0};
          \draw (0.5,0) node[rectangle,draw]{1};
\end{tikzpicture}
  the operation is not logically reversible and hence it is not thermodynamically
reversible. But, we can use more tape cells to solve this
difficulty, which writes sum of the two bits in  different part of
the memory, using the additional bits to save the inputs and the
outputs. Therefore, we can construct a logically reversible computer
operation. For instance in the sum of two bits, the first two memory
cell store the inputs and the second two memory cell which are
initially zeros save the outputs as  bits.
\begin{center}
\begin{tikzpicture}
         \draw (0,0) node[rectangle,draw]{0};
          \draw (0.5,0) node[rectangle,draw]{0};
          \draw (1,0) node[rectangle,draw]{0};
          \draw (1.5,0) node[rectangle,draw]{0};
           \draw (2.5,0) node[rectangle]{$\longrightarrow$};
           \draw (3.5,0) node[rectangle,draw]{0};
          \draw (4,0) node[rectangle,draw]{0};
          \draw (4.5,0) node[rectangle,draw]{0};
          \draw (5,0) node[rectangle,draw]{0};
          %line2
 \draw (0,-1) node[rectangle,draw]{0};
          \draw (0.5,-1) node[rectangle,draw]{1};
          \draw (1,-1) node[rectangle,draw]{0};
          \draw (1.5,-1) node[rectangle,draw]{0};
           \draw (2.5,-1) node[rectangle]{$\longrightarrow$};
           \draw (3.5,-1) node[rectangle,draw]{0};
          \draw (4,-1) node[rectangle,draw]{1};
          \draw (4.5,-1) node[rectangle,draw]{0};
          \draw (5,-1) node[rectangle,draw]{1};
          %line3
           \draw (0,-2) node[rectangle,draw]{1};
          \draw (0.5,-2) node[rectangle,draw]{0};
          \draw (1,-2) node[rectangle,draw]{0};
          \draw (1.5,-2) node[rectangle,draw]{0};
           \draw (2.5,-2) node[rectangle]{$\longrightarrow$};
           \draw (3.5,-2) node[rectangle,draw]{1};
          \draw (4,-2) node[rectangle,draw]{0};
          \draw (4.5,-2) node[rectangle,draw]{0};
          \draw (5,-2) node[rectangle,draw]{1};
          %line 4
            \draw (0,-3) node[rectangle,draw]{1};
          \draw (0.5,-3) node[rectangle,draw]{1};
          \draw (1,-3) node[rectangle,draw]{0};
          \draw (1.5,-3) node[rectangle,draw]{0};
           \draw (2.5,-3) node[rectangle]{$\longrightarrow$};
           \draw (3.5,-3) node[rectangle,draw]{1};
          \draw (4,-3) node[rectangle,draw]{1};
          \draw (4.5,-3) node[rectangle,draw]{1};
          \draw (5,-3) node[rectangle,draw]{0};
           %\draw[->] (description) to [out = 180, in = 0, looseness = 2] (text);
    \end{tikzpicture}
   \end{center}

  Now, we identify the free energy, namely energy that can be transformed into expendible work, and free memory.

 The problem is that, if we want to operate reversibly by storing extra information, the free memory will soon become saturated and demand data erasure. This process consumes free energy by generating heat: one would then try to compress as much as possible the garbage data before erasing them.

If $T$ is the temperature at which the computation is performed, the heat generated, equivalently the free energy consumed, by erasing one bit of information is given by
  $$\Delta S=\frac{\Delta Q}{T}  ,$$
where  $\Delta S=\kappa \log 2$,  $\kappa=1.38 \times 10^{-23} J/K$. Therefore, when $1$ or $0$ is written in a bit the amount of free energy.

 Suppose the garbage data occupying the free memory correspond to the string
$\ii$. The best way to compress it
 is to use  a  program $p^*$ with shortest length such
that $U(p^*)=\ii$ where $U$ is a universal Turing machine.

Now, the
minimal free energy consumption amounts to
$\Delta_{opt}F=-\kappa T C(\ii)$ which is a lower bound to $\Delta F= -n \kappa T \log 2 $, where $n$ is the length of $\ii$.

%\begin{figure}
%\includegraphics[width=2.5cm]{figure.pdf}
%\end{figure}

%In ~\cite{zurek} the total energy
%of computer is defined as follows
%$$S_{com}=S_{th}+ \kappa C(M) \log 2,$$
%where $C(M)$ is the kolmogorov complexity of binary string
%representing the memory state of the computer and $S_{th}$ is
%ordinary thermodynamic entropy.

%%%%%%%%%%%%%%\newpage \thispagestyle{empty} \mbox{}
%\def\baselinestretch{1}
\chapter{Classical and Quantum Entropy }
\goodbreak
In this chapter, we    introduce  the Kolmogorov-Sinai
entropy for  classical dynamical systems with the aid of  symbolic
models. Symbolic models will then be associated to the algebraic description of classical spin chains. This will lead us to the introduction of quantum spin chains and of two quantum dynamical entropies, that of Connes, Narnhofer and Thirring (CNT) and that of Alicki and Fannes (AF).

\smallskip

\section{Classical dynamical systems}\label{classical dyanamical}
Classical dynamical systems can be defined as abstract mathematical objects in terms of triples
$(\chi, T, \nu)$ where
\begin{enumerate}
    \item $\chi$ is a phase-space; namely,  $\chi$ is a measure space  endowed with a $\sigma$-algebra
     $\Sigma$ of measurable sets.
    \item $T$ is a measurable map such that for any $A\in\Sigma\Rightarrow T^{-1}(A)\in
    \Sigma$.
    \item $\nu$ is a  $T$-Invariant  probability measure on $\chi$; namely,
    $\nu(\chi)=1$ and $\nu\circ T^{-1}=\nu$.
\end{enumerate}
 A reversible dynamical system is a dynamical system such that for  the discrete time evolution $T$, $T^{-1}$  is also measurable
such that, $\nu \circ T=\nu$ and if $A\in\Sigma$
then $T(A)\in \Sigma$.
\begin{defn}
 Let $(\chi, T, \nu)$ be a classical dynamical system. A finite measurable partition $\mathcal{P}$ of $\chi$ is a
finite set of  disjoint measurable subsets $P_1, P_1, \ldots,
P_n$ of $\chi$ such that $\chi=\cup_{i=1}^n P_i$. The elements $P_i$
of $\mathcal{P}$
 are usually called atoms.
 \end{defn}
  Composition
of two partitions $\mathcal{P}$ and $\mathcal{Q}$ are also a
partition $P\vee Q=\{P_k \cap Q_l| P_k\in \mathcal{P},
Q_l\in\mathcal{Q} \}$.

One way to study continuous phase-spaces with discrete time dynamics  is by discretizing the continuous  phase-space using a finite partition, a process called coarse-graining. Firstly, we introduce the meaning of trajectory in a dynamical system $(\chi, T, \nu)$ with discrete time evolution $T$.

In general, for a given element $x\in\chi$ , the trajectory of $x$ is defined as set $\{T^kx \}$ where $k\in\zz$. Indeed, it shows the position of an  element on phase-space after after $k$ time-steps. Then, one defines a coarse-grained trajectory issuing from $x$ by using finite partitions.
\begin{defn}
Let  $(\chi, T, \nu)$ be a dynamical system  with
the finite measurable partition $\mathcal{P}$ of $\chi$ with $p$ elements. The
coarse-grained trajectory through $x\in\chi$ dependent on partition $\mathcal{P}$ is defined by the string
$\Omega_p\ni\bi(x):=i_1 i_2 i_3 \ldots $ where $T^k(x)\in P_{i_k}$. By varying $x\in\chi$,
the set
of such strings will be denoted by $\widetilde{\Omega}_p^\z$ where
$\widetilde{\Omega}_p^\z\subseteq \Omega_p^\z$.
Therefore, for a phase point $x\in\chi$, the trajectory $\{T^k x
\}_{k\in\z}$ can be encoded by a string dependent on a specified finite measurable partition of phase-space.
\end{defn}

%$\mathcal{P}=\{P_i\}_{i=1}^p$ be a finite partition of dynamical system $(\chi, T, \nu)$, which all of the atoms $\mathcal{P}$ are
For  a given  dynamical system and a finite measurable partition
$\mathcal{P}$, the symbolic dynamical system $(\widetilde{\Omega}_p^\z,
T_\sigma, \nu_\mathcal{P})$, is defined as
follows
\begin{enumerate}
  \item The $\sigma$-algebra of measurable sets is generated by
  cylinders consisting a cylinder consists of all strings whose elements have fixed
  values in chosen intervals:
   $$C_{i_l}^{\{l \}}=\{\bi\in\Omega_p:\,\, \bi_l=i_l\},\quad C^{[j, k]}_{\underbrace{{i_j i_{j+1}\ldots i_k}}_{\bi^{k-j+1}}}=\{\bi\in \Omega_p: \bi_{j+l}=i_{j+l}, \, l=0,1, \ldots, k-j\}.$$
  \item $T_\sigma$ is a left shift dynamics along strings on $\Omega_p$. In other
  words, for a string $\bi\in\Omega_p$,
  $(T_\sigma(\bi))_j=\bi_{j+1}$.
  \item The probability measure $\nu_p$ is defined by $\nu_{\mathcal{P}}(\ii)=\nu(\p)$, which
\begin{equation}
  \p:=P_{i_0}\bigcap T^{-1}(P_{i_1})\bigcap \dots
  \bigcap T^{-n+1}(P_{i_{n-1}}).
  \end{equation}
\end{enumerate}
\begin{rem}
It is straightforward to see that, in the symbolic dynamical system
$(\widetilde{\Omega}_p, T_\sigma, \nu_{\mathcal{P}})$, the invariance condition  $\nu\circ T^{-1}=\nu$   is equivalent to
\begin{equation}\label{compatibility}
\sum_{i=1}^p\nu_\mathcal{P}(i i_2 \ldots i_n)=\nu_\mathcal{P}(i_2
\ldots i_n),
\end{equation}
Notice that the invariance condition  is different from
the compatibility condition
\begin{equation}
\sum_{i_n=1}^p\nu_\mathcal{P}(i_1 i_2 \ldots i_n)=\nu_\mathcal{P}(i_1
i_2\ldots i_{n-1}),
\end{equation}
which must hold for all probability measures $\nu$.
\end{rem}

The Kolmogorov-Sinai (KS) entropy of classical dynamical systems is, roughly speaking, the highest Shannon entropy rate for all its symbolic models.
  Indeed, let
$$\mathcal{P}^{(n)}:=\{\p| \bi^{(n)}=i_0 i_1 \dots i_n,\,\, i_j\in
I_p\}$$ be a refinement of the partition $\mathcal{P}$. The entropy of $\mathcal{P}^{(n)}$ is measured by the Shannon entropy of the probability
distribution $\{\nu_\mathcal{P}(\p)\}_{\Omega_p^{(n)}}$,
\begin{equation}
H_\nu(\mathcal{P}^{(n)}):=-\sum_{\Omega_p^{(n)}}\nu_{\mathcal{P}}(\ii)\log\nu_{\mathcal{P}}(\ii).
\end{equation}
Now,  KS entropy associated with $\nu, T, \mathcal{P}$ is defined as the
shannon entropy rate
\begin{equation}
h_\nu^{KS}(T, \mathcal{P}):=
\lim_{n\to\infty}\frac{1}{n}H_\nu(\mathcal{P}^{(n)})=\inf_n\frac{1}{n}H_\nu(\mathcal{P}^{(n)}).
\end{equation}
  Now, by taking $\sup$
over all partitions, one can get a definition independent of
partitions.

\begin{defn}
The KS entropy of the classical dynamical system $(\chi, T, \nu)$ is
defined by
\begin{equation}
h_\nu^{KS}(T):=\sup_\mathcal{P}h_\nu^{KS}(T, \mathcal{P}),
\end{equation}
where the sup is taken over all finite measurable partitions
$\mathcal{P}$.
\end{defn}
\begin{rem}
It is not easy to compute  $\sup$ in the KS entropy definition. But, by the
Kolmogorov-Sinai Theorem \cite{billingsley}, if there exists a
generating partition
$\mathcal{P}$, then $$h_\nu^{KS}(T_\sigma)=h_\nu^{KS}(T,\mathcal{P}),$$
where  a generating partition is a finite partition such that  the
set of refined partitions $\mathcal{P}^{(n)}$ for all $n\in\n$,
generates the $\sigma$-algebra $\Sigma$ of phase space $\chi$.
\end{rem}
The following simple example shows us the computation of the KS-entropy for the Bernoulli shift dynamics.
\begin{example}(\textbf{Bernoulli shifts})
Let us consider a shift dynamical system $(\Omega_2, T_\sigma, \nu)$
where the measure $\nu$  is locally defined as follows
$$\nu(C^{[j, k]}_{i_j
i_{j+1}\ldots i_k})=p^{(k-j+1)}(i_j i_{j+1}\ldots i_k),$$ where
$$p^{(n)}(i_1\ldots i_n)=\prod_{j=1}^{n} p(i_j), \quad p(i)\geq 0, \,\, \sum_{i=1}^d p(i)=1.   $$
On the other hand,  $\mathcal{P}:=\{ C_j^{\{0\}}\}_{j=1}^p $ is a
generating partition for the $\sigma$-algebra of cylinders.
Therefore,
$$h_\nu^{KS}(T_\sigma)= h_\nu^{KS}(T_\sigma, \mathcal{P})=\lim_{n\to\infty} \frac{1}{n} H_\nu(\mathcal{P}^{(n)})=-\sum_{i=1}^p p(i) \log p(i)=H_\nu(\mathcal{P}).$$
\end{example}

Ergodic theory developed  in \cite{khinchin} explains when and why mean values of observables coincide with their time-averages  why trajectories in ergodic
systems  fill the phase-space  densely.
 \begin{defn}
A dynamical system $(\chi, T, \nu)$ is called ergodic if for every
$\psi, \phi\in L^2_\nu(\chi)$,
\begin{equation}
\lim_{t\to\infty} \frac{1}{t}\sum_{s=0}^{t-1} \nu(\psi\phi \circ
T^s)=\nu(\psi)\nu(\phi).
\end{equation}
\end{defn}
The quantity $C(x, \mathcal{P}):=\limsup_n\frac{1}{n} (\min_{\ii}
C(\ii, \mathcal{P}))$, where $C(\ii, \mathcal{P}):=C(\ii)$,  is called the complexity of a point $x\in\chi$
with respect to a finite measurable partition $\mathcal{P}$. The quantity
$C(x):=\sup_\mathcal{P} C(\ii, \mathcal{P})$ is called the
complexity of the trajectory of $x\in\chi$.

 The two following theorems  proved by  Brudno \cite{Brudno},
shows a relation between compression of data
and the Kolmogorov complexity. Actually, it sets a relation between different subjects in mathematics, computer science and physics.
\begin{thm}
In a binary ergodic source $(\Omega_2, T_\sigma, \pi)$, with entropy
rate $h_\sigma^{KS}(T_\sigma)$, we have
\begin{equation}\label{Brudno}
\lim_{n\rightarrow
\infty}\frac{1}{n}C(\textbf{i}^{(n)})=h_\pi(T_\sigma),
\end{equation}
for almost all $\textbf{i}^{(n)}\in\Omega_2$ with respect to $\pi$.
\end{thm}
\begin{thm}
Let $(\chi, T, \nu)$ be an ergodic dynamical system and
$\mathcal{P}$ be a finite measurable partition of $\chi$; then
\begin{equation}
C(x, \mathcal{P})=h_\sigma^{KS}(T_\sigma, \mathcal{P})\quad \nu-a.e.
\end{equation}
If $\mathcal{P}$ is a generating partition then,
\begin{equation}
C(x, \mathcal{P})=h_\sigma^{KS}(T_\sigma)\quad \nu-a.e.
\end{equation}

\end{thm}
%A partition $\mathcal{P}$ is called generation for a dynamical
%system if $\mathcal{P}^{(n)}$ generate $\Sigma$-algebra.
%\begin{thm}\textbf{(Kolmogorov-Sinai Theorem)}
%If the partition $\mathcal{P}$ is generating for dynamical system
%$((\chi, T, \nu))$, then $h^{KS}_\nu(T)=h^{KS}_\mu(T, \mathcal{P})$.
%\end{thm}
%\begin{thm}\textbf{(Asymptotic Equipartition Property)(EAP)}
%For any $\epsilon>0$ and $\delta>0$, there exists $N_{\epsilon,
%\delta}$ such that, for all $n>N$, the high probability subsets
%$\A_{\epsilon, \delta}^{(n)}\subset\Omega_2^{(n)}$ are such that,
%for all $\ii\in\A_{\epsilon, \delta}^{(n)}$,
%\begin{equation}
%e^{-n(H(A)+\epsilon)}\leq p(\ii)\leq e^{-n(H(A)-\epsilon)},
%\end{equation}
%while, their cardinalities $\#(\A_{\epsilon, \delta}^{(n)})$ satisfy
%\begin{equation}
%(1-\delta)e^{n(H(A)-\epsilon)}< \#(\A_{\epsilon, \delta}^{(n)}) \leq
%e^{n(H(A)+\epsilon)}
%\end{equation}
%\end{thm}
%\begin{thm}
%Let $(\chi, T, \nu)$ be a reversible, ergodic dynamical system.
%Then, for any finite, measurable partition $\mathcal{P}=\{
%P_i\}_{i=1}^p$,
%$$\lim_{n\to\infty} h_n(x)=h_\nu^{KS}(T)\quad \nu-a.e,$$
%which $h_n(x)=-\frac{1}{n}\log\nu(\mathcal{P}_0^{n-1}(x))$ and
%$\mathcal{P}_0^{n-1}(x)=\bigcap_{j=0}^{n-1}T^{-j}(P_{i_j})\Longleftrightarrow
%T^i x\in P_{i_j}, \quad \forall j=0, 1, \ldots, n-1.$
%\end{thm}
%In the case of a bilateral Bernoulli shift we have
%$$\lim_{n\to\infty} h_n(\bi)= H_\nu(\mathcal{C})=h_\nu^{KS}(T_\sigma)\quad \nu-a.e,$$
%where the generating partition $\mathcal{C} $ is for partition
%$\mathcal{P}$.
%\begin{example}
%Bernoulli shift
%\end{example}

\section{Classical Spin Chains and Algebraic Formulation}
In many cases, it proves useful  to investigate
classical dynamical systems using algebraic tools. Namely,
instead of  working with phase-space trajectories, one considers
observables (suitable functions over the phase-space) and their
time-evolution. In  other words, to a given dynamical system $(\chi,
T, \nu)$, where $\chi$ is a compact metric space, one can associate
a $C^*$-algebraic triplet $(C(\chi), \Theta_T, \omega_\nu)$
\footnote{$C(\chi)$ is the Banach algebra $*$-algebra (with
identity) of continuous complex value functions on $\chi$ endowed
with the uniform topology given by the $\sup$ norm.} and a von
Neumann triplet $(\mathds{L}_\nu^\infty(\chi), \Theta_T,
\omega_\nu)$ where state $\omega_\nu$ and automorphism $\Theta_T$
are defined as follows:
\begin{equation}
\omega_\nu(f)=\int_\chi d\nu(x) f(x),
\end{equation}
\begin{equation}
\Theta_T(f)=f\circ T ,
\end{equation}
 for all $ f\in C(\chi)$ or  $\mathds{L}_\nu^\infty(\chi)$.

\begin{example}(Koopmann-von Numann formalism\footnote{The previous one  is a technique which allows one to reformulate classical dynamical systems in terms of Hilbert spaces and unitary time-evolutions, as one does with quantum mechanical systems where one encounters the following basic concepts. })
Let $(\chi, T, \nu)$ be a dynamical system. The Koopmann-von Neumann
unitary operator $U_T$ is defined as follows
$$(U_T\psi)(x)=\psi(Tx),$$
for any $\psi\in\mathds{L}_\nu^2(\chi)$ and $x\in\chi$. Let define $<f|g>=\int_\chi \overline{f(x)} g(x) dx$ be the scalar product of any $f, g\in\mathds{L}_\nu^2(\chi)$. The
automorphism $\Theta_T$ is implemented by $U_T$ as follows
\begin{eqnarray*}
<x|U_T f U_T^\dag \psi>&=&f(Tx)<Tx|U_T^\dag
\psi>\\
&=& f(Tx)<T^{-1}\circ T x|\psi>)\\
&=&<x|\Theta_T(f)\psi>,
\end{eqnarray*}
 for any $f\in
C(\chi)$. Of course, the state $\omega_\nu $ is defined like  the above
definition.
\end{example}
Now, we introduce some definitions here as follows
\begin{defn}
A positive operator $\rho$ of  Hilbert space $\mathbb{H}$ is called density matrix if ${\rm Tr}(\rho)=1$.
\end{defn}
\begin{defn}
For a given density matrix $\rho$ with spectral decomposition $\sum_i \lambda_i |\lambda_i><\lambda_i|$, the von Neumann entropy is defined as follows
$$S(\rho)=-{\rm Tr}(\rho \log\rho)=-\sum_i \lambda_i \log\lambda_i.$$
In addition, relative entropy for given two density matrices $\rho$ and $\sigma$ is given by
$$S(\rho, \sigma)={\rm Tr}(\rho(\log\rho-\log\sigma)).$$
\end{defn}

It is useful to look at symbolic models of classical dynamical systems as classical spin chains.

A classical spin chin is the mathematical way of modeling a classical ferromagnet as a one-dimensional lattice $\mathbb{Z}$ whose sites support identical classical spins capable of assuming $p$ possible states. In this case, to each site
corresponds an algebra of $p\times p$ diagonal matrices over
$\com$ which is denoted by $D_p(\com)$.

The diagonal matrices $P_j$  whose
elements are all zero but for the $jj$-th entry which is equal to $1$, constitute a set of generating   projections $P_j$, $1\leq j\leq p$, for
the algebra $D_p(\com)$. Thus,  an element $D$ of  $D_p(\com)$ is of the
form $\sum_{j=1}^p d_j P_j$, where $d_j$'s are complex numbers.
 The spin algebra of $n$ particles
located at the lattice sites $-n\leq j\leq n$   will be denoted by
$\mathcal{D}_{[-n, n]}:=\otimes_{j=-n}^{n}(D_p(\com))_j$ where
$(D_p(\com))_j=D_p(\com)$, for each $-n\leq j\leq n$. Indeed, each
element of that algebra is a $p^n\times p^n$ matrix of the
form
$$D^{(2n+1)}_{[-n,n]}:=\sum_{\ii\in\Omega_p^{2n+1}} d(\ii)P_{\ii}^{[-n, n]},\quad\quad \ii=i_{-n}\ldots i_n,$$
where $d(\ii)$'s are complex numbers and $P_{\ii}^{[-n,
n]}:=P_{i_{-n}}\otimes P_{i_{-n+1}} \otimes \ldots \otimes P_{i_n}$
are projectors.

  Let us consider  the symbolic dynamical system $(\Omega_p^{\zz},
T_\sigma, \nu)$, that is a shift dynamical system over  two-sided
infinite sequences of symbols from an alphabet with $p$ elements.
The $C^*$-algebraic    triplet $(\mathcal{D}_{\zz}, \Theta_\sigma,
\omega_\nu)$ associated  with  the symbolic dynamical system
$(\Omega_p^{\zz}, T_\sigma, \nu)$ as outlined before is indeed a classical spin chain.
\begin{itemize}
    \item Let us define the commutative algebra $\mathcal{D}_{\zz}:=\overline{\bigcup_{n\in\n}\mathcal{D}_{[-n,
    n]}}^{\text{uniform}}$,  inductively extended from local
    algebras by a method which is known as $C^*$-inductive limit \cite{takesak}.
    \item  $\Theta_\sigma$ is an  algebraic automorphisms
    \begin{equation}
    \Theta_\sigma(\mathds{1}_{-n-1]}\otimes A
    \otimes\mathds{1}_{[n+1} )=\mathds{1}_{-n]}\otimes A
    \otimes\mathds{1}_{[n+2},
    \end{equation}
    for each $A\in \mathcal{D}_{[-n, n]}$. Therefore,
    \begin{equation}
    \Theta_\sigma(\mathcal{D}_{[-n, n]})=\mathcal{D}_{[-n+1, n+1]}.
    \end{equation}
    Indeed, $\mathcal{D}_{[-n, n]}$ is embedded in
    $\mathcal{D}_{\zz}$ by the map $A \mapsto \mathds{1}_{-n-1]}\otimes A
    \otimes\mathds{1}_{[n+1}$, for $A\in \mathcal{D}_{[-n, n]}$, and from now on,
     we identify $\mathcal{D}_{[-n, n]}$ with
     $\mathds{1}_{-n-1]}\otimes \mathcal{D}_{[-n, n]}
    \otimes\mathds{1}_{[n+1}$.
    \item    Let us consider the local density matrix
    \begin{equation}
    \rho_\nu^{(n)}(\ii):=\sum_{\ion} \nu(\ii)P_{\ii}^{[0, n-1]},
    \end{equation}
    on $\mathcal{D}_{[-n, n]}$. Then, the global density matrix  $\rho_\nu$ is defined as $\lim_{n\to\infty} \rho_\nu^{(n)}$. Furthermore,
     a global state  is
    defined  by
    \begin{equation}
    \omega_\nu(A):= {\rm Tr}_{\mathcal{D}_{[-n, n]}}(A\rho^{(n)}_\nu) \quad \forall A\in\mathcal{D}_{[-n, n]},
    \end{equation}
\end{itemize}
With the notations of   (\ref{sec:quantum spin}), the
KS-entropy for classical spin dynamics computes by the following relation:
\begin{equation}
h^{KS}_\nu(T)=s(\omega):=\lim_{n\to\infty}\frac{1}{n}
S(\omega\upharpoonright \mathcal{D}_{[-n,n]}),
\end{equation}
 where $S(\omega\upharpoonright \mathcal{D}_{[-n,n]})=S(\rho_\nu^{(n)})$ and $s(\omega)$ is called the von Neumann entropy rate.

\section{Quantum Dynamical Systems} Quantum dynamical systems, are in general introduced as
non-commutative algebraic triplets.
\begin{defn}\label{de:quantum dynamical system}
A quantum dynamical system is a triplet $(\A, \Theta, \omega)$ where
$\A$ is a $C^*$-algebra with identity $\mathds{1}$, and
\begin{itemize}
    \item the dynamics $\Theta$ corresponds to a group of automorphisms $\Theta_t: \A\to\A$, $t\in G$, which $G=\rr$
    or $G=\zz $, and, for
    any $t, s\in G$,  $\Theta_t \circ \Theta_s=\Theta_s\circ \Theta_t=\Theta_{t+s}$.
    \item The state $\omega:\A\to\cc$ is a normalized, positive,
    $\Theta$-invariant expectation, namely
    $\omega\circ\Theta_t=\omega$ for all $t\in G$.
\end{itemize}
\end{defn}
Classical spin chains are particular cases of quantum dynamical
systems, where their associated  $C^*$-algebras are commutative.

\subsection{Quantum Spin Chains}\label{sec:quantum spin}
A quantum spin chain is the $C^*$-algebra that arises from the norm
completion of local quantum spin algebras of the tensor product form
\begin{equation}
\label{algtensor} M_{[-n,n]}=\underbrace{M_d(\com)\otimes
M_d(\com)\otimes\cdots M_d(\com)}_{2n+1\quad times} =M^{\otimes
2n+1}_d=M_{d^{2n+1}}(\com)\ .
\end{equation}
The interpretation is straightforward: one is dealing with a
one-dimensional lattice each site of which supports a $d$-level
quantum system (or $d$-dimensional spin). In the norm-topology (the
norm is the one which coincides with the standard matrix-norm on
each local algebra) the limit $n\to+\infty$ of the nested sequence
$M_{[-n,n]}$ gives rise to the norm-complete infinite dimensional
algebra
\begin{equation}
\label{infalg} \mathcal{M}:=\lim_p M_{[-p,p-1]}\ ,
\end{equation}
that describes an infinite quantum spin lattice, that is a quantum
spin chain. In the following we shall consider $d=2$, namely a chain
of $2$-level quantum spins, or spin $1/2$ particles, or in the
modern jargon, qubits.

Any local spin operator, say $A\in M_{[-n,n]}$, is naturally
embedded into $\mathcal{M}$ as
\begin{equation}
\label{embedding} M_{[-n,n]}\ni A\mapsto \one_{-n-1]}\otimes
A\otimes \one_{[n+1}\in\mathcal{M}\ ,
\end{equation}
where $\one_{-n-1]}$ stands for the infinite tensor products of
$2\times 2$ identity matrices up to site $-n-1$, while $\one_{[n+1}$
stands for the infinite tensor product of infinitely many identity
matrices from site $n+1$ onwards. In this way, the local algebras
are sub-algebras of the infinite one sharing  a same identity
operator.

The simplest dynamics on such quantum spin chains is given by the
right shift
\begin{equation}
\label{shift} \Theta[M_{[-n,n]}]=M_{[-n+1,n+1]}\ ,\quad
\Theta[\one_{-n-1]}\otimes A\otimes \one_{[n+1}]=\one_{-n]}\otimes
A\otimes \one_{[n+2}\ .
\end{equation}
Any state $\omega$ on $\mathcal{M}$ is a positive, normalized linear
functional whose restrictions to the local sub-algebras are density
matrices $\rho_{[-n,n]}$, namely positive matrices in
$M_{[-n,n]}(\com)$ such that ${\rm Tr}_{[-n,n]}\rho_{[-n,n]}=1$:
\begin{equation}
\label{states} M_{[-n,n]}\ni A\mapsto \omega(A)={\rm
Tr}_{[-n,n]}\Big(\rho^{(n)}\,A\Big)\ .
\end{equation}
The degree of mixedness of such density matrices is measured by the
von Neumann entropy
\begin{equation}
\label{vNent} S(\rho_{[-n,n]})=-{\rm
Tr}_{[-n,n]}(\rho_{[-n,n]}\log\rho_{[-n,n]})=-\sum_jr_{[-n,n]}^j\log
r_{[-n,n]}^j\ ,
\end{equation}
where $0\leq r_{[-n,n]}^j\leq 1$, $\sum_jr_{[-n,n]}^j=1$, are the
eigenvalues of $\rho_{[-n,n]}$. Notice that the von Neumann entropy
is nothing but the Shannon entropy of the spectrum of
$\rho_{[-n,n]}$ which indeed amounts to a discrete probability
distribution.

In the above expressions ${\rm Tr}_{[-n,n]}$ stands for the trace
computed with respect to any orthonormal basis of the Hilbert space
$\h_{[-n,n]}=(\com^2)^{\otimes 2n+1}$ onto which $A$ linearly acts.
Let $\vert i\rangle\in\com^2$, $i=0,1$, be a chosen orthonormal
basis in $\com^2$; then, a natural orthonormal basis in
$\h_{[-n,n]}$ will consist of tensor products of single spin basis
vectors:
\begin{equation}
\label{ONBn} \vert\bi_{[-n,n]}\rangle=\bigotimes_{j=-n}^n\vert
i_j\rangle=\vert i_{-n}i_{-n+1}\cdots i_n\rangle\ ,
\end{equation}
namely its elements are indexed by binary strings
$\bi_{[-n,n]}\in\{0,1\}^{2n+1}$. By going to the limit of an
infinite chain, a corresponding representation Hilbert space is
generated by orthonormal vectors again denoted by
$\ket{\bi_{[-n,n]}}$ where $n$ arbitrarily varies and every
$\bi_{[-n,n]}$ is now a binary sequence in $\{0,1\}^\z$ where all
$i_k\notin[-n,n]$ are chosen equal to $0$. We shall denote by $\bi$
such binary strings, by $\bo$ their set and by $\ket{\bi}$ the
corresponding orthonormal vectors which form the so-called standard
basis of $\h$.
\begin{rem}
\label{remark1} While all representations of finite size quantum
spin chains are unitarily equivalent to the Fock representation
\cite{Strocchi}, what we are considering here is just one of the
infinitely many inequivalent Hilbert space representations for the
genuinely infinite quantum spin chain. Indeed, the representation
Hilbert space we are considering is a particular case of the
so-called GNS construction \cite{Strocchi}: it is created acting
with finitely many spin flips $\vert 0\rangle\mapsto\vert 1\rangle$
on the GNS cyclic state represented by all spins being in the state
$\vert 0\rangle$. By choosing $i_k=1$ outside any finite interval
$[-n,n]$ one gets another representation of the same algebra
$\mathcal{M}$. However, the new representation is inequivalent to
the previous one as there is no unitary operator mapping one Hilbert
space into the other. Such a unitary operator should indeed
flip infinitely many spins.
\end{rem}
\medskip

From~(\ref{states}), a compatibility relation immediately follows;
namely,
$$
\omega(A\otimes \one_{n})={\rm Tr}_{[0,n]}\rho_{[0,n]}(A\otimes
\one_n)=\omega(A)={\rm Tr}_{[0,n-1]}(\rho_{[0,n-1]}A)\qquad \forall
A\in M_{2^n}(\com) \ ,
$$
so that
\begin{equation}
\label{compatibility} {\rm Tr}_{n}\rho_{[0,n]}=\rho_{[0,n-1]}\ .
\end{equation}
On the other hand, if
$$
\omega(\one_{0}\otimes A)={\rm Tr}_{[0,n]}(\rho_{[0,n]}\one_{0}
\otimes A)=\omega(A)={\rm Tr}_{[0,n-1]}(\rho_{[0,n-1]}A)\qquad
\forall A\in M_{2^n}(\com)\ ,
$$
namely, if $\omega$ is a transationally invariant state, then
\begin{equation}
\label{translationinvariance} \rho_{[0,n-1]}={\rm
Tr}_{0}\rho_{[0,n]}\ ,\quad \forall n\ .
\end{equation}
To any translationally invariant state $\omega$ on a quantum spin
chain there is associated a well-defined von Neumann entropy
rate (see for instance \cite{Benattib}):
\begin{equation}
\label{vNentrate} s(\omega)=\lim_{n\to+\infty}\frac{1}{n}
S(\rho_{[0,n-1]})=-\lim_{n\to+\infty}
 \frac{1}{n} {\rm Tr}_{[0,n-1]}\,\Big(\rho_{[0,n-1]}\,\log\rho_{[0,n-1]}\Big)\ .
\end{equation}

\subsection{AF entropy}

The AF or AFL entropy  developed by Alicki, Fannes and Lindblad \cite{Alicki,
lindblad} is  an extension of the concept of KS entropy  in classical
dynamical systems to  discrete-time non-commutative quantum dynamical systems.
The construction of the AF entropy is based on the notion of quantum partitions of unity. These later together with the dynamics give rise, similarly to the classical case, to quantum symbolic models of quantum dynamical systems. By means of the von Neumann entropy, one then defines the AF entropy of a quantum dynamical system as  the optimal von Neumann entropy rate over all its quantum symbolic models.  Let $(\A, \Theta, \omega)$ be a quantum dynamical system.
\begin{defn}
A finite collection of operators
$\mathcal{Z}=\{Z_i\}_{i=1}^{|\mathcal{Z}|}$, where $Z_i\in\A$ is
called an operational partition of unity (OPU) if
\begin{equation}
 \sum_{i=1}^{|\mathcal{Z}|} Z_i^\dag Z_i = \mathds{1},
\end{equation}
 where $|\mathcal{Z}|$ is the cardinality of
$\mathcal{Z}$.
\end{defn}
\begin{itemize}
 \item The refinement of  two partitions $\mathcal{Z}_1=\{Z_{1i}
\}_{i=1}^{|\mathcal{Z}_1|}$ and $\mathcal{Z}_2=\{Z_{2j}
\}_{j=1}^{|\mathcal{Z}_2|}$, is defined naturally  by
$$\mathcal{Z}_1\circ\mathcal{Z}_2:=\{Z_{1i} Z_{2j}\}_{i,j=1}^{|\mathcal{Z}_1|
|\mathcal{Z}_2|},$$ which  is also an OPU. Moreover, time-evaluation
of an OPU $\mathcal{Z}= \{Z_i \}_{i=1}^{|\mathcal{Z}|}$ at time
$t=k\in\z$ under the dynamics $\Theta$ is OPU which  is defined by
\begin{equation}
\mathcal{Z}:=\Theta^k(\mathcal{Z})=\{\Theta^k(Z_i)
\}_{i=1}^{|\mathcal{Z}|}.
\end{equation}
\item Let $Z_{\ii}:= \Theta^{n-1}(Z_{n-1}) \ldots
\Theta(Z_{i_1})Z_{i_0}$. Clearly  the set
$\mathcal{Z}^{(n)}:=\{Z_{\ii}\}_{\ii\in
\Omega_{|\mathcal{Z}|}^{(n)}}$ is again an OPU. Now, for
$\mathcal{Z}=\{Z_j \}_{j=1}^{|\mathcal{Z}|}$, one can define a
$|\mathcal{Z}|\times |\mathcal{Z}|$ density matrix
$\rho[\mathcal{Z}]$ as follows
\begin{equation}
M_{|\mathcal{Z}|}(\com)\ni
\rho[\mathcal{Z}]:=\sum_{i,j=1}^{|\mathcal{Z}|}|z_i><z_j|\,\,
\omega(Z^\dag_j Z_i),
 \end{equation}
  where $\{|z_i>
\}_{i=1}^{|\mathcal{Z}|}$ is a fixed orthonormal basis in the finite
dimensional Hilbert space $\com^{|\mathcal{Z}|}$. Moreover, the density
matrix assiocated with $\mathcal{Z}^{(n)}$ has the
form
\begin{equation}
M_{|\mathcal{Z}|}(\com)^{\otimes n}\ni
\rho[\mathcal{Z}^{(n)}]:=\sum_{\ii,\jj\in
\Omega_{|\mathcal{Z}|}^{(n)} }^{n}|z_{\ii}><z_{\jj}| \omega(Z^\dag_{\jj}
Z_{\ii}),
\end{equation}
where
$$|z_{\ii}>:=|z_{i_1}> \otimes |z_{i_2}> \otimes \ldots |z_{i_n}>.$$
\end{itemize}
The translation invariance  $\omega \circ
\Theta=\omega$ and the compatibility relation are expressed by
\begin{equation}
Tr_{\{1\}}(\rho[\mathcal{Z}^{n+1}])= Tr_{\{n+1\}}\big(
\rho[\mathcal{Z}^{n+1}]\big)=\rho[\mathcal{Z}^{(n)}].
\end{equation}
 Thus the
family $\rho[\mathcal{Z}^{(n)}]$ denoted by
$\rho^{(n)}$ in section \ref{sec:quantum spin}, $n\in\n$ gives
a state $\omega_\z$ over $\A_\z:=\mathcal
{M}$, where  $\mathcal{M}$ is defined in  \ref{sec:quantum spin}. For a given quantum dynamical
system $(\A, \Theta, \omega)$ and a chosen OPU in a suitable subalgebra $\A_0$,  the AF-entropy is constructed over their
associated quantum symbolic system $(\A_\z, \Theta_\sigma,
\omega_\z)$, or quantum spin dynamics with right shift dynamics, together with given OPU. We
restrict ourselves to  subalgebra $\A_0$ because in general  the
mean von Neumann entropy of $(\A_\z, \Theta_\sigma, \omega_\z)$,
 with  the translation invariance  $\Theta$, may  not exist.
\begin{defn}
Let  $\A_0\subseteq\A$ be a $\Theta$-invariant subalgebra  and let $\mathcal{Z}\subseteq \A_0$ be an OPU. Let us define
\begin{equation}\label{AFentropyrate}
h_\omega^{AFL}(\Theta,
\mathcal{Z}):=\limsup_{n\to\infty}\frac{1}{n}S\big(
\rho[\mathcal{Z}^{(n)}]\big),
\end{equation}
where $S\big( \rho[\mathcal{Z}^{(n)}]\big)$ is the von Neumann
entropy of  the density matrix associated with the OPU
$\mathcal{Z}^{(n)}$. The AF entropy of the quantum dynamical system $(\A,
\Theta, \omega)$ is defined by
\begin{equation}
h_\omega^{AFL}(\Theta):=\sup_{\mathcal{Z}\subseteq \A_0}
h_\omega^{AFL}(\Theta, \mathcal{Z}).
\end{equation}
\end{defn}

\begin{rem}
\label{rem2} The $\limsup$ in~(\ref{AFentropyrate}) has to be used
for the sequence of density matrices $\rho[\mathcal{Z}^{(n)}]$ is
not a stationary one~\cite{Alicki,AFbook}. In fact, while
consistency holds as tracing $\rho[\mathcal{Z}^{(n)}]$ over the
$n$-th factor yields the density matrix corresponding to the first
$n-1$ factors, ${\rm
Tr}_n\rho[\mathcal{Z}^{(n)}]=\rho[\mathcal{Z}^{(n-1)}]$,
stationarity does not; indeed, in general, ${\rm
Tr}_1\rho[\mathcal{Z}^{(n)}]\neq \rho[\mathcal{Z}^{(n-1)}]$.
\end{rem}

\begin{example}
As a concrete example consider a set of $4$ matrix units $U_{ij}\in
M_2(\mathbb{C})$ such that $U^\dag_{ij}=U_{ji}$,
$U_{ij}U_{k\ell}=\delta_{jk}U_{i\ell}$ and $\sum_{i=1}^2U_{ii}=2$.
Dividing them by $\sqrt{2}$ one gets a partition of unit
$$
\mathcal{U}=\left\{\frac{U_{ij}}{\sqrt{2}}\right\}_{i,j=1,2}\in
M_2(\mathbb{C}) \ ,
$$
the simplest choice being
$$
U_{11}=\begin{pmatrix}1&0\cr0&0\end{pmatrix}\ ,\
U_{22}=\begin{pmatrix}0&0\cr0&1\end{pmatrix} \ ,\
U_{12}=U^\dag_{21}=\begin{pmatrix}0&1\cr0&0\end{pmatrix}\ .
$$
The refined partition that results after $n$ applications of the
right shift is
\begin{equation}
\label{nstepPU}
\mathcal{U}^{(n)}=\left\{\frac{U_{i^{(n)}j^{(n)}}}{2^{n/2}}\right\}\
,\ U_{i^{(n)}j^{(n)}}=U_{i_0j_0}\otimes U_{i_1j_1}\otimes \cdots
U_{i_{n-1}j_{n-1}}\in M_2^{\otimes n}(\mathbb{C})=M_{[0,n-1]}\ .
\end{equation}
The associated density matrices $\rho[\mathcal{U}^{(n)}]\in
M_{4^n}(\mathbb{C})$ have entries and von Neumann entropy given by
\begin{eqnarray}
\nonumber \frac{1}{2^n}{\rm Tr}\left(\rho^{(n)}\,
U^\dag_{i^{(n)}j^{(n)}}U_{k^{(n)}\ell^{(n)}}\right) &=&
\frac{1}{2^n}{\rm Tr}\left(\rho^{(n)}\,
U_{j_0i_0}U_{k_0\ell_0}\otimes
U_{j_1i_1}U_{k_1\ell_1}\otimes\cdots\right)
\\
\label{nsteprho} &=& \frac{1}{2^n}{\rm Tr}\left(\rho^{(n)}\,
\delta_{i_0k_0}\, U_{j_0\ell_0}\otimes
\delta_{i_1k_1}\,U_{j_1\ell_1}\otimes\cdots\right)
\\&=&
\frac{\mathbf{1}}{2^n}\otimes \rho^{(n)}\\
\label{nstepvNent}
S\Big(\rho[\mathcal{U}^{(n)}]\Big)&=&S(\rho^{(n)})\,+\,n\ .
\end{eqnarray}
The last equality in~(\ref{nsteprho}) comes from the fact that ${\rm
Tr}\Big(\rho^{(n)}U_{i^{(n)}j^{(n)}}\Big)$ are the matrix elements
of $\rho^{(n)}$ with respect to the orthonormal basis defined by the
choice of matrix units. Entropy rate and the Alicki-Fannes entropy
then result
\begin{equation}
\label{AFshiftent}
h^{AF}_\omega(\Theta)=h^{AF}_\omega(\Theta,\mathcal{U})=\limsup_{n\to\infty}\frac{1}{n}S(\rho[\mathcal{U}^{(n)}])=s(\omega)\,+\,1\
.
\end{equation}
\end{example}
\textbf{Properties of the AF entropy }
\begin{itemize}
        \item When a quantum dynamical system $(\A, \Theta,
    \omega)$ reduces to a classical dynamical system $(\chi, T,
    \nu)$,  the AF entropy of the triplet  $(\mathcal{M}, \Theta_T,
    \omega_\nu)$ is
    \begin{equation}
  h_\omega^{AF}(\Theta_T, \mathcal{M})=h_\nu^{KS}(T),
  \end{equation}
   where
    $\mathcal{M}:=L^\infty_\mu(\chi)$ and $h_\nu^{KS}(T)$ is the Kolmogorov-Sinai entropy.
       %It is better to write some causes about why we consider the AFL
    %entropy while in the finite level systems it is $0$ [fabio book page
    %466 ]. [remark 8.2.3 page 469 fabio book]
    \item Let $(\A_\z, \omega)$ be a quantum spin chain with
    single site matrix algebra $M_d(\com)$. The AF entropy with
    respect to every local subalgebra $\A_{[p, q]}\subseteq
    \A_0$ is given by
    \begin{equation}
     h_\omega^{AF}(\Theta_\sigma)= s(\omega)+\log{d},
    \end{equation}
    where the dynamics is the right-shift $\Theta_\sigma$ over
    $\A_\z$, and $s(\omega)$ is the mean von Neumann entropy of the translation-invariant $\omega$ (see  Section 2.3.4).
\end{itemize}

\subsection{CNT Entropy}
 The CNT entropy introduced by Connes, Naranhofer and Thiring \cite{CNT} is a generalization of
the KS-entropy to quantum dynamical systems which is based on convex
decompositions of the state $\omega$.

\begin{defn}
Let $M$ and $\A$
be two $C^*$-algebras. A linear map $\gamma:M\to\A$ is called
completely positive if $\gamma\otimes id_n: M\otimes Mat_n(\com)\to
\A\otimes Mat_n(\com)$ is a positive operator for each $n\in\n$.
\end{defn}
Each state $\omega$ is the form  $\omega(A)={\rm Tr}(A \rho)$ for a unique positive  element, or density matrix, $\rho$. The entropy of the state $\omega$ is defined the von Neumann entropy of the associated density matrix.

Let us consider a convex
decomposition
\begin{equation}\label{cnt1}
\omega=\sum_{\ii\in I^{(n)}} \lambda_{\ii} \omega_{\ii}, \quad
I^{(n)}= I_1 \times I_2 \times \ldots \times I_n,
\end{equation}
 where $\lambda_{\ii}$ are
positive weights and $I_j$'s are generic index sets. The marginal density matrices arising from this
decomposition is denoted by $\omega=\sum_{i_j\in I_j}
\lambda^j_{i_j} \omega_{i_j}^j, j=1, 2, \ldots, n$, where
\begin{equation}\label{cntd}
\omega_{i_j}^j=\sum_{\ii, i_j \,\text{fixed}}
\frac{\lambda_{\ii}}{\lambda^j_{i_j}}\omega_{\ii}, \,\,\,
\lambda^j_{i_j}=\sum_{\ii, i_j \,\text{fixed}} \lambda_{\ii}.
\end{equation}
 Let
$\Lambda^{(n)}=\{\lambda_{\ii}^{(n)}   \}$ and
$\Lambda_j=\{\lambda_{i_j}^{j}   \}$ be probability distributions
associated with the scaler products in (~\ref{cnt1}) and (~\ref{cntd}).
\begin{defn}
Let $\mathcal{A}$ be a $C^*$-algebra  endowed with a state $\omega$. Let $\gamma_i:
M_i\subset\A$, $i=1, 2, \ldots, n$ be CPU
maps from finite dimensional $C^*$-algebras into $\A$. Their entropy
with respect to $\omega$ is:
\begin{eqnarray*}
& &H_\omega(\gamma_1, \gamma_2, \ldots, \gamma_n):=\\
&=&
 \sup_{\omega=\sum_{\ii}\lambda_{\ii}\omega_{\ii}} \left\{
H(\Lambda^{(n)})-
 \sum_{j=1}^n H(\Lambda_j)+ \sum_{j=1}^n\sum_{i_j\in I_j} \lambda^j_{i_j}
  S(\omega^j_{i_j} \circ\gamma_j,
  \omega\circ\gamma_j)\right\},
\end{eqnarray*}
where $\omega\circ\gamma_j$ is a state over $M$ and $H$ is the Shannon entropy.

\end{defn}
The CNT entropy rate for a completely positive map $\gamma:
\textbf{M}\longmapsto \A$ where $\textbf{M}$ is a finite dimensional
$C^*$-subalgebra of $\A$, is defined as follows
\begin{equation}
 h^{CNT}_\omega(\Theta, \gamma)=\lim_{n\to\infty} \frac{1}{n}
H_\omega(\gamma, \Theta\circ\gamma, \ldots,
\Theta^{n-1}\circ\gamma).
\end{equation}
 The exsitece of the above limit is shown in \cite{CNT}. The CNT
dynamical entropy is defined by
\begin{equation}
h_\omega^{CNT}(\Theta)=\sup_\gamma h_\omega^{CNT}(\Theta, \gamma).
\end{equation}

It is proved in \cite{CNT} that in $d$-level quantum spin chains
with shift dynamics,
\begin{equation}
h_\omega^{CNT}(\Theta)=s(\omega).
\end{equation}

\begin{example}
Let us consider the quantum spin chain $(\mathcal{M}, \Theta_\sigma, \omega)$ with right  shift dynamics, where the state $\omega$ is defined using the density matrix $\rho^{(n)}=\underbrace{\rho\otimes\cdots\otimes\rho}_{n\,\, \text{times}}$. Then, we have
$$h_\omega^{CNT}(\Theta)=s(\omega)=\lim_{n\to\infty}\frac{S(\rho^{(n)})}{n}=\lim_{n\to\infty}\frac{n S(\rho)}{n}=S(\rho).$$
\end{example}
\subsection{Relation Between CNT and AF Entropies in Quantum Spin Chains }
In this section we show that from physical point of view
 in $2$-level quantum spin chains with shift dynamics
$h^{AF}_\omega(\Theta)=h^{CNT}_\omega(\Theta)+ 1$ \cite{Benatti1}.

Consider a two level spin chain $\mathcal{M}_2$ where
$h_\omega^{CNT}(\Theta)=s(\omega)$ and
$h_\omega^{AFL}(\Theta_\sigma)= s(\omega)+1$ The origin of the
difference by $1=\log 2$ between the AF-entropy and the entropy rate
(which is equal to the CNT-entropy) lies in that the AF-entropy accounts
for measurement-like disturbances on the state of the quantum chain.
In quantum mechanics generic measurement processes on a system in a
state described by density matrix $\rho$ are identified by
partitions of unity $\mathcal{Z}=\{Z_i\}$ and the state is
changed by the measurement process as follows:
\begin{equation}
\label{measure} \rho\mapsto\sum_iZ_i\,\rho\,Z^\dag_i\ .
\end{equation}
Suppose
\begin{equation}
\label{purification0} M_{2^n}(\mathbb{C})\ni\rho^{(n)}=\sum_i
r_i^{(n)}\,\vert r_i^{(n)} \rangle\langle r_i^{(n)} \vert
\end{equation}
is the spectral decomposition of a local state for $n$ qubits
described by the local algebra $M_{[0,n-1]}$; any such mixed state
can be purified, that is transformed into a projector, by coupling
$M_{[0,n-1]}$ to itself and by doubling $\rho^{(n)}:=\rho_{[0,n-1]}$
into
\begin{equation}
\label{purification1}
\com^{2^n}\otimes\com^{2^n}=\com^{4^n}\ni\vert\sqrt{\rho^{(n)}}\rangle=\sum_i\sqrt{r_i^{(n)}}\,\vert
r_i^{(n)}\rangle\otimes \vert r_i^{(n)}\rangle\ .
\end{equation}
Given the refined partition of unity $\mathcal{U}^{(n)}$
in~(\ref{nstepPU}), one further amplifies the Hilbert space from
$\mathbb{C}^{4^n}$ to $\mathbb{C}^{4^n}\otimes\mathbb{C}^{4^n}$ and
constructs the following vector state
\begin{equation}
\label{purification}
\mathbb{C}^{4^n}\otimes\mathbb{C}^{4^n}\ni\vert\Psi[\mathcal{U}^{(n)}]\rangle
=\sum_i\sum_{(k^{(n)}\ell^{(n)})} \sqrt{r^{(n)}_i}\,
U_{k^{(n)}\ell^{(n)}}\vert r^{(n)}_i\rangle\otimes\vert
r^{(n)}_i\rangle\otimes\vert k^{(n)}\ell^{(n)}\rangle\  ,
\end{equation}
where the vectors $\vert k^{(n)}\ell^{(n)}\rangle$ indexed by pairs
of binary strings in $\Omega_2^n$ form an auxiliary orthonormal
basis in $\mathbb{C}^{4^n}$ of cardinality $2^n\times 2^n$.

One thus sees that $\vert\Psi[\mathcal{U}^{(n)}]\rangle$ is the
vector state of a three-partite system consisting of the  $n$
qubits, system $I$, a copy of the latter, system $II$, and a copy of
the first two, system $III$. From the projection
$P=\vert\Psi[\mathcal{U}^{(n)}]\rangle\langle\Psi[\mathcal{U}^{(n)}]\vert$,
by tracing over the first two systems, respectively over the last
one, one obtains the following marginal states on
$M_{[0,n-1]}\otimes M_{[0,n-1]}$,

\begin{eqnarray}
\label{marg1}{\rm Tr}_{I,II}(P)&=&\rho[\mathcal{U}^{(n)}]\ ,\qquad\hbox{respectively}\\
\label{marg2} {\rm
Tr}_{III}(P)&=&\sum_{(k^{(n)}\ell^{(n)})}U_{k^{(n)}\ell^{(n)}}\otimes
1\,
\vert\sqrt{\rho^{(n)}}\rangle\langle\sqrt{\rho^{(n)}}\vert\,U^\dag_{k^{(n)}\ell^{(n)}}\otimes
1 =R[\mathcal{U}^{(n)}]\ .
\end{eqnarray}
Since the latter states are marginal density matrices of a pure
state, they have the same spectrum and thus the same von Neumann
entropy (see for instance \cite{Benattib})
$$
S\Big(\rho[\mathcal{U}^{(n)}]\Big)=S\Big(R[\mathcal{U}^{(n)}]\Big)
=S(\rho^{(n)})\,+\,n\ .
$$
Thence, the entropy associated to $\omega$ and to the partition of
unity $\mathcal{U}^{(n)}$, that is $\rho[\mathcal{U}^{(n)}]$, is also
the entropy of the state $R[\mathcal{U}^{(n)}]$ which results from
the action of the POVM $\{ U^\dag_{k^{(n)}\ell^{(n)}}\otimes 1\}$ on
the purified state
$\vert\sqrt{\rho^{(n)}}\rangle\langle\sqrt{\rho^{(n)}}\vert$.

\chapter{Semi-computable States and Semi-computable Density Matrices
}\label{ch:semi-compuatble density matrises}
\goodbreak

In this chapter, we will look at quantum mechanical tools as Hilbert space vectors, density matrices and generic linear operators on them from the point of view of computability theory. This is necessary in order to introduce the concept of Gacs complexity which is based on a quantum extension of the classical notion of universal semi-measure  devised for finite-dimensional quantum systems to infinite dimensional separable Hilbert spaces. We shall then use Gacs complexity to present a Brudno's like relation for quantum spin chains.

\smallskip
\section{Universal Semi-computable Semi-density Matrices on Infinite Separable Hilbert spaces}
\label{sec:unversal semi-measure} We start by fixing the necessary notations and symbols.

\begin{enumerate}
%\item
%We will define a canonical isometric isomorphism from $\mathbb{H}\bigotimes\mathbb{H}$
%onto $\mathbb{H}$: this will allow us to identify semi-computable semi-density matrices on
%$\mathbb{H}\bigotimes\mathbb{H}$ with semi-computable semi-density matrices on $\mathbb{H}$.
\item
Let the set $Q'$  be defined as follows
$$
Q'=\{(\varepsilon, p, q)\in\{0,1\}\times \n_+\times \n_+|\, p\,\,
\text{and}\,\, q \,\,\text{are coprime}\}\bigcup \{(0, 0, 0)\}\ .
$$
The mapping $\iota:Q\rightarrow Q'$  defined by
$$\iota(0)=(0, 0, 0),\quad \iota(\frac{p}{q})=(0,p,q),\quad
\iota(-\frac{p}{q})=(1,p,q)$$ is bijective and the mapping
$\iota':Q'\rightarrow \n$ defined by $(\varepsilon,p,q)\rightarrow
<\varepsilon,<p,q>>$, where $<x,y>=2^x(2y+1)-1$, is injective. We
can identify $Q$ with the subset $\iota'\circ\iota(Q)$ of $\n$.
%%\begin{equation}
%integers\label{integers} (p,q)\mapsto <p, q> = 2^p( 2q+1)-1\ .
%\end{equation}
Similarly, any finite dimensional rational matrix will be
represented by a natural number.
\item
With reference to the indexing of the standard basis
in ~(\ref{ONBn}), we shall consider the set of all functions from
$\z$ into the set $\{0, 1\}$ with finite support and denote it by
$\bo$. Let $\bi\in\bo$ and let $\theta :\z\rightarrow \z$ be the
left shift $\theta(n) = n-1$. Then $\theta$ induces the map
$(\theta(\bi ))_n = \bi_{n+1}$ on $\bo$. The restriction of $\bi$ to
the subinterval $I$ will be denoted by $\bi_I$. Furthermore, let
$p,q\in \z$ and $p\leq q$. Assume that the support of $\bi\in\bo$ is
contained in the interval $[p, q]\subset \z$. Then, $\bi =
\0]\bi_{[p, q]}[\0$ , where $\0] = \bi_{p-1]}$ and $[\0 =
\bi_{[q+1}$ are infinite sequences of $0$'s.
\item
The map $\bo\rightarrow \n\times\n$ that associates to $\bi\in \bo$
the integers
$$
(x =\sum_{k<0} \bi_k 2^{-k}\ ,\ y =\sum_{k\geq 0} \bi_k 2^k)
$$
is bijective. Therefore, the following two maps
\begin{equation}
\label{bijections} \bo\ni\bi\mapsto \eta(\bi) = <x, y>\ ,\quad \nu(
\bi ) = y- sign(x)[(<x, y>+1)/2+y]\ ,
\end{equation}
where  $sign(x)=0$ if $x=0$ otherwise $sign(x)=1$, are bijections
between $\bo$ and $\n$, respectively $\z$. Then, the inverse mapping
\begin{equation}
\label{inversemap} \zeta :(\bi, \bj)\mapsto
\nu^{-1}(\eta(\bj)-sign(\eta(\bi))[(<\eta(\bi),\eta(\bj)>+1)/2+\eta(\bj)])
\end{equation}
identifies $\bo \times\bo$ with $\bo$.
\item
Let $\Sigma$ be the power set of $\bo$. For $A\in\Sigma$, let
$\mu(A) =\#(A)$. Given the measure space $(\bo,\Sigma ,\mu)$, by the
identification of $\bo$ with $\z$, the Hilbert space $L^2( \bo,
\Sigma, \mu)$ consists of the square-summable functions
$f:\z\mapsto\com$, i.e. $\sum_{x\in\z}\,|f(x)|^2 < \infty$. For any
$\bi\in\bo$, consider the function $\delta_\bi$ defined by
$$
\delta_\bi(\bi)=1\ ,\qquad \delta_\bi(\bj)=0\quad\forall\
\bj\neq\bi\ .
$$
The set of these functions which is in one-to-one correspondence
with $ \bo$ is a Hilbert basis for $L^2( \bo, \Sigma, \mu)$ and for
each $\bi\in\bo$, $\delta_\bi$ will be denoted by $|\bi>$.
Therefore, the representation Hilbert space $\h$  for the quantum
spin chain is isomorphic to $L^2( \bo, \Sigma,\mu)$.
\item
The mapping $\zeta$ identifies $\h\bigotimes\h$ with $\h$.
Furthermore, the set of all elements $\bi\in \bo$  with support
included in $[ -n, n ]$ will be denoted by $\bo_{[-n,n]}$. The
subspaces of $L^2( \bo,\Sigma,\mu)$ generated by $\bo_{[-n,n]}$,
namely $L^2( \bo_{[-n,n]})$, are isomorphic to the local quantum
spin Hilbert spaces $\h_{[-n,n]}=\com^{\otimes 2n+1}$. The
corresponding orthogonal projections from $\h$ onto $\h_{[-n,n]}$
will be denoted by  $P_n$, and the canonical injection from
$\h_{[-n,n]}$ into $\h$ will be denoted by $\bi_n$. In the following
we will identify $\h_{[-n,n]}$ with the subspace
$\bi_n(\h_{[-n,n]})$ of the Hilbert space $\h$.
\item
For a  linear operator $T$ on $\h$, $P_n\,T\,P_n$ will
be denoted by $T_n$.
\end{enumerate}

\begin{defn}
\label{def1}\hfill
\begin{enumerate}
\item
A vector $\ket{\psi}=\sum_{\bi\in\bo}a_\bi|\bi>\in \mathbb{H}$ will
be termed elementary if of its expansion coefficients $a_\bi$ with
respect to the fixed orthonormal basis $\{|\bi>\}$ only a finite
number is not zero and those are algebraic numbers.
\item
A state $| \psi>=\sum_{\bi\in\bo}a_\bi|\bi>\in\h$ where $a_i\in\rr$,
will be termed semi-computable if there exist a computable sequence
of elementary vectors $| \psi_n> = \sum_{i\in \n} a_{n,i}| i>$ and a
computable function $k: \n \rightarrow \q$, such that
$\lim_{n\rightarrow\infty} k_n = 0$, and for each $n, | a_\bi-
a_{n,i}| \leq k_n$. Since the set of all computable functions is
countable, the set of all semi-computable elements of $\h$  is
countable.
\item
A linear operator $M_{2^{2n+1}}(\com)\ni T:\h_{[-n,n]}\rightarrow
\h_{[-n,n]}$, will be called elementary if the real and imaginary
parts of all of its matrix entries are rational numbers. It follows
that the elementary operators can be numbered.
\item
The linear operator $T:\h\rightarrow \h$, is a semi-density matrix
if  $T$ is positive and $0\leq {\rm Tr}(T)\leq 1$.
\item
Let  $n_1,n_2\in\n$ and $n_1\leq n_2$. Let
$T_j:\h_{[-n_j,n_j]}\rightarrow \h_{[-n_j,n_j]}$, $j=1,2$, be two
linear operators: $T_2$ will be said to be quasi-greater than $T_1$,
$T_1\leq_q T_2$, if $P_{n_1}\,T_2\,P_{n_1}-T_1\geq0$, where
$P_{n_1}$ is the canonical projection from $\h_{n_2}$ to $\h_{n_1}$.
A sequence of linear operators $T_n:\h_{[-n,n]}\rightarrow
\h_{[-n,n]}$ will be called quasi-increasing if for all $n\geq 1$,
$T_{n+1}\geq_q T_n$.
\end{enumerate}
\end{defn}

\begin{lem}
Each elementary state can be identified by a  natural number.
\end{lem}
\begin{proof}
The  complex number $z$ is said algebraic number if there are integer numbers $x_0,\ldots,x_n$, not all zero, such that $p(z)=x_0 z^{n}+x_1 z^{n-1}+\ldots+x_{n-1} z+x_n=0$.

Now, we arrange the roots of any polynomial $p(z)=0$ by the lexicographical order as $(z_0, \ldots z_n)$.
 Let's define
 $$w(z_i)=2^n 3^{x'_0} \cdots p_{n+2}^{x'_n}p_{n+3}^i,$$ where $x'_j=f(x_j)$, $f:\zz\to\n$ is one-to-one and surjective function.

Let $\ket{\psi}=\sum_{\bi\in\bo}a_\bi|\bi>\in \mathbb{H}$ be an elementary state, where $a_{\bi}$ is algebraic number. We also define $$w'(\ket{\psi})= 2^n 3^{w(a_0)} \cdots p_{n+2}^{w(a_n)},$$ where $n$ is the smallest  number such that  $a_{\bi}= 0$, for $\bi\notin [-n, n]$. Therefore, each state can be identified by a natural number.
\end{proof}

\begin{prop}
\label{thm:4} Let $T_n$ be a quasi-increasing sequence of
semi-density matrices on $\h$. Then $\lim_{n\rightarrow\infty}T_n$
converges in the trace-norm to a semi-density matrix.
\end{prop}
\begin{proof}
Since the sequence $T_n$ is quasi-increasing, (${\rm Tr}(T_n)$) is an
increasing sequence and since for every $n$, ${\rm Tr}(T_n)\leq 1$,
the sequence converges in trace-norm, $\|X\|_{tr}={\rm
Tr}\sqrt{X^\dag X}$ to an operator $T$ in the Banach space $T(\h)$
of trace-class operators on $\h$, moreover
$$
{\rm Tr}(T)=\lim_{n\to+\infty}{\rm
Tr}(T_n)=\lim_{n\to+\infty}\|T_n\|_{tr}=\|T\|_{tr}\leq 1\ .
$$
Therefore, $T$ must be positive; otherwise, if $T$ had negative
eigenvalues then $\|T\|_{tr}>{\rm Tr}(T)$ and this would contradict
the previous equality.
\end{proof}
Now,  we give the
definition of semi-computable semi-density matrices.
\begin{defn}
\label{semicompsemidens} A linear operator $T$ on $\h$ is a
semi-computable semi-density matrix, if there exists a computable
quasi-increasing sequence of elementary semi-density matrices
$T_n\in B(\h_{[-n,n]})\subseteq B(\h)$ such that
$\lim_{n\rightarrow\infty}\|T-T_n\|_{tr}=0$.
\end{defn}
The following lemma gives us a method for checking
the positivity of a matrix.

A polynomial $P\in\com [x]$ of degree
$n$ is called of type $\Pi$ if it has the following form:
\[P(x)=\sum_{0\leq k\leq n} (-1)^k \lambda_k x^{n-k},\,\ \text{and}\, \lambda_0=1.\]
\begin{lem}
Let $P \in \com [ x ]$ be of type $\Pi$. Assume that all solutions
of the equation $P( x ) = 0$ are real. Then these solutions are all
 positive.
 \end{lem}
\begin{proof}
 Consider the following system of equations:
\[\sum_{1\leq k_1\leq n}\lambda_{k_1}= a_1, \sum_{1\leq k_1< k_2\leq
 n}\lambda_{k_1}\lambda_{k_2}= a_2,
 \sum_{1\leq k\leq l\leq_q\leq n}\lambda_k\lambda_l\lambda_q= a_3,
 \dots , \lambda_{1}\lambda_{2} \dots\lambda_{n}=a_n,\]
 where $a_i\geq 0$ for any $1\leq i\leq n$.
 To prove the lemma it is sufficient to prove that under the above conditions, all $\lambda_i$'s are positive. Assume that $\lambda_n$ is negative. From the
above system we obtain the following one:
 \[\sum_{1\leq k_1\leq {n-1}}\lambda_{k_1}= a_1-\lambda_n ,\]
 \[\sum_{1\leq k_1< k_2\leq {n-1}}\lambda_{k_1}\lambda_{k_2}=
 a_2+\lambda_n(a_1-\lambda_n),\]
 \[\lambda_{1}\lambda_{2} \dots \lambda_{n-1}=a_{n-1}-\lambda_n
\sum_{1\leq
 k_1< k_2< \dots < k_{n-2}\leq n-1}\lambda_{k_1}\lambda_{k_2}\dots
\lambda_{k_{n-2}}.\] All right hand sides are positive. Therefore,
$\lambda_1\lambda_2 \ldots \lambda_{n-1}$ is positive. But
$\lambda_n$ is negative and $\lambda_1\lambda_2 \ldots
\lambda_{n-1}\lambda_n$ is positive. This
 is a contradiction.
\end{proof}

\begin{thm}\label{thm:semicom}
 The set of all semi-computable semi-density matrices on $\h$ can be
 enumerated.
\end{thm}

\begin{proof}
Let $\phi_0, \phi_1,\dots , \phi_n, \dots$ be the standard
enumeration of all partially computable functions on $\n$. For
$n\in\n$, we change $\phi_n$ into $\psi_n$ which represents a
semi-computable semi-density matrix $\rho_n$ on $\h$. Let
$\psi_n(0)=0$. Assume that
  $\psi_n(x)$ is defined for $0\leq x\leq t-1$ and $z$ is the smallest integer number such that $\psi_n(t-1)=\phi_n(z)$.
  To define $\psi_n(t)$, assume that
  there is a least integer number $x_0$, $0\leq x_0\leq
   t$, greater than $z$, satisfying the relation
    $STP^{(1)}(x_0, n, t) = 1$ and
   %predicate $STP^{(1)}(n, x, t)$ for $0\leq x_0\leq t$. Let $\bigvee_{0\leq x\leq t} STP^{(1)}(n, x, t) = 1$.
   % Assume that $x_0$ is the greatest integer $x$, $0\leq \underline{x}\leq x\leq t$ ($\underline{x}$ is the smallest
    % integer $z$ such that $\phi_n(z) = \phi_n(t-1)$) such that $STP^{(1)}(n, x_0, t) = 1$
   $\phi_n(x_0)$ can
be interpreted as an elementary semi-density matrix $\rho_n(t)$
strictly quasi-greater than $\rho_n(t-1)$.  Then we set $\psi_n(t)
=\phi_n(x_0)$. Otherwise, $\psi_n(t) =
 \psi_n(t-1)$.
% Then we set $\psi_n(t) =\phi_n(x_0)$. Otherwise, $\psi_n(t) = \psi_n(t-1)$.
Clearly, $\psi_n$ is a computable function and by Theorem
~\ref{thm:4}, $\lim_{t\rightarrow\infty}\rho_n(t)$ is a
semi-computable semi-density matrix. Conversely to each
semi-computable semi-density matrix on $\h$ there corresponds a
computable function $\psi:\n\to\n$ of the above form.
\end{proof}
\begin{thm}
Let $S$ and $T$ be semi-density matrices on $\mathbb{H}$, and
let $T$ be invertible. If $S \leq T$, then $\sqrt{S}(\log
S)\sqrt{S}\leq\sqrt{S}(\log T)\sqrt{S}$.
 \end{thm}

\begin{proof}
For $0< t \in R$, both $t + S$ and $t + T$ are invertible and $( t+S
)^{-1} \geq( t+T)^{-1}$. Therefore,
$$\sqrt{S} \left( \int_0^\infty{(1/(t+S)- 1/(t+T))dt} \right)\sqrt{S} =
\sqrt{S}\left(\log{(t+S)}-\log{(t+T)}\right)\sqrt{S}\Big|_0^{+\infty}\geq 0.$$
 But
$$\sqrt{S} \log{(t+S)}  \sqrt{S} = \sum_0^\infty
\lambda_i\log{(t+\lambda_i)}|\varphi_i><\varphi_i|,$$
 where, $\lambda_i$'s are eigenvalues
of $S$ with associated eigenvectors $|\varphi_i>$. Since by
convention $0log0=0$, the operators $\sqrt{S}\log{(S)}\sqrt{S}
=\sum_0^\infty \lambda_i\log{(\lambda_i)}|\varphi_i><\varphi_i|\leq
0$ is  well defined. On the other hand for
 $t \neq 0$, $$\sqrt{S}\left(\log{(t+S)}-\log{(t+T)}\right)\sqrt{S} =\sqrt{S}\left(\log{(1+S/t )}- \log{(1+T/t
 )}\right)\sqrt{S}.$$
Therefore,
\begin{eqnarray*}
0&\leq& \sqrt{S}\left( \log{(t+S)} -
\log{(t+T)}\right)\sqrt{S}\Big|_0^{+\infty}\\
&=&
\sqrt{S}(\log{T}-\log{S})\sqrt{S}+
\lim_{t\rightarrow\infty}\sqrt{S}\left(\log{(1+S/t
)}- \log{(1+T/t )}\right) \sqrt{S}\\
&=&
\sqrt{S}(\log{T}-\log{S})\sqrt{S}
.
\end{eqnarray*}
\end{proof}

\begin{defn}
\label{unvsemdens} A semi-computable semi-density matrix $\hat{\mu}$
is called universal if for any semi-computable semi-density matrix
$\hat{\mu}$ there exists a constant $C_\rho>0$ such that
$C_\rho\,\rho\leq\hat{\mu}$.
\end{defn}
The existences of a universal semi-density matrix in finite dimensional Hilbert spaces
and its applications to algorithmic complexity is proved in
\cite{Zvonkin70thecomplexity}. Based on the preceding discussion, we are now able to show that universal semi-densities exist in infinite dimensional separable Hilbert space, and that they are related to each other by a universality condition.
\begin{thm} \label{thm:universal}
There exists a universal semi-computable semi-density matrix on
 any infinite dimensional, separable Hilbert space $\mathbb{H}$.
\end{thm}

\begin{proof}
Let $\hat{\mu}_0 , \hat{\mu}_1 , \ldots ,\hat{\mu}_n , \ldots$ be
the enumeration of all semi-computable semi-density matrices and set
\begin{equation}
\label{chain-univ} \hat{\mu}=\sum_{k\geq 0}2^{-k}\hat{\mu}_k\ .
\end{equation}
Clearly, $\hat{\mu}$ is a semi-computable semi-density matrix, and
for each semi-computable semi-density matrix $\hat{\mu}_k$ we have
$2^{-k} \hat{\mu}_k\leq \hat{\mu}$. Therefore, $\hat{\mu}$ is a
universal semi-computable semi-density matrix.
\end{proof}
\section{Semi-computable operators}
Let $T$ be an bounded operator on $\mathbb{H}$. Then, $T$ can be written as
 $T=(T_1+ i T_2)/2$ where $T_1=(T+T^\dag)/2$ and
$T_2=(T-T^\dag)/2i$ are self-adjoint operators.
Moreover, each
self-adjoint operator $T_1, T_2\in\h$ can be written as $T_i= T_{i1} - T_{i2}$, $i=1,2$, where $T_{ij}$, $i, j=1,
2$, are positive operators and $T_{ij}/\|T_{ij}\|\leq I$. Indeed, $T_{i1}=(|A|+A)/2$ and $T_{i2}=(|A|-A)/2$, $i=1, 2$ \cite{beratteli1}.

Now, let $T$ be a positive linear operator $\leq I$. Then
$T_n=P_nTP_n$ is called elementary if all of its matrix
elements are rational numbers.

A mapping $\phi:\n\to\n$ is interpreted as a semi-computable linear operator
$T$ from $\h$ into $\h$ if for each $n\in\n$, $\phi(n)$ has the form
$\phi(n)=<\lambda, <\phi_1(n), \phi_2(n)>, <\phi'_1(n),
\phi'_2(n)>>$, where $\lambda$ is an integer number independent of
$n$ and $\phi_1(n)$, $\phi_2(n)$, $\phi'_1(n)$, and $\phi'_2(n)$ can
be interpreted as elementary positive operators $T_{1n}$, $T_{2n}$,
$T'_{1n}$ and $T'_{2n}$ all less than or equal to $\lambda I$ and
the sequences $T_{1n}$, $T_{2n}$, $T'_{1n}$ and $T'_{2n}$ are all
quasi-increasing and
$T=\lim_{n\to\infty}(T_{1n}-T_{2n})/2+(T'_{1n}-T'_{2n})/2$. If for
each $n$, $(T_{1n}-T_{2n})=0$, or $(T'_{1n}-T'_{2n})=0$, $T$ is a
semi-computable self-adjoint operator, and if for each $n$, three of
four operators $T_{1n}$, $T_{2n}$, $T'_{1n}$ and $T'_{2n}$, are zero,
then $T$ is a semi-computable bounded positive operator.
\begin{defn}\label{defnsemiope}
With the above notations $T\in B(\h)$  is called a semi-computable semi-unitary
operator if for each $n$,
$$T_n T^\dag_n\leq I \quad\text{and}\quad T^\dag_n T_n\leq I,$$
where, $T_n=(T_{1n}-T_{2n})/2+(T'_{1n}-T'_{2n})/2$.
\end{defn}
\begin{lem} Let $T$ and $S$ be  semi-computable semi-unitary
operators. Then \begin{enumerate}
    \item $T\circ S$ is also a semi-computable semi-unitary  operator.
    \item $T^\dag$ is also a semi-computable semi-unitary operator.
\end{enumerate}
\end{lem}
\begin{proof}
Since $T$ and $S$ are semi-computable semi-unitary operators then they are constructed by sequences $T_n$ and $S_n$ convergent in trace-norm to $T$ and $S$.  For each $n\in\n$, we have
$$T_n=(T_{1n}-T_{2n})/2+(T'_{1n}-T'_{2n})/2,$$
 $$S_n=(S_{1n}-S_{2n})/2+(s'_{1n}-S'_{2n})/2,$$
 where $T_{1n}$, $T_{2n}$, $T'_{1n}$, $T'_{2n}$, $S_{1n}$, $S_{2n}$, $S'_{1n}$ and $S'_{2n}$ are elementary operators.

 It is clear that multiplications and adjoint of  elementary operators are also elementary. Therefore, $T^\dag_n$ and $T_n\circ S_n$ are constructed from elementary operators and hence
 $T^\dag$ and $T\circ S$ are also semi-computable semi-unitary operators.
 \end{proof}

\section{Lower and Upper Gacs Complexities}
In this section with the help of a universal semi-computable semi-density matrix we will give the lower and upper Gacs algorithmic complexities in an infinite dimensional separable Hilbert space.
\begin{defn}
Let $\rho$ be a semi-computable semi-density matrix on the Hilbert space
$\mathbb{H}$. The lower and upper Gacs algorimic complexities are
defined by
\begin{equation}
 \lh(\rho)=-\log{\rm Tr}(\rho \hat{\mu}),
\end{equation}
and
\begin{equation}\label{qalgentr1}
 \oh(\rho)=-{\rm
Tr}(\rho \log\hat{\mu}).
\end{equation}
\end{defn}

By the Levin's theorem \ref{thm:levin2}, we have $K(x)\overset{+}{=}\mu(x)$, $x\in\n$. Now, it is natural to define a like relation in $\mathbb{H}$ for $\hat{\mu}$. Since, $\hat{\mu}$ is a positive operator less that $\mathds{1}$, we define
\begin{equation}
\label{qcompop} \kappa=-\log\hat{\mu}\ .
\end{equation}

\begin{thm}
Let $f$ be a convex function on an interval $[a, b]$
containing all the eigenvalues of positive operator $A$, then for
all density matrices $\rho$ such that $tr(\rho f(A)) < \infty$,

\begin{equation}
 f({\rm Tr}(\rho A)) \leq {\rm Tr}(\rho f(A)).
 \end{equation}
\end{thm}

\begin{proof}
Let us consider the spectral decomposition $\rho=\sum_i r_i
|r_i><r_i|$. Since $f$ is a convex function, by \cite{Wehrl} for
each $i$,
$$f(< r_i|A|r_i >) \leq \,\,< r_i|f(A)|r_i > .$$
By taking summation over all $i$, we have
\begin{eqnarray*}
f({\rm Tr}(\rho A))&=& f(\sum_i r_i <r_i|A|r_i>)\\
&\leq&
 \sum_i r_i f(< r_i|A|r_i>)\\
 &\leq&\
 \sum_i r_i < r_i|f(A)|r_i>\\
 &\leq&
  {\rm Tr}( \rho f(A)).
\end{eqnarray*}
We deduce that $f({\rm Tr}(\rho A)) \leq {\rm Tr}(\rho f(A))$.
\end{proof}
\begin{cor}\label{cor:convex}
$-\log x$ is a convex function for $x>0$, then $\lh(\rho)\leq
\oh(\rho)$ for each density matrix $\rho\in \mathbb{H}$.
\end{cor}
\begin{rem}
Both complexities can be infinite. Indeed, let $|u_n><u_n|$ be a
eigenvector of $\hat{\mu}$ in the spectral decomposition of it. Now,
$$\lh(|u_n><u_n|)=-\log {\rm Tr}(|u_n><u_n|\hat{\mu})=-\log<u_n|\hat{\mu}|u_n>=-\log r_n,$$
where $r_n$ is an eigenvalue that can be made as small as one likes. Since, ${\rm Tr}(\hat{\mu})\leq 1$ and hence $\sum_{n} r_n\leq 1$.  Therefore, by the Corollary
\ref{cor:convex} $\oh(|u_n><u_n|)$ can be also infinite.
\end{rem}
\begin{thm}\label{gk}
Let $|\ii>$, $\ii\in\Omega_2$, be a orthonormal basis for the Hilbert space $\mathbb{H}$. Then, we have
$$\oh(|\ii>)=\lh(|\ii>)\overset{+}{=}K(\ii),$$
where $K$ is the Kolmogorov complexity.
\end{thm}
\begin{proof}
Let's define the function $f(\ii)=<\ii|\hat{\mu}|\ii>$, which is semi-computable and $\sum_{\ii} f(\ii)\leq 1$. Therefore, by the universality of $\mu$, there exists a constant number $c>0$ such that $c f(\ii)\leq \mu(\ii)$. Thus,
$$-\log\mu(\ii)\leq -\log f(\ii)-\log c\Longrightarrow K(\ii)\overset{+}{\leq} \lh(\ii).$$
On the other hand, the semi-density matrix $\rho=\sum_{\ii} \mu|\ii><\ii|$ is semi-computable and hence $\rho\overset{*}{\leq} \hat{\mu}$. Therefore,
$$K(\ii)=<\ii|-\log\rho |  \ii>\, \overset{+}{\geq}\,<\ii|-\log\hat{\mu}  |  \ii>=\oh(\ii).$$
From Corollary \ref{cor:convex}, we have
$$K(\ii)\overset{+}{=}\lh(\ii) \overset{+}{=} \oh(\ii) . $$
\end{proof}
The next property is related to composite systems. Indeed, let $X$ and $Y$ be two physical systems and $\mathbb{H}_X$ and $\mathbb{H}_Y$
be their related Hilbert spaces. Then, $\mathbb{H}_{XY}:=\mathbb{H}_X\otimes \mathbb{H}_Y$  is the Hilbert space
system associated to $XY$. Let $\rho_{XY}$ be a density matrix on $\mathbb{H}_{XY}$. Then
$\rho_X={\rm Tr}_Y(\rho)$ and  $\rho_Y={\rm Tr}_X(\rho)$ are called
marginal density matrices for $\mathbb{H}_X$ and $\mathbb{H}_Y$, respectively. The
subadditivity relation \cite{Nielsen} tells us that
\begin{equation}
S(\rho_{XY})\leq S(\rho_{X}) + S(\rho_{Y})   .
\end{equation}
 The lower and upper
Gacs complexities have also subadditivity properties.

Let  $\overset{+}{<}$  denote inequality to within an additive constant, and
$\overset{*}{<}$ inequality to within a multiplicative
constant.
\begin{thm}\label{thm:semicomputable}
Let $XY$ be a composite system of   two subsystems $X$
and $Y$. Let  $\hat{\mu}_{XY}$, $\hat{\mu}_X$ and $\hat{\mu}_Y$ be associated
universal semi-density  matrices. Then,
\begin{equation}
     \hat{\mu}_X\otimes\hat{\mu}_Y\overset{*}{<} \hat{\mu}_{XY}.
    \end{equation}
    Moreover, for each $\rho\in \mathbb{H}_X$ and $\sigma\in\mathbb{H}_Y$,
    \begin{equation}
    \oh(\rho\otimes\sigma)\overset{+}{<}  \oh(\rho)+ \oh(\sigma),
    \end{equation}
    and
    \begin{equation}
     \lh(\rho\otimes\sigma)\overset{+}{<}  \lh(\rho)+ \lh(\sigma).
    \end{equation}
\end{thm}
\begin{proof}
It is clear that ${\rm Tr}(\hat{\mu}_X \otimes \hat{\mu}_Y)\leq 1$.
 Since, $\hat{\mu}_X$ and $\hat{\mu}_Y$ are universal semi-density matrices then there exist two increasing sequence of semi-computable semi-density matrices  converging to them, respectively. Therefore, the tensor product of  the two sequences\footnote{If $A_n$ and $B_n$ are two operators in $\mathbb{H}_X$ and $\mathbb{H}_Y$, then $A_n\otimes B_n$ is a sequence in $\mathbb{H}_{XY}$} is also an  increasing sequence\footnote{If $A, B, C, D$ are positive bounded operators with  $A\leq B$ and $C\leq D$ then $A\otimes C\leq B\otimes D$ \cite{stochel}.} converging   to $\hat{\mu}_X \otimes \hat{\mu}_Y$.
Thus, $\hat{\mu}_X \otimes \hat{\mu}_Y$ is a semi-computable semi-density matrix on $\mathbb{H}_{XY}$. By the universality of $\hat{\mu}_{XY}$, there exists a constant $c>0$ such that
$$c \hat{\mu}_X \otimes \hat{\mu}_Y\leq \hat{\mu}_{XY}.$$

The proof of the two next parts follows from the following equality:  $$\log \hat{\mu}_X \otimes \hat{\mu}_Y= \log\hat{\mu}_X \otimes \mathds{1} + \log\mathds{1}\otimes \hat{\mu}_Y.$$
\end{proof}
In the classical Kolmogorov complexity, we have monotonicity property
$K(x)\overset{+}{<} K(x, y)$ where $K(x,y):=K(<x,y>)$, for $x, y\in
\n$. This property is also true for the Gacs algorithmic complexities.
\begin{thm}
\begin{equation}
{\rm Tr}_Y
    \hat{\mu}_{XY}\overset{*}{=}\hat{\mu}_X,
    \end{equation}
    Moreover, for $\rho\in \mathbb{H}_X$ and  $\sigma\in \mathbb{H}_Y$,
    \begin{equation}
     \oh(\rho)\overset{+}{<}  \oh(\rho\otimes\sigma),
    \end{equation}
    and
    \begin{equation}
     \lh(\rho)\overset{+}{<}  \lh(\rho\otimes\sigma).
    \end{equation}
\end{thm}
\begin{proof}
Let us define $\rho_X={\rm Tr}_Y\hat{\mu}_{XY}$. It is clear that $\rho_X$ is a semi-density matrix. On the other hand, there exists a sequence of semi-computable semi-density matrices $\rho_{XY}^{(n)}$ such that $\rho_{XY}^{(n)}\nearrow\hat{\mu_{XY}}$. Thus, we have
 ${\rm Tr}_Y(\rho^{(n)}_{XY})\nearrow \rho_X$. Therefore, $\rho_X$ is a semi-computable semi-density matrix on $\mathbb{H}_X$. By the universality of $\hat{\mu}_X$, there exists a constant $c>0$ such that $c\rho_X\leq \hat{\mu}_X$.

Now, let's define the density matrix $\sigma_{XY}=\hat{\mu}_X\otimes |\psi><\psi|$, $|\psi>\in\mathbb{H}_Y$, $||\psi||=1$, where $|\psi><\psi|$ is a fixed semi-computable density matrix. Like the proof of Theorem \ref{thm:semicomputable},  $\sigma_{XY}$ is a semi-computable semi-density matrix. Therefore, there exists a constant $c'>0$ such that $c' \sigma_{XY}\leq \hat{\mu}_{XY}$. Then,
$$c' \hat{\mu}_X \leq {\rm Tr}_Y\hat{\mu}_{XY}=\rho_X.$$
Thus,
$${\rm Tr}_Y
    \hat{\mu}_{XY}\overset{*}{=}\hat{\mu}_X.$$

Now, let $|\psi><\psi|$ and $\rho $ be  density matrices, where $\sum_i r_i|\phi_i><\phi_i|$ is the spectral decomposition of $\rho$. Then, we have
\begin{eqnarray*}
 {\rm Tr}(\rho\otimes|\psi><\psi|\hat{\mu}_{XY} )&=&\sum_i r_i<\phi_i\psi|\hat{\mu}_{XY}|\phi_i\psi>\\
 & \leq&
 \sum_i r_i<\phi_i|{\rm Tr}_Y\hat{\mu}_{XY}|\phi_i>\\
 &\overset{+}{=}&{\rm Tr}(\rho\hat{\mu}_{X}).
\end{eqnarray*}
Finally, let $\sigma$ be a density matrix on $\mathbb{H}_Y$ with the spectral decomposition $\sum_j s_j |\psi_j><\psi_j|$. Then, we have
\begin{eqnarray*}
{\rm Tr}(\rho\otimes \sigma \hat{\mu}_{XY})&=&\sum_j s_j{\rm Tr}(\rho\otimes |\psi_j><\psi_j|)\\
&\overset{+}{\leq}&
\sum_j {\rm Tr}(\rho \hat{\mu}_X)\leq {\rm Tr}(\rho \hat{\mu}_X)
\end{eqnarray*}
\end{proof}

It is important to know wether the evolution of a quantum dynamical system has
effects on the Gacs complexities or not. In the following theorem we
show that when time evolution is an elementary unitary operator
then modulo a constant number, the Gacs complexities is invariant.

\begin{thm}\label{thm:semi-unitary}
Let $U$ be  any elementary unitary operator. Then, for any
semi-density matrix $\rho\in\mathbb{H}$,
\begin{equation}
    \oh(U\rho U^\dag)\overset{+}{=}\oh(\rho),\quad
    \lh(U\rho U^\dag)\overset{+}{=}\lh(\rho).
    \end{equation}
\end{thm}
\begin{proof}

Since $U$ is an elementary unitary operator, then $U\hat{\mu}
U^\dag$ and $U^\dag \hat{\mu} U$ are semi-computable semi-density
matrices and hence there are constants $c_{U\hat{\mu} U^\dag}$ and
$c_{ U^\dag\hat{\mu} U}>0$ such that
$$c_{U\hat{\mu} U^\dag} U\hat{\mu} U^\dag\leq \hat{\mu}, \quad  c_{ U^\dag\hat{\mu} U} U^\dag\hat{\mu} U\leq \hat{\mu}. $$
From the second one, we have
$$  c_{ U^\dag\hat{\mu} U}\hat{\mu}\leq U \hat{\mu}U^\dag.$$
Therefore, $U\hat{\mu} U^\dag$ is also a universal semi-measure and
thus the result follows.
\end{proof}

\begin{thm}
Let $P\neq 0$ be a lower semi-computable projection with $d={\rm
    Tr}P<\infty$. then,
    \begin{equation}
    \lh(\rho)\overset{+}{<} \log d-\log({\rm Tr}P),
    \end{equation}
\end{thm}
\begin{proof}
Let $\rho$ be the semi-computable semi-density matrix $P/d$. Then,
there exists a constant $c_\rho>0$ such that $c_\rho P/d\leq \hat{\mu}$.
$$\lh(\rho)=-\log{\rm Tr}(\rho \hat{\mu})\overset{+}{<} -\log{\rm Tr}(\rho (P/d))\overset{+}{=}\log{d} -\log{\rm Tr}(\rho P).$$
\end{proof}

Let us consider  the spectral decomposition of $\hat{\mu}=\sum_i
u_i|u_i><u_i|$ where $u_1\geq u_2\geq \ldots$. Let $E_k=\sum_{i=1}^k
|u_i><u_i|$ be a projection on $\mathbb{H}$. The following theorem gives
a lower bound of the  Gacs algorithmic complexities.
\begin{thm}
(Lower bounds). Let $\rho$ be a semi-density matrix and let
$\lambda>1$. If
    $\oh(\rho)<k$, then
    $${\rm Tr}(\rho E_{2^{\lambda k}})>1-1/\lambda.$$
    In addition, if $\lh(\rho)<k$ then
    $${\rm Tr}(\rho E_{2^{\lambda k}})>2^{-k}(1-1/\lambda),$$
    where $E_{2^{\lambda k}}:=E_{\lfloor2^{\lambda k}\rfloor}$.
\end{thm}

\begin{proof}
Let us consider  the spectral decomposition of $\hat{\mu}=\sum_i
u_i|u_i><u_i|$ where $u_1\geq u_2\geq \ldots$\,\,. By the assumption
$\oh(\rho)<k$. Therefore, we have
$$\sum_i -\log{u_i} <u_i|\rho|u_i> < k.$$
Now, let $m$ be the first $i$ with $u_i\leq 2^{-\lambda k}$. Since
$\sum_i u_i\leq 1$, $m\leq 2^{\lambda k}$. In addition,
$$\lambda k \sum_{i\geq m}<u_i|\rho|u_i>\,\, \leq  \sum_{i\geq m} -\log{u_i}   <u_i|\rho|u_i> \leq\lh(\rho \hat{\mu})< k.$$
Therefore, $ \sum_{i\geq m}<u_i|\rho|u_i> < 1/\lambda$.

By the assumption  $\lh(\rho)\leq k$, we have
$$-\log\sum_i u_i<u_i|\rho|u_i>\leq k\Rightarrow \sum_i u_i<u_i|\rho|u_i>\geq 2^{-k}.$$
Let $m$ be the first $i$ with $u_i<2^{-k}/\lambda$. Since, $\sum_i u_i\leq 1$ we have
$m\leq 2^k \lambda$. Therefore,
$$\sum_{i\geq m} u_i<u_i|\rho|u_i> \leq \frac{2^{-k}}{\lambda} \sum_{i\geq m} <u_i|\rho|u_i> \leq \frac{2^{-k}}{\lambda}.$$
Now,
$${\rm Tr}(\rho E_{m})=\sum_{i\leq m} <u_i|\rho|u_i>\geq \sum_{i\leq m} u_i <u_i|\rho|u_i>\geq 2^{-k}-\sum_{i\geq m} u_i <u_i|\rho|u_i>\geq 2^{-k}(1-\frac{1}{\lambda}).$$

\end{proof}
\section{Applications of Upper Gacs Complexity}\label{Applications of Upper Gacs Complexity}
\begin{defn}
\label{defqcomp} Since Theorem \ref{thm:universal} establishes the
existence of a universal semi-density matrix for an infinite
dimensional quantum spin chain, we take ~(\ref{qcompop}) with
$\hat{\mu}$ as in (~\ref{chain-univ}) as the complexity operator of a
quantum spin chain and (~\ref{qalgentr1}) as the Gacs entropy of any
density matrix $\rho$ associated with the chain.
\end{defn}

Notice that the complexity operator of the quantum chain assigns the
following Gacs entropy to a local density matrix $\rho_{[-n,n]}$ on
$\h_{[-n,n]}$:
\begin{equation}
\label{localcompl} \oh(\rho_{[-n,n]})={\rm Tr}(\rho_{[-n,n]}\,P_n
\kappa\, P_n)\ ,
\end{equation}
where $P_n$ projects  the Hilbert space $\h$ on which
$\hat{\mu}$ acts onto the finite dimensional Hilbert space
$\h_{[-n,n]}$ on which $\rho_{[-n,n]}$ acts.

On the other hand, one could consider the restriction
$\hat{\mu}^{(n)}=P_n\,\hat{\mu}\,P_n$ of the universal density
matrix $\hat{\mu}$ to $\h_{[-n,n]}$; the natural guess is that
$\hat{\mu}^{(n)}$ might indeed be a universal semi-computable
semi-density matrix for the local spin algebra $M_{[-n,n]}$.

 That is indeed so is proved in the next Lemma. Then,
given a local spin algebra $M_{[-n,n]}$, we obtain the original
finite dimensional formulation of~\cite{Gacs}. Indeed, given
$\hat{\mu}^{(n)}=P_n\,\hat{\mu}\,P_n$ its complexity operator will
be
\begin{equation}
\label{localcomplfin} \kappa^{(n)}=-\log\hat{\mu}^{(n)} \ ,
\end{equation}
and, given a state $\rho^{(n)}=\rho_{[-n,n]}$ on $\h_{[-n,n]}$, its
Gacs algorithmic entropy will be
\begin{equation}
\label{Gacsentropyqc} \oh^{(n)}(\rho^{(n)})=-{\rm
Tr}(\rho^{(n)}\,\log\hat{\mu}^{(n)})\ ,
\end{equation}
where the trace is computed on $\h_{[-n,n]}$.

\begin{lem}\label{thm:l}
Let $T$ be a universal semi-computable semi-density matrix  which is
the limit of a computable quasi-increasing sequence of elementary
semi-density matrices $T_n$. Then, for each $k, P_k T P_k$ is a
universal semi-computable semi-density matrix on $\h_{[-k,k]}$.
\end{lem}

\begin{proof}
Clearly, the sequence $P_k T_n P_k$, $n\geq k$, is a computable
quasi-increasing sequence of elementary semi-density matrices;
moreover,
$$
\lim_{n\rightarrow\infty} P_k T_n P_k = P_k T P_k\ .
$$
Since $T$ is a universal semi-computable semi-density matrix, for
each semi-computable semi-density matrix $R_k$ on $\h_{[-k,k]}$,
 there exists a positive constant $C_k$
such that $$T- C_k R_k \geq 0 \longrightarrow  P_k T P_k - C_k R_k\geq
0$$.
\end{proof}
\medskip
 Based on the infinite dimensional formulation of the
complexity operator, we can now study the Gacs algorithmic
complexity per site of translation invariant states of quantum spin
chains and relate it to their von Neumann entropy rate and
AF-entropy.

\begin{thm}
\label{thm:ex} Let $\rho^{(n)}\in B^+_1(\h_{[-n,n]})$ be a
computable sequence of semi-computable density matrices giving rise
to a shift-invariant state $\omega$ on the quantum spin chain
$\mathcal{M}$. Then
\begin{equation}
\label{Gacs1}
\lim_{n\rightarrow\infty}\frac{\oh^{(n)}(\rho^{(n)})}{2n+1} =
\lim_{n\rightarrow\infty}\frac{\oh(\rho^{(n)})}{2n+1}=s(\omega)\ ,
\end{equation}
where $s(\omega)$ is the von Neumann entropy rate
in (~\ref{vNentrate}). Also, with reference to the Alicki-Fannes
entropy and the density matrices $R[U^{(n)}]$ on the doubled local
sub-algebras $M_{[-n,n]}\otimes M_{[-n,n]}$ in~(\ref{marg2}), it
holds that
\begin{equation}
\label{Gacs2} \lim_{n\rightarrow\infty}
\frac{\oh^{(n)}(R[U^{(n)}])}{2n+1} = \lim_{n\rightarrow\infty}
\frac{\oh(R[U^{(n)}])}{2n+1} =s(\omega)+1\ .
\end{equation}
\end{thm}

\begin{proof}
By normalizing $\hat{\mu}^{(n)}$ with ${\rm Tr}(\hat{\mu}^{(n)})\leq
1$ and using that for any two density matrices $\rho_{1,2}$,
$\rho_2$ invertible, $\displaystyle{\rm
Tr}(\rho_1(\log\rho_1-\log\rho_2))\geq 0$, \cite{ohya}. one
estimates

$$
S(\rho^{(n)}) \leq -{\rm
Tr}\left(\rho^{(n)}(\log\hat{\mu}^{(n)}-\log{\rm
Tr}(\hat{\mu}^{(n)}))\right)\leq \overline{H}^{(n)}(\rho^{(n)}).$$
Analogously, $S(\rho^{(n)}) \leq \overline{H}(\rho^{(n)})$. Observe
that $\hat{\mu}$ on $\mathbb{H}$ and $\hat{\mu}^{(n)}$ on
$\mathbb{H}_{[-n, n]}$ for each $n$ are invertible.
 %%Then, one rewrites
%5$$
%%\overline{H}^{(n)}(\rho^{(n)})= -{\rm
%%Tr}\Big(\rho^{(n)}\Big(\log\mu^{(n)}-P_n\,(\log\mu)\,P_n\Big)\Big)+\oh(\rho^{(n)})\
%%.$$

Let $\rho=\sum_{n\geq 2} \rho^{(n)}/n(\log{n})^2$. Then, $\rho$ is a
semi-computable semi-density matrix. So, there exists $p \in N$ such
that $\rho\leq 2^p\hat{\mu}$. Because of the operator monotonicity
of the logarithm, one estimates
\begin{eqnarray}
S(\rho^{(n)})\leq \oh(\rho^{(n)})&=&-{\rm
Tr}\Big(\rho^{(n)}\,\log\hat{\mu}\, \Big) \leq p-{\rm
Tr}(\rho^{(n)}\log\rho\,) \nonumber\\
&\leq & S(\rho^{(n)}) + p + \log{ n} + 2 \log\log n .\nonumber
\end{eqnarray}
Since $p$ is independent of $n$, then clearly we have
$$\lim_{n\rightarrow\infty}\frac{\oh(\rho^{(n)})}{2n+1}=s(\omega).$$
On the other hand, $\rho^{(n)}\leq 2^p n(\log n)^2 \hat{\mu}$ and
hence
%$$\rho^{(n)}=\leq 2^p n(\log n)^2 P_n\, \mu\, P_n=2^p n(\log n)^2 \,\mu^{(n)}.$$
%Therefore, we have also
\begin{eqnarray}
S(\rho^{(n)})&\leq& \oh^{(n)}(\rho^{(n)})=-{\rm Tr}_{[-n,
n]}\Big(\rho^{(n)}\,\log\hat{\mu}^{(n)}\, \Big)\nonumber \\&=& -{\rm
Tr}_{[-n, n]}\Big(\rho^{(n)}\,\log P_n\,\hat{\mu} \, P_n\, \Big)\nonumber\\
&\leq&
  -{\rm
Tr}_{[-n, n]}\Big(\rho^{(n)}\,\log P_n\,\rho^{(n)}\, P_n\, \Big) + p
+ \log{ n} + 2 \log\log n \nonumber\\ &\leq&
 S(\rho^{(n)}) + p + \log{ n} + 2 \log\log n ,
\end{eqnarray}
 where $p$ is
independent of $n$, then
$$\lim_{n\rightarrow\infty}\frac{\oh^{(n)}(\rho^{(n)})}{2n+1}
=s(\omega).$$
%~(\ref{Gacs1}) follows from $$ \lim_{n\rightarrow\infty}\frac{{\rm
%Tr}\Big(\rho^{(n)}\Big(\log\mu^{(n)}-P_n\,(\log\mu)\,P_n\Big)\Big)}{2n+1}=0\
%,$$ which in turn is a consequence of the inequality
%$$ \Big|{\rm Tr}\Big(\rho^{(n)}\Big(\log\mu^{(n)}-P_n\,(\log\mu)\,P_n\Big)\Big)\Big|\leq \|\log\mu^{(n)}-P_n\,(\log\mu)\,P_n\|,$$
%and of the fact that,  according to  Lemma~\ref{thm:l}, the right
%hand side goes to zero when $n\to+\infty$.

The relations in~(\ref{Gacs2}) can be proved in the same way, once
one extends the construction of a universal semi-computable
semi-density matrix to the case of the $C^*$-algebra arising from
the inductive limit of the nested net of double local sub-algebras
$M_{[-n,n]}\otimes M_{[-n,n]}$. This can be done by means of the map
in~(\ref{inversemap}).
\end{proof}
\medskip

In~\cite{Benatti}, both the above relations have been proved under
the condition that the Kolmogorov complexity rates
\begin{equation}
\label{condition}
\lim_{n\to\infty}\frac{\kappa(\rho^{(n)})}{2n+1}=0=\lim_{n\to\infty}\frac{\kappa(R[U^{(n)}])}{2n+1}\
.
\end{equation}

This restriction is not necessary; indeed,  by constructing, as done
before, an infinite dimensional  universal semi-computable
semi-density matrix, one can control all universal semi-computable
semi-density matrices of the local sub-algebras of the quantum
chains, independently of $n$.

The following example indeed shows an instance of quantum spin chain
which does not satisfy the conditions \eqref{condition} and
nevertheless fulfils the conclusions of Theorem \ref{thm:ex}.

\begin{example}
Let $P_0 $ and $P_1$ be two orthogonal projections in $M_2(\com)$
and let $P_{\ii}=\bigotimes_{j=0}^{n-1}P_{i_j}$ denote the
orthogonal projections obtained by tensor products. Let the starting
one site density matrix be $\rho_{\{0\}}=\frac{P_0 + P_1}{2}$ and
assume that $\rho^{(n)}=\rho_{[0,n-1]}$ be defined such that its
complexity $K(\rho^{(n)})\geq n^2$. We now recursively construct
$\rho^{(n+1)}$ so that, on one hand the family of density matrices
satisfies the compatibility and translation invariant conditions
\eqref{compatibility} and \eqref{translationinvariance}, whence
$$
\lim_{n\to\infty}\frac{S(\rho^{(n)})}{n}=s(\omega)<+\infty\ ,
$$
and, on the other hand, so that $K(\rho^{(n+1)})\geq (n+1)^2$,
whence
$$
\lim_{n\to\infty}\frac{K(\rho^{(n)})}{n}=+\infty\ .
$$
Write $\rho^{(n)}=\sum_{\ii} a_{\ii}\, P_{\ii}$. Then, the
conditions \eqref{compatibility} and \eqref{translationinvariance}
yield
$$
{\rm Tr}_{\{0\}}\rho^{(n+1)}={\rm Tr}_{\{n+1
\}}\rho^{(n+1)}=\rho^{(n)}\ ,
$$
whence
$$
\sum_{\ion}(a_{0\ii}+a_{1\ii}) P_{\ii}= \sum_{\ion}(a_{\ii 0}+a_{\ii
1}) P_{\ii}=\sum_{\ion} a_{\ii}\, P_{\ii}\ .
$$
Then, because of the orthogonality of the projections $P_{\ii}$,  it
follows that
 \begin{eqnarray*}
 a_{0\bi^{(n-2)}0}+ a_{0\bi^{(n-2)}1}&=& a_{0\bi^{(n-2)}}\\
 a_{0\bi^{(n-2)}1}+ a_{1\bi^{(n-2)}1}&=& a_{\bi^{(n-2)}1}\\
 a_{1\bi^{(n-2)}1}+ a_{1\bi^{(n-2)}0}&=& a_{1\bi^{(n-2)}}\\
  a_{1\bi^{(n-2)}0}+ a_{0\bi^{(n-2)}0}&=& a_{\bi^{(n-2)}0}\ ,
\end{eqnarray*}
for any of the $2^{n-2}$ strings $\bi^{(n-2)}\in\Omega_2^{n-2}$. In
this way, the system of $2^n$ equations can be subdivided into
$2^{n-2}$ sub-systems of $4$ equations each. Let us focus upon the
system above defined by the string $\bi^{(n-2)}$; the values at the
right hand side have been chosen at step $n-1$. They are positive,
with all the others they sum up to $1$. Without loss of generality,
we may assume they are in decreasing order: $a_{0\bi^{(n-2)}}\geq
a_{\bi^{(n-2)}1}\geq a_{1\bi^{(n-2)}}\geq a_{\bi^{(n-2)}0}>0$. We
can now choose $a_{1\bi^{(n-2)}1}=x_{\bi^{(n-2)}}$ , a positive real
number such that $ x_{\bi^{(n-2)}}\leq a_{1\bi^{(n-2)}}$ with
Kolmogorov complexity $K(x_{\bi^{(n-2)}})\geq n^2$. Then,
$$
a_{1\bi^{(n-2)}0}=a_{1\bi^{(n-2)}}-x_{\bi^{(n-2)}}\ ,\
a_{0\bi^{(n-2)}1}=a_{\bi^{(n-2)}1}-x_{\bi^{(n-2)}}\ ,\
a_{0\bi^{(n-2)}0}=a_{0\bi^{(n-2)}}-a_{\bi^{(n-2)}1}+x_{\bi^{(n-2)}}\
.
$$
Therefore, the coefficients at step $n$ are positive, the sum of all
of them is $1$ and they satisfy the desired condition on the
increase of the algorithmic complexity of $\rho^{(n)}$.
\end{example}

\chapter{The Classical Gacs Algorithmic Complexity}

\goodbreak

%%% ----------------------------------------------------------------------

In this chapter we apply Gacs complexity
to classical dynamical systems. Here,  we  assume
that the probability measure of the  symbolic dynamical system associated with a
given dynamical system and the considered finite measurable partition are semi-computable. Of course,  the
probability measure condition   forces
semi-computable probability measures to be  computable. We will also prove a
version of the Brudno's theorem based on a given universal semi-measure.

\smallskip

%%% ----------------------------------------------------------------------
\goodbreak

\section{Gacs algorithmic complexity in classical dynamical systems}

\begin{defn}
Let $(\chi, T, \nu)$ be a dynamical system. let $\mathcal{P}$ be a
finite measurable partition of $\chi$. The associated symbolic dynamical system
$(\Omega_p, T_\sigma, \nu_\mathcal{P})$ (see section \ref{classical dyanamical}) is called a semi-computable
symbolic dynamical system  if $\nu_\mathcal{P}$ as a function of
$\Omega_2$ into $\rr$ is a semi-computable probability measure.
\end{defn}

 Notice  that since $\sum_{\ion} \nu_\mathcal{P}(\ii)=1$, the semi-computable $\nu_\mathcal{P}$ is computable. Hence, we can always take
$\nu_\mathcal{P}$  computable.
\begin{rem}
We mention that $\nu_\mathcal{P}$ is not a measure from $\Omega_p$
into $\rr$. Because,
$$\sum_{\ii\in\Omega_p}\nu_\mathcal{P}(\ii)=\sum_{n=1}^\infty
1=\infty.$$
\end{rem}
The definition  of classical Gacs algorithmic complexity mimics the construction of the
KS entropy.  Indeed, we define the Gacs algorithmic complexity
for a given semi-computable symbolic dynamical system $(\Omega_p,
T_\sigma, \nu_\mathcal{P})$  as follows
\begin{equation}
G(T, \mathcal{P}^{(n)})=-\sum_{\ion}
\nu_\mathcal{P}(\ii)\log{\mu(\textbf{i}^{(n)})},
\end{equation}
where   $\mu$ is a universal semi-computable semi-measure  on
$\Omega_p$. We can interpret this definition as giving the
information content of the semi-computable probability measure contained in a
universal semi-measure. On the other hand, $\mu(\textbf{i}^{(n)})$ $>0$, so that $G(T,
\mathcal{P}^{(n)})$ is a finite quantity.

The rate of Gacs algorithmic complexity  is naturally defined as
\begin{equation}\label{eq:classical gacs}
 G(T, \mathcal{P})=\limsup_{n\rightarrow\infty}\frac{1}{n}G(T, \mathcal{P}^{(n)}).
\end{equation}

\begin{defn}
Let $(\chi, T, \nu)$ be a  dynamical system. The rate of Gacs algorithmic complexity
is
\begin{equation}
G(T):= \sup_{\nu_\mathcal{P}} G(T, \mathcal{P}),
\end{equation}
where $\mathcal{P}$ is a finite measurable partition such that $\nu_{\mathcal{P}}$ is  computable.
\end{defn}
\begin{rem}
 In general the $\sup$  in the above definition is computed over all finite measurable partitions.
However, in order to use semi-universal semi-computable measures,  we restrict ourselves to computable finite
measurable partitions cases.
\end{rem}
Now, we encounter the following natural question: Is there any relation between the Gacs algorithmic complexity and KS entropy in ergodic
classical dynamical systems? We are going to provide the answer.

\begin{thm}\label{thm:class}
Let $(\chi, T, \nu)$ be a semi-computable dynamical system, then
\begin{equation}
 G_\nu(T)\leq h_\nu^{KS}(T).
 \end{equation}
\end{thm}
\begin{proof}
Let $\mathcal{P}$ be  a  finite measurable  partition  of $\chi$ such that $\nu_\mathcal{P}$
 is a computable.
Since $\nu_\mathcal{P}$ cannot be a measure on $\Omega_p$, we consider a semi-computable semi-measure  $f$ on
$\Omega_p$, defined as follows
$$f(\ii)=\frac{1}{\sum_{n=1}^\infty \delta(n)}\delta(n)\nu_\mathcal{P}(\ii), \quad \ii\in \Omega_p,$$
where $\delta(n)=\frac{1}{n \log^2{n}}$.
 Then,
%Let $\nu_\mathcal{P}$ be $m^{\rm th}$ semi-computable semi-measure $\mu$ be a universal semi-measure from $\Omega_2$ to $\rr$.
there exists a constant $c_{\nu_\mathcal{P}}>0$, dependening  on
$\nu_\mathcal{P}$, such that for any $\ii\in\Omega_p$,
$$c_{\nu_\mathcal{P}}\delta(n)\nu_{\mathcal{P}}(\ii)\leq\mu(\ii).$$
Thus,
$$-\sum_{\ion} \nu_\mathcal{P}(\ii)\log{\mu(\textbf{i}^{(n)})}\leq
-\sum_{\ion}\nu_\mathcal{P}(\ii)\log\nu_\mathcal{P}(\ii)-\log
c_{\mathcal{P}}-\log\delta(n).$$
%%-\log{c_{\nu_\mathcal{P}}}}.$$
Therefore,
$$G(T,\mathcal{P})\leq h_\nu^{KS}(T, \mathcal{P})\leq h_\nu^{KS}(T).$$
Now, we take the $\sup$ over all computable
$\nu_\mathcal{P}$. Then,
$$G(T)\leq h_\nu^{KS}(T).$$
\end{proof}

\begin{thm}\label{GB}
Let $(\Omega_p, T_\sigma, \nu)$ be a binary ergodic dynamical system
where $\nu$ is  computable. Then,
\begin{equation}
G(T)=h_\nu^{KS}(T_\sigma),\quad   \nu-a.e.
\end{equation}
\end{thm}
\begin{proof}
%By [vitany 253], there are constants $c_{1\mathcal{P}}$ and $c_{2\mathcal{P}}$ such that %$c_{1\mathcal{P}}\leq -\log{\mu(\ii)-K(\ii ,\mathcal{P})}\leq c_{2\mathcal{P}}$, which
%this equation is dependent to $\mathcal{P}$ and $\mathcal{P}$.
Let $\mathcal{P}$ be a finite measurable partition of $\Omega_p$.
Let $\nu^{(n)}_\mathcal{P}$ be its related probability measure where
$\nu^{(n)}_\mathcal{P}(\ii)=\nu(C^{[0, n-1]}_{i_0, i_1, \ldots,
i_{n-1}})$. It is clear that $\sum_{\ii\in
\Omega_p^{(n)}}\nu^{(n)}_\mathcal{P}(\ii)=1$. Notice that, $\nu_\mathcal{P}$ is $\nu^{(n)}_\mathcal{P}$ for each $n\in\n$.
 By Theorem \ref{thm:levin} and Inequality (~\ref{re:kol}) for all $x\in\n$. We have,
\begin{equation}\label{re:uni}
-\log{c_1}+ C(x)\leq  -\log\mu(x)\leq K(x)+\log{c}\leq C(\ii)+2
\log{n} + \log{c_2},
\end{equation}
where $c_1>0$ and $c_2>0$ are constant numbers. Since we can
represent each finite length binary string $\bi\in\Omega_p$ by an
integer number,  by applying the theorem \ref{re:uni}, we obtain
\begin{equation}
G(T,\mathcal{P}^{(n)})=\limsup_{n\to\infty} \frac{1}{n}\sum_{\ion}
\nu_\mathcal{P}(\ii) C(\ii).
\end{equation}
By Brudno's theorem \ref{Brudno} for  $\epsilon>0$ there is an
integer number $N$ such that for any $\n\ni n\geq N$,
$$\frac{1}{n}C(\ii)\geq h^{KS}_\nu(T_\sigma)-\epsilon.$$
Therefore, by Theorem \ref{thm:class},
\begin{eqnarray*}
h_\nu^{KS}(T_\sigma)\geq G(T)\geq G(T, \mathcal{P})
 &=&
\limsup_{n\to\infty} \frac{1}{n}\sum_{\ion} \nu_\mathcal{P}(\ii)
C(\ii)
 \\&\geq&
\limsup_{n\to\infty} \frac{1}{n}\sum_{\ion}
\nu(\ii)(h^{KS}_\nu(T_\sigma)-\epsilon)
\\&\geq&
h^{KS}_\nu(T_\sigma)-\epsilon .
\end{eqnarray*}
 Thus,
$$h_\nu^{KS}(T_\sigma)= G(T)=k(\bi)=c(\bi)\quad \nu-a.e ,$$
where $k$ and $c$ are rate of the prefix Kolmogorov and Kolmogorov complexities, respectively, which are defined as follows: $$k(\bi)=\lim_{n\to\infty} \frac{K(\ii)}{n},\quad c(\bi)=\lim_{n\to\infty} \frac{C(\ii)}{n} $$
\end{proof}
Now, the question is:
Can we give a short proof for the classical Brudno
theorem in   ergodic semi-computable cases? In the
following theorem, we will give   a short proof for a ergodic source dynamical systems.
\begin{thm}
Let $(\Omega_2, T_\sigma, \nu)$ be a semi-computable binary ergodic
source with KS entropy rate $h_\nu^{KS}(T_\sigma)$. Then,
$$\lim_{n\to\infty} -\frac{\log\mu(\ii)}{n}= h_\nu^{KS}(T_\sigma),\quad \nu-a.e,$$
for almost all $\bi\in\Omega_p$ with respect to $\nu$.
\end{thm}
\begin{proof}\hfill

\noindent \textbf{Part 1:} Let us consider the function $f$ from $\Omega_2$
into $\rr$ as follows,
$$f(\ii)=\frac{1}{\sum_{n=1}^\infty n^{-2}}\frac{1}{n^2}\nu(\ii), $$
where $\quad\sum_{n=1}^\infty n^{-2}=\pi^2/6$.
 It is straight word  to
check that  the function $f$ is a measure.
 Since the probability measure $\nu$ is a computable  measure and hence  $f$ is also a computable probability distribution.
 Then,
 by universality of the semi-measure $\mu$ there exists  constant number $c>0$
such that $$\frac{6c}{\nu^2} \frac{1}{n^2}\nu(\ii)\leq c f(\ii)\leq
\mu(\ii).$$ Therefore,
$$\limsup_{n\to\infty} -\frac{\log\mu(\ii)}{n}\leq
\limsup_{n\to\infty} -\frac{\log\nu(\ii)}{n}\leq
h_\nu^{KS}(T_\sigma),\quad \nu-a.e,$$ where  we used the
Shannon-Mc Millan-Breiman theorem \cite{billingsley} for the second
inequality.

\noindent\textbf{part 2:} The proof of inverse inequality is exactly like  the Brudno's theorem.

From the Asymptotic Equipartition Property (AEP) and Shannon-Mc
Millan-Breiman theorem \cite{billingsley}, we know
that for the set $A_\epsilon^{(n)}=\{\ii\in\Omega_2^{(n)}|
2^{-n(h_\nu^{KS}(T_\sigma)+\epsilon)}\leq \nu(\ii)\leq
2^{-n(h_\nu^{KS}(T_\sigma)-\epsilon)}\}$,
$$Prob(A_\epsilon^{(n)})\approx 1 \quad\text{and}\quad (1-\epsilon)
2^{n(h_\nu^{KS}(T_\sigma)-\epsilon)}<\#(A_\epsilon^{(n)})<2^{-n(h_\nu^{KS}(T_\sigma)+\epsilon)}.$$
By Theorem  \ref{thm:levin2} and Inequality (~\ref{re:compres})
we have
$$\#\{\ii: \mu(\ii)\geq 2^{-c'+\log\delta(n)+\log{c}}\}\leq
2^{c'}-1.$$ Therefore,
\begin{equation}\label{compreseeing-semimeasure}
\#\{\ii: \mu(\ii)\geq 2^{-c'}\}\leq 2^{c'+\alpha}-1,
\end{equation}
where $\alpha=-\log\delta(n)-\log{c'}
>0$.
we define the subset of $A_\epsilon^{(n)}\subseteq \Omega_2$  as
follows
\begin{equation}\label{AEPT}
\hat{A}_\epsilon^{(n)}=\{\ii \in A_\epsilon^{(n)} | \mu(\ii)\geq
2^{-n(H_\nu-2\epsilon)}\}.
\end{equation}
This means that each element $\ii\in\hat{A}_\epsilon^{(n)}$ is the initial
prefix of length $n$, of  some strings in $\Omega_2$.
 Then,
\begin{align}
 \nu(\hat{A}_\epsilon^{(n)})&=\nu(\{\ii |
\mu(\ii)\geq 2^{-n(H_\nu-2\epsilon)}, \ii\in A^{(n)}_\epsilon\}) \nonumber\\
 &\leq  \#(\hat{A}_\epsilon^{(n)})
\cdot \max_{\ion} \nu(\ii) \nonumber\\
 & \leq 2^{n(H_\nu-2\epsilon)+\alpha+1}\cdot
2^{-n(H_\nu-\epsilon)}=2^{-n\epsilon+\alpha+1}. \nonumber
\end{align}
We know that there is some strings $\ii\notin A_\epsilon^{(n)}$
such that $\mu(\ii)\geq 2^{-n(h_\nu-2\epsilon)}$, so let
$$\tilde{A}_\epsilon^{(k)}=\{\ii\, |\,\, \mu(\bi^{(k)})\geq
2^{-k(H_\nu-2\epsilon)},\quad \ii\in (\hat{A}_\epsilon^{(k)})^c\}.$$
where $(\hat{A}_\epsilon^{(k)})^c=\Omega_2 \backslash
\hat{A}_\epsilon^{(k)}$.

 Let $B_\epsilon^{(n)}=\bigcup_{k\geq
n}\tilde{A}_\epsilon^{(k)}$ then $\nu(B_\epsilon^{(n)})\leq
\nu(\bigcup_{k\geq n} \hat{A}_\epsilon^{(k)})^c=1-\nu(\bigcap_{k\geq
n}\hat{A}_\epsilon^{(k)})$. Therefore,
\begin{multline}
\nu(\bigcup_{k\geq n} \{\hat{A}_\epsilon^{(k)}\bigcup
\tilde{A}_\epsilon^{(k)}\})
 \leq
\nu(\bigcup_{k\geq n}\tilde{A}_\epsilon^{(k)})+\nu(B_\epsilon^{(k)})
\\
 \leq \sum_{k\geq
n} 2^{-k\epsilon+\alpha_\nu+1}+\nu(B_\epsilon^{(k)})\leq
\frac{2^{-k\epsilon+\alpha_\nu+1}}{1-2^{-\epsilon}}
+1-\nu(\bigcap_{k\geq n}\hat{A}_\epsilon^{(k)})
\end{multline}

 Now, let $i_1, i_2, \ldots \in\Omega_2$ be a
binary sequence whose initial prefixes are typical for $k\geq n$, namely $i_1,
i_2,\ldots, i_k\in A^{(k)}_\epsilon$ . Then $\bi\in \bigcap_{k\geq
n}\hat{A}_\epsilon^{(n)}$.

  It is clear that
$\lim_{n\to\infty}\nu(\bigcap_{k\geq n}\hat{A}_\epsilon^{(n)})=1$.
Therefore,
$$\liminf_{n\to\infty} -\frac{\log\mu(\ii)}{n}\geq
h_\nu^{KS}(T_\sigma)-\epsilon,\quad \nu-a.e,$$

\end{proof}

%%% ----------------------------------------------------------------------

%\def\baselinestretch{1}

\chapter{Brudno's Theorem in Quantum Spin Chains with Shift Dynamics}\label{ch:quasi-free}

%%% ----------------------------------------------------------------------

%%% ----------------------------------------------------------------------
\goodbreak

%%% ----------------------------------------------------------------------
In this chapter we investigate the extensions of the classical Brudno's theorem to quantum spin chains with right-shift dynamics  using the quantum Shannon-MacMillan theorem.

\section{Extension of Brudno's Theorem}
In the classical case systems, Brudno proved  a relation between ergodicity theory and Kolmogorov complexity \cite{Brudno}. It is natural to ask ourselves that what is the extension of this theorem in quantum dynamical systems?
To extend this theorem, we should extend the meaning of Kolmogorov Complexity and $KS$-entropy from  the classical dynamical systems to the quantum cases.

In this thesis, we  focus on the Gacs extension of Kolmogorov complexity and on AF and CNT extensions  of the $KS$-entropy. Now, what about is the generalization of  Brudno's theorem?  The first step is to define the concept of trajectory in quantum systems. Unfortunately, the definition of trajectory as defined in symbolic dynamical system using partitions in quantum systems is not easy. Therefore, we will proceed without using trajectories.
 Our mehod is independent of the partition of unity  used  in the definition of the $AF$-entropy.

Our method is used   the notion of semi-computability  which is described in chapter \ref{ch:semi-compuatble density matrises}. Since, the space of the Fermionic algebras using the Jordan-Wigner transformation is infinite tensor product of $d$-level matrices and  semi-computability  is defined on infinite dimensional Hilbert spaces, then,  it will be an appropriate  method to investigate the dynamics of Fermionic particles.  In the Bosonic case, the semi-computability concept should be extended to $C^*$-algebras which is another problem and we don't consider in this thesis.

Therefore, we proceed to extend the Brudno's theorem based on semi-computability concept in quantum spin chains with shift dynamics. The quantum Shannon-MacMillan theorem for translation invariant
ergodic quantum spin systems on $\zz$ lattice is formulated in
\cite{Bjelakovictheshannon-mcmillan}. Now, we want to investigate a
version of the Brudno theorem using the quantum Shannon-MacMillan
theorem. Here, we use projections  instead of "almost every"  in  Brudno theorem.

 Before going further, we give a definition of quantum ergodic theory which is based on the algebraic formalism.
\begin{defn}
For a given quantum dynamical system $(\A, \Theta, \omega)$,
ergodicity corresponds to the behavior of the discrete time-average of
two-point correlation functions and is defined by
\begin{equation}
\lim_{n\to\infty}  \frac{1}{2 T+1}\sum_{t=-T}^T  \omega(A^\dag
\Theta_t(B)C)=\omega(A C) \omega(B),
 \end{equation}
 where $A, B, C\in\A$ and
$t\in \z$.
\end{defn}

The quantum Shannon-MacMillan theorem is as follows \cite{Bjelakovictheshannon-mcmillan}:
\begin{thm}\label{qaep}
 Assume that  $(\A_\z, \Theta_\sigma, \omega)$, with $A=M_d(\com)$ as a site
algebra, is an
ergodic quantum spin-chain with mean entropy $s(\omega)$. Then, for all $\delta > 0$
there exists $N_\delta \in N$ such that for all $n \geq N_\delta$,  there is an orthogonal projection $p_n(\delta) \in \A_n$ such that
\begin{enumerate}
  \item $\omega(p_n(\delta)) = {\rm Tr}_n(\rho(n) p_n(\delta)) \geq 1 - \delta$,
  \item for all minimal projections $0 \neq  p_n \in A_n$ dominated by $p_n(\delta),\quad (p \leq p_n(\delta))$
$(1 - \delta)2^{-n(s(\omega)+\delta)} < \omega(p_n(\delta)) < 2^{−n(s(\omega)−\delta)}$ ,
  \item $2^{n(s(\omega)−\delta)} < {\rm Tr}_n(p_n(\delta)) < 2^{n(s(\omega)+\delta)}.$
  \end{enumerate}
\end{thm}
In other words, in ergodic quantum
dynamical systems with shift dynamics, there is a sequence of
projections, with high probability, such that for any sequence of
minimal projectors dominated by them, the rate of lower Gacs
complexity of them is equal to the von Neumann entropy rate
$s(\omega)$.

In the following definition the density matrices $\rho^{(n)}$ are semi-computable. Then, there exists a sequence of elementary matrices $\rho^{(n)}_m$ such that $\rho^{(n)}_m\nearrow \rho^{(n)}$ in the trace-norm. By chapter \ref{ch:semi-compuatble density matrises}, each elementary matrix $\rho^{(n)}_m$ corresponds to a natural number $a_{n m}$.
\begin{defn}
A faithful state $\omega$ on $\A_{\zz}$ is called a semi-computable (computable) state if the associated local density matrices   $\rho^{(n)}$ on $( M_d(\com))^{\otimes n}\subseteq A_{\zz}$ (\ref{Applications of Upper Gacs Complexity}) are semi-computable (computable) semi-density matrices and the function $(m, n)\to a_{n m}$ from $\n\times\n \to \n$ is  computable, $\rho^{(n)}_m\nearrow \rho^{(n)}$ and $rank \rho^{(n)}_m=n$.
\end{defn}

An important question is: are the  eigenvalues of a semi-computable semi-density matrix  $\rho^{(n)}$ semi-computables  \cite{E.Kowalski}?

\begin{thm}
Let $T$ be a compact positive  operator \footnote{The operator $T$ is called compact if there exists a sequence of operators
$T_n$  with $\dim Im(T_n) < \infty$ for all $n$, and
$\lim_{n\to \infty} T_n= T$; in the norm topology on B(H).}in $B(\mathbb{H})$ with $dim(\mathbb{H})<\infty$, and eigenvalues $\lambda_1\geq \lambda_2\geq \cdots > 0$  listed in decreasing order tending to $0$. Therefore,
$$    \lambda_k=\max_{\text{dim} V=k} \min_{v\in V-\{0\}}
\frac{<v|\,T\,|v>}{\| v\|^2},$$
and
$$ \lambda_k=\max_{\text{dim} V=k-1} \min_{v\in V^\perp-\{0\}} \frac{<v|\,T\,|v>}{\| v\|^2}.  $$
In both cases, $V$ runs over subspaces of $\mathbb{H}$ of the stated dimension, and in the first case it is assumed that $V \subseteq Ker(T)^\perp$. Moreover, If $T_n$ is a sequence of positive compact  operators such that $T_n\to T$ in norm topology on $B(\mathbb{H})$, then $T$ is a positive compact operator such that $$\lim_{n\to\infty}\lambda_k(T_n) = \lambda_k(T),$$
where $\lambda_k(T_n)$'s are eigenvalues of the $T_n$ in decreasing listed order, for each $n\in\n$.
\end{thm}
Let $\rho_n$ be a sequence of semi-computable  semi-density  matrices where $\rho_n\rightarrow \rho$. Therefore, $\rho_n$'s are  compact operators. Then,
$$\lim_{n\to\infty}\lambda_k(\rho_n) = \lambda_k(\rho).$$

On the other hand for a given semi-computable semi-density matrix $\rho_n$, there exists a sequence $\rho_{m n}$ of elementary matrices  such that $\rho_{m n}\to\rho_n$. Therefore, the eigenvalues of $\rho_n$ can be considered as limit of the eigenvalues of $\rho_{m n}$ and the eigenvectors of $\rho_n$ are also semi-computable.

Let  $U$ be a semi-computable semi-unitary operator. The operator  $U^\dag\hat{\mu} U$ may not be a semi-computable semi-density matrix and hence it may not be a universal semi-density matrix.

\begin{lem}‎\label{semi-unive}
‎Let $(A_{\zz}, \Theta_\sigma, \omega)$ be a quantum spin chain, with $A = M_d(\com)$ as its site-algebras and $\omega$  a semi-computable faithful state. Let $\rho$ be  a associated density matrix to $\omega$ and $\rho^{(n)}={\rm Tr}_{-n], [n}\rho$.
Let's define $U$ be a unitary operator with $U_n|\mu_{\ii}>=|r_{\ii}>$,  for $\ii\in\Omega_2^{(n)}$, where $\hat{\mu}=\sum_{\ii} \mu_{\ii}|\mu_{\ii}><\mu_{\ii}|$ and $\rho^{(n)}=\sum_{\ii\in\Omega_2^{(n)}} r_{\ii}|r_{\ii}><r_{\ii}| $. We have
‎
$$ \limsup_{n\to\infty }\frac{1}{n}\log{\rm Tr}(\sigma_n \hat{\mu})=\limsup_{n\to\infty }\frac{1}{n}\log{\rm Tr}(\sigma_n U_n \hat{\mu} U_n^\dag),$$
for any density matrix $\sigma^{(n)}\in\A_\zz^{(n)}$.
‎\end{lem}‎
‎\begin{proof}‎
By \ref{Applications of Upper Gacs Complexity}   elements of the sequence  $\hat{\mu}^{(n)}=P_n \hat{\mu} P_n$  are universal semi-density matrices on $( M_d(\com))^{\otimes n}$, where $P_n$ is a projection from $\mathbb{H}$ to $\mathbb{H}^n$.

 Let us consider the spectral decompositions $\hat{\mu}^{(n)}$ and $\rho^{(n)}$, respectively as follows:
$$ \sum_{\ii\in \Omega_2^{(n)}} \mu_{\ii}^{(n)} |\mu_{\ii}^{(n)}><\mu_{\ii}^{(n)}|, \quad
\sum_{\ii\in \Omega_2^{(n)}} r_{\ii} |r_{\ii}><r_{\ii}|.$$
Let us define $U_n |\mu_{\ii}^{(n)}>=|r_{\ii}>$.

Because  $\rho^{(n)}$ and $\hat{\mu}^{(n)}$ are semi-computable density matrices and hence there exist computable sequences of elementary matrices $\rho^{(n)}_m$ and $\hat{\mu}^{(n)}_l$ such that $\rho^{(n)}_m\nearrow \rho^{(n)}$ and $\hat{\mu}^{(n)}_l\nearrow \hat{\mu}^{(n)}$, respectively. Moreover, ranks of the $\rho^{(n)}_m$ and $\hat{\mu}^{(n)}_m$ are equal to $n$, for each $n\in\n$. Therefore, the operator $U_{mn}$ defined by $U_{mn}^\dag\hat{\mu}^{(n)}_m U_{mn}=\rho^{(n)}_m$ is an elementary unitary operator and the function $m\to U_{mn} $ is computable.

  On the other hand, $\hat{\mu}$ is a semi-computable density matrix and hence there exits a sequence of elementary operators $\hat{\mu}_k$ such that $\hat{\mu}_k\nearrow\hat{\mu}$ in trace-norm. But, $U_{m n}^\dag\hat{\mu}_k U_{m n}$ is a  computable sequence of elementary which convergence increasingly  to $U^\dag_{m n} \hat{\mu} U_{m n}$ and thus $U^\dag_{m n} \hat{\mu} U_{m n }$ is a semi-computable semi-density matrix.

 Let us consider the following  operator
$$ \hat{K}_m=\sum_n \frac{1}{n \log^2 n} U_{mn}^\dag \hat{\mu} U_{mn}.$$
Using Theorem \ref{thm:semi-unitary}, it is clear that $\hat{K}_m$ is a semi-computable semi-density matrix. Therefore, there exists a constant $c_m>0$ such that
$$c_m \frac{1}{n \log^2 n} U_{m n}^\dag \hat{\mu} U_{m n}\leq c_m \hat{K}_m \leq \hat{\mu}.$$
Now, we have
$$  \hat{\mu} \leq  \frac{1}{c_m} \,\, n \log^2 n\,\, U_{mn} \hat{\mu} U^\dag_{m n}  .$$
Let $\sigma^{(n)}$ be a semi-density matrix on $(M_d(\com))^{\otimes n}$. Then,
\begin{eqnarray*}
\limsup_{n\to\infty }\frac{1}{n}\log{\rm Tr}(\sigma_n \hat{\mu})&\leq& \limsup_{n\to\infty }\frac{1}{n}\log {\rm Tr}(  \frac{1}{c_m} n \log^2 n\,\, \sigma^{(n)}U_{m n} \hat{\mu} U^\dag_{m n} )\\
&\leq&
 \limsup_{n\to\infty }\frac{1}{n}\log {\rm Tr}(\sigma^{(n)} U_{m n} \hat{\mu} U_{m n}^\dag )
.
\end{eqnarray*}
On the other hand,
\begin{eqnarray*}
\left|{\rm Tr}(\sigma^{(n)} U_{m n} \hat{\mu} U_{m n}^\dag- \sigma^{(n)} U_{n} \hat{\mu} U_{ n}^\dag) \right|
&\leq&
 ||\sigma^{(n)}||{\rm Tr}\left| U_{m n} \hat{\mu} U_{m n}^\dag-  U_{n} \hat{\mu} U_{ n}^\dag \right|\\
 &\leq&
 {\rm Tr}\left| U_{m n} \hat{\mu} U_{m n}^\dag-  U_{n} \hat{\mu} U_{ n}^\dag \right|\\
 &\leq&
 {\rm Tr}\left| U_{m n} \hat{\mu} U_{m n}^\dag +U_{m n} \hat{\mu} U_{n}^\dag-U_{m n} \hat{\mu} U_{n}^\dag -  U_{n} \hat{\mu} U_{ n}^\dag \right|\\
 &\leq&
|| U_{m n} \hat{\mu}|| {\rm Tr}\left| U_{m n}^\dag - U_{n}^\dag\right| + ||\hat{\mu} U_{ n}^\dag ||
 {\rm Tr}\left|(U_{m n}- U_{n}   ) \right|\\
 &\leq&
  {\rm Tr}\left|  U_{m n}^\dag - U_{n}^\dag \right| +  {\rm Tr}\left|U_{m n}- U_{n}  \right|\\
  &\leq&
  \epsilon +\epsilon.
\end{eqnarray*}
Therefore,
$$\limsup_{n\to\infty }\frac{1}{n}\log{\rm Tr}(\sigma_n \hat{\mu})\leq \limsup_{n\to\infty }\frac{1}{n}\log {\rm Tr}(\sigma^{(n)} U_{n} \hat{\mu} U_{ n}^\dag ).$$

With this method we can also prove the other hand of the above inequality.
‎\end{proof}‎

In the following theorem we prove the extension of the Brudno's theorem in quantum dynamical systems with shift dynamics. Of course, the projections defined in
\cite{Bjelakovictheshannon-mcmillan} are replaced by new projections
which  satisfy  all the needed properties.

\begin{thm}\label{thm:brudnoq}
Let $(\A_\z, \Theta_\sigma, \omega)$, with $\A=M_d(\com)$ as a site
algebra, be an ergodic  quantum spin-chain with mean
entropy $s(\omega)$ where  $\omega$ is  faithful and semi-computable. Then, for any
$\epsilon >0$, there exists a
sequence of projections $p_n(\epsilon)\in\A_n$ and a number $N_\epsilon\in \n$ such that for
any $n\geq N_\epsilon$, we have,
\begin{enumerate}
    \item $\omega(p_n(\epsilon))={\rm Tr}(\rho^{(n)}
    p_n(\epsilon))>1-\epsilon$,
    \item  for any minimal projection $0 \neq p_n\in \A_n$ dominated by $p_n(\epsilon)$ ($p_n\leq
    p_n(\epsilon)$), we have
    $$ 2^{-n(s(\omega)+\epsilon)}\leq \omega(p_n)\leq    2^{-n(s(\omega)-\epsilon)}   .$$
    \item $(1-2^{-n \epsilon})2^{n(s(\omega)-\epsilon)+\alpha_n}< {\rm Tr}_n(p_n(\epsilon))  < 2^{n(s(\omega)+\epsilon)} . $
    \item $\lim_{n\to\infty} -\frac{1}{n}\log{\rm Tr}(\hat{\mu}
    \,p_n)=s(\omega),$
\end{enumerate}
where $\lim_{n\to \infty} \frac{\alpha_n}{n} =0$.
\end{thm}
\begin{proof}
Let $\rho^{(n)}$ be a local density matrix on the local algebra
$\A_n=M_{[0, n]}$ such that $\omega(A)={\rm Tr}_{[0, n]}(\rho^{(n)}
A)$, $A\in\A_n$. Let $\sum_{l} r^{(n)}_l |r^{(n)}_l><r^{(n)}_l|$ be
the spectral decomposition of $\rho^{(n)}$, $n\in\n$, which
is sorted  decreasingly in accordance to eigenvalues.  We also define two sets as
follow:
$$A^{(n)}_\epsilon=\{\bar{l}\in\Omega_2^{(n)}| 2^{-n(s(\omega)+\epsilon)}
   \leq r^{(n)}_l\leq 2^{-n(s(\omega)-\epsilon)}      \} ,    $$
and  $$B^{(n)}_\epsilon=\{\bi\in\Omega_2 : \mu(\ii)<
2^{-n(s(\omega)-2\epsilon)}, \ii \,\,\text{is the initial prefix of
string}\,\,\bi \},$$ where $\bar{\l}$ is the binary expansion of the
number $\l$.
 According to  \ref{compreseeing-semimeasure},
$$ \#(B_\epsilon^{(n)})^c \leq
2^{n(s(\omega)-2 \epsilon)+\alpha_n},$$ where $\alpha_n
>0 $ is a constant number and $\lim_{n\to \infty} \frac{\alpha_n}{n} =0$.
 Now, we define a sequence of
projections $p_n(\epsilon)$ on the GNS representation of $\A_n$ by
$$p_n(\epsilon)=\sum_{\ii\in A^{(n)}_\epsilon\cap
B^{(n)}_\epsilon}|r_{\ii}><r_{\ii}|.$$ Let $p_n\leq p_n(\epsilon)$
be a minimal projection on $\A_n$. Then, its representation is as
follows
$$p_n=|\psi_n><\psi_n|, \quad |\psi_n>=\sum_{\ii\in A^{(n)}_\epsilon\cap
B^{(n)}_\epsilon} c_{\ii} |r_{\ii}> \text{where} \sum_{\ii\in
A^{(n)}_\epsilon\cap B^{(n)}_\epsilon} |c_{\ii}|^2=1$$ In the
definition of $p_n(\epsilon)$, we restrict ourselves to the set $
A^{(n)}_\epsilon\cap B^{(n)}_\epsilon$ which is smaller than
$A^{(n)}_\epsilon$ in \ref{compreseeing-semimeasure}.

\noindent\textbf{Proof of  1:}
\begin{eqnarray*}
{\rm Tr}(\rho^{(n)} p_n(\epsilon))
 &\geq&
\sum_{\ii\in A^{(n)}_\epsilon} r_{\ii}-\sum_{\ii\in
A^{(n)}_\epsilon\backslash B^{(n)}_\epsilon} r_{\ii}\\
&\geq& 1-\epsilon - \sum_{\ii\in A^{(n)}_\epsilon\backslash
B^{(n)}_\epsilon} 2^{-n(s(\omega)-\epsilon)}\\
&\geq&
 1-\epsilon - 2^{-n(s(\omega)-\epsilon)}\#(B^{(n)}_\epsilon)^c\\
 &\geq&
 1-\epsilon - 2^{-n(s(\omega)-\epsilon)} 2^{n(s(\omega)-2\epsilon)+\alpha_n}\\
 &\geq&
 1-\epsilon - 2^{-n\epsilon+\alpha_n}\\
 &\geq&
1-\frac{3\epsilon}{2}.
\end{eqnarray*}
In the fifth inequality above,, it is clear that $2^{-n \epsilon +
\alpha_n}\to 0$.

 \noindent\textbf{ Proof of  2:}

According to  Theorem \ref{qaep}, we have
\begin{eqnarray*}
\omega(p_n) &=& \sum_{\ii\in A^{(n)}_\epsilon\cap B^{(n)}_\epsilon}
r_{\ii} |c_{\ii}|^2\\
&\leq&
 2^{-n(s(\omega)-\epsilon)} \sum_{\ii\in A^{(n)}_\epsilon\cap B^{(n)}_\epsilon}
|c_{\ii}|^2\\
&\leq&
 2^{-n(s(\omega)-\epsilon)} ,
\end{eqnarray*}
 and
\begin{eqnarray*}
 \omega(p_n)&=&\sum_{\ii\in A^{(n)}_\epsilon\cap
B^{(n)}_\epsilon} r_{\ii}|<r_{\ii}|\psi^{(n)}>|^2\\
& \geq& 2^{-n(s(\omega)+\epsilon)}\sum_{\ii\in A^{(n)}_\epsilon\cap
B^{(n)}_\epsilon} |c_{\ii}|^2 \\
&\geq&
 2^{-n(s(\omega)+\epsilon)}.
\end{eqnarray*}
\noindent\textbf{Proof of 3:}
 \begin{eqnarray*}
 {\rm Tr}_n(p_n(\epsilon))
 &\leq&
 \sum_{\ii\in A^{(n)}_\epsilon\cap
B^{(n)}_\epsilon} 1\\
 &\leq&
  \sum_{\ii\in A^{(n)}_\epsilon} 1\\
  &\leq&
2^{n(s(\omega)+\epsilon)}.
\end{eqnarray*}
We also have
\begin{eqnarray*}
 {\rm Tr}_n(p_n(\epsilon))
 &\geq&
 \sum_{\ii\in A^{(n)}_\epsilon\cap
B^{(n)}_\epsilon} 1\\
 &\geq&
  \sum_{\ii\in A^{(n)}_\epsilon} 1 - \sum_{\ii\in
A^{(n)}_\epsilon\backslash B^{(n)}_\epsilon} 1\\
  &\geq&
2^{n(s(\omega)-\epsilon)}- 2^{n(s(\omega)-2\epsilon) + \alpha_n}\\
&\geq& (1-2^{-n \epsilon})2^{n(s(\omega)-\epsilon) +\alpha_n}.
\end{eqnarray*}
 \noindent\textbf{Proof of 4:}
 Since the quantum system is semi-computable, thus the
density matrix
$$\eta=\sum_{n=2}^\infty\delta(n) \rho^{(n)}, \quad \text{where}\,\,\omega(A)={\rm
Tr}(\rho^{(n)}A),$$ is also a semi-computable semi-density matrix
and hence there exists a constant number $c>0$ such that $c\delta(n)
\rho^{(n)}\leq c\eta\leq \hat{\mu}$, for all $n\in\n$. Therefore,
\begin{eqnarray*}
\limsup_{n\to\infty}-\frac{1}{n}\log{\rm Tr}(\hat{\mu}p_n)
 &\leq&
\limsup_{n\to\infty}-\frac{1}{n}\log{\rm Tr}(\rho^{(n)}
p_n)\\
 &\leq&
\limsup_{n\to\infty}-\frac{1}{n}\log\omega(p_n)\\
& \leq&
\limsup_{n\to\infty}-\frac{1}{n}\log  \left((1-2^{-n\epsilon}) 2^{-n(s(\omega)+\epsilon)+\alpha_n} \right)  \\
&\leq&
 s(\omega)+\epsilon +\alpha_n.
\end{eqnarray*}
Thus,
$$ \limsup_{n\to\infty} -\frac{{\rm Tr}(\hat{\mu} p_n)}{n}\leq  s(\omega)+\epsilon $$
%$\beta_{\epsilon, n}=\min\{\log_2{\rm
%Tr}(q): \A_n\ni q=q^\dag=q^2 , \{\rm Tr}(\rho^{(n)} q)\geq
%1-\epsilon \}$ then
%$${\rm Tr}(\hat{\mu} p_n)\leq {\rm Tr}(p_n)\leq      $$
Since $\omega$ is a faithful state then the number of  eigenvectors
of $\rho^{(n)}=\rho\upharpoonleft_{\A_n}$ is exactly $2^n$. Hence,
the operator $\hat{T}^{(n)}=\sum_{\ii\in\Omega_2^{(n)}} \mu(\ii) |r_{\ii}><r_{\ii}|$ is a
 semi-computable semi-density matrix. Therefore,  there exists a constant number $c_T>0$ such that $c_T
\hat{T}\leq \hat{\mu}$. Let us define the linear map
$U_n|\mu_{\ii}>=|r_{\ii}>$, where $\sum_{\ii}
\mu_{\ii}|\mu_{\ii}><\mu_{\ii}|$ is the spectral decomposition of
$\hat{\mu}$. Now,  by Lemma \ref{semi-unive}, we have
%$$Q_n= \sum_{\ii\in B_n(\epsilon)} |r_{\ii}><r_{\ii}| , \quad  B_n(\epsilon)=\{\ii: \mu(\ii)\geq 2^{-n(s(\omega)+2\epsilon)},\,\, \ii\in A^{(n)}_\epsilon   \}   $$
\begin{eqnarray*}
\liminf_{n\to\infty}-\frac{1}{n}\log{\rm Tr}(\hat{\mu}p_n)
 &\geq&
\liminf_{n\to\infty}-\frac{1}{n}\log{\rm Tr}( U_n \hat{\mu}^{(n)} U_n^\dag p_n)\\
&\geq&
 \liminf_{n\to\infty}-\frac{1}{n}\log{\rm Tr}(\hat{T}^{(n)}p_n)\\
 &\geq&
\liminf_{n\to\infty}-\frac{1}{n}\log(\sum_{\ii\in \Omega^{(n)}_2}
\mu(\ii)<r_{\ii}|p_n|r_{\ii}>)\\
&\geq&
\liminf_{n\to\infty}-\frac{1}{n}\log(2^{-n(s(\omega)-2\epsilon)}
\sum_{\ii\in \Omega^{(n)}_2}<r_{\ii}|p_n|r_{\ii}>)\\
&\geq&
 s(\omega)-2\epsilon.
\end{eqnarray*}
\end{proof}
The main important quantum correlation is entanglement which dose't holds in classical dynamical systems.

Now, we say that the density matrix $\rho$ on the Hilbert space $\mathcal{H}_{XY}:=\mathbb{H}_X\otimes \mathbb{H}_Y$ associated with a composite system with the two subsystems $X$ and $Y$ is separable if
$$\rho =\sum_{(i_1, i_2)\in I_1 \times I_2 } \lambda_{i_1 i_2}\rho^1_{i_1} \otimes \rho^2_{i_2},\quad \lambda_{i_1 i_2}\geq 0, \quad \sum_{(i_1, i_2)\in I_1 \times I_2 } \lambda_{i_1 i_2}=1\,. $$
The density matrix $\rho$ is called entangled if it is not a separable state. 

For example the density matrix $\rho=|\psi><\psi|$, $|\psi>=\frac{|00> +|11>}{\sqrt{2}}$ on the Hilbert space $M_2(\com)\otimes M_2(\com)$ is  entanglement. Indeed, we cannot write $|\psi>=|a>|b>$ where $|a>$ and $|b>$ are states on  $M_2(\com)$.

In the following theorem, we will prove that  entanglement in pure
states dose not change the von Neumann entropy rate. In this case,
we  consider the product of two universal semi-measures which is not in
general a universal semi-measure, instead of, a universal semi-measure
on space of tensor product of two Hilbert spaces related to the GNS
representation. Then, we show that the Gacs entropy rate is also
equal to   two times  von Neumann entropy rate. Thus, it
shows us that the entanglement dose not exceed of the lower Gacs
entropy rate.  Indeed, we know that  entanglement is a quantum correlation and when we consider a many number of spins in large scale in the classical dynamical systems, or thermodynamical limit, the following theorem tells us that the effects of entanglement and pure states are equal.

\begin{thm}
Let $((\A_\z)_{XY}, \Theta'_\sigma, \omega_{XY})$ be a composite
quantum spin chain consisting of two ergodic spin chains
$((\A_\z)_X, \Theta_\sigma, \omega)$ and $((\A_\z)_Y, \Theta_\sigma,
\omega)$ with the same mean entropy $s(\omega)$ and  faithful state
$\omega$ where $\omega_{XY}=\omega\otimes \omega$ and $\Theta'_\sigma=\Theta_\sigma\otimes \Theta_\sigma$. Let $\rho_{[-n,
n]}$ be a semi-computable semi-density matrix on both
$\mathbb{H}_{X,Y}$ and consider the universal semi-measures
$\hat{\mu}_X$ and $\hat{\mu}_Y$ on $\mathbb{H}_{X,Y}$.
Then, there is a sequence of projectors $p_{2n}(\epsilon)\in
\mathbb{H}_{2n}\subseteq \mathbb{H}_{XY}$ such
that for a sequence of minimal density matrices
$\sigma^{(2n)}\leq p_{2n}(\epsilon)$  on $\mathbb{H}_{X
Y}$, one has:
\begin{equation}
\lim_{n\to\infty} -\frac{1}{n}\log{\rm Tr}(\hat{\mu}_X \otimes
\hat{\mu}_Y \sigma^{(2 n)})=s(\omega).
\end{equation}
\end{thm}
\begin{proof}
%Let a fix $\epsilon>0$. By the theorem (\ref{thm:brudnoq}) there are
%two sequences  $p_{1n}(\epsilon)$ and $p_{2n}(\epsilon)$  of
%projections on $\mathbb{H}_X$ and $\mathbb{H}_Y$, respectively,
%where $\omega_1(p_{X n}(\epsilon))> 1-\epsilon $ and $\omega_2(p_{Y
%n}(\epsilon))> 1-\epsilon$. Let $p_n(\epsilon)=p_{X n}\otimes p_{Y
%n}$.
 %  $$\omega(p_n(\epsilon))=(\omega_1\otimes \omega_2) (p_n(\epsilon))=\omega_1(p_{X n}(\epsilon)) \omega_2(p_{Y
  %  n}\epsilon) \geq (1-\epsilon)^2 .$$
It is clear that the sequence of projectors $(p_X)_n(\epsilon)
\otimes (p_Y)_n(\epsilon)$ satisfy  the conditions $1, 2, 3$  of
 Theorem \ref{thm:brudnoq}, where $(p_X)_n(\epsilon)$ and
$(p_Y)_n(\epsilon)$ are  projections related to the mentioned
conditions on $\mathbb{H}_X$ and $\mathbb{H}_Y$, respectively.
 Now, consider the sequence of minimal projections     $\sigma^{(2n)}=|\psi^{(2n)}><\psi^{(2n)}| \leq p_{Xn} \otimes
p_{Yn}$. According to the proof of Theorem \ref{thm:brudnoq}, we can
write $|\psi^{(2n)}>$ as follows:
$$|\psi^{(2n)}>=\sum_{\ii,
\jj\in A_\epsilon^{(n)}\cap B_\epsilon^{(n)}} a_{\ii\jj} |r_{\ii}
s_{\jj}>, \quad\sum_{\ii, \jj\in A_\epsilon^{(n)}\cap
B_\epsilon^{(n)}} |a_{\ii\jj}|^2=1,$$ where $A_\epsilon^{(n)}$ and
$B_\epsilon^{(n)}$ are defined in Theorem \ref{thm:brudnoq}.
%$$\sigma^{(n)}=\sum_{\ii, \jj, \textbf{K}^{(n)}, \textbf{l}^{(n)}} a_{\ii, \jj, \textbf{K}^{(n)}, \textbf{l}^{(n)}} |\ii \textbf{k}^{(n)}><\jj \textbf{l}^{(n)}|.$$
The remaining of the proof is like  that  of Theorem
\ref{thm:brudnoq},
\begin{eqnarray*}
& &\limsup_{n\to\infty} -\frac{1}{n}\log{\rm
Tr}(\sigma^{(2n)}\hat{\mu}_X\otimes \hat{\mu}_Y)\leq\\
 &\leq&
\limsup_{n\to\infty} -\frac{1}{n}\log{\rm Tr}(\sigma^{(2n)} \hat{\mu}_X\otimes \hat{\mu}_Y)\\
 &\leq&
\limsup_{n\to\infty}  -\frac{1}{n}\log{\rm Tr}(\sigma^{(2n)} \rho_X\otimes \rho_Y)\\
 &\leq&
\limsup_{n\to\infty}  -\frac{1}{n}\log   < \psi^{(2n)} |\rho_X\otimes \rho_Y| \psi^{(2n)} >\\
 &\leq&
\limsup_{n\to\infty} -\frac{1}{n}\log\sum_{\ii, \jj} r_{\ii} s_{\jj} |a_{\ii\jj}|^2\\
 &\leq&
\limsup_{n\to\infty} -\frac{1}{n}\log\left(\sum_{\ii, \jj \in
A_\epsilon^{(n)}\cap B_\epsilon^{(n)}} |a_{\ii, \jj}|^2
 2^{-2 n(s(\omega)+\epsilon)}\right )\\
 &\leq&
2 s(\omega)+ 2 \epsilon,
\end{eqnarray*}
and
 %Now, consider the map $f:\Omega_2^{(2n)}\to\rr$ where
%$$f(\ii\jj)=\frac{1}{n^2}|a_{\ii, \jj}|^2 r_{\ii} s_{\jj}.$$
%It is clear that $f$ is a semi-computable semi-measure and hence
%there exists a constant $C>0$ such that $c f(\ii\jj)\leq
%\hat{\mu}_X(\ii\jj)$. By using theorem []
\begin{eqnarray*}
& &\liminf_{n\to\infty}-\frac{1}{n}\log{\rm Tr}(\hat{\mu}_X\otimes
\hat{\mu}_Y \sigma^{2n})\geq\\
 &\geq&
 \liminf_{n\to\infty}-\frac{1}{n}\log{\rm Tr}(\hat{T}_X\otimes
\hat{T}_Y \sigma^{2n})\\
 &\geq&
\liminf_{n\to\infty}-\frac{1}{n}\log\left(\sum_{\ii, \jj\in
\Omega_2}
\mu(\ii) \mu(\jj) <r_{\ii}s_{\jj}|\sigma^{2n}|r_{\ii} r_{\jj}>\right)\\
&\geq&
\liminf_{n\to\infty}-\frac{1}{n}\log\left(2^{-2
n(s(\omega)-2\epsilon)} \sum_{\ii, \jj \in A_\epsilon^{(n)}\cap
B_\epsilon^{(n)}}<r_{\ii}s_{\jj}|\sigma^{2n}|r_{\ii} r_{\jj}>\right)\\
&\geq&
 2 s(\omega)-4 \epsilon.
\end{eqnarray*}
% $p_n(\epsilon)$ is also satisfy in the conditions of theorem
% (\ref{thm:brudnoq}).
%$$ (1-2^{-n \epsilon})^2
%2^{2 n(s(\omega)-\epsilon)}\leq \omega(p_n(\epsilon)) = {\rm
%Tr}_n(p_n(\epsilon) \rho_X\otimes \rho_Y)
 %\leq
 %{\rm Tr}_n(p_{Xn}(\epsilon)\rho_X){\rm Tr}_n(p_{Yn}(\epsilon)\rho_Y)
 %\leq
%2^{2 n(s(\omega)+\epsilon)}$$
%$$-\frac{1}{n}\log{\rm Tr}(\sigma^{(n)} \mu_X\otimes \mu_Y)\geq 2 s(\omega)-\frac{\delta}{2n}.$$
In the first inequality, we use Lemma \ref{thm:brudnoq}.
\end{proof}

\addcontentsline{toc}{chapter}{Conclusion}%
\chapter*{Conclusion}
\goodbreak
In this work we have extended the notions of computability, semi-computability, semi-computable
vector states, and semi-computable density matrices to infinite dimensional Hilbert
spaces. These extensions are necessary to describe algorithmically by classical Turing machines
quantum systems with infinitely many degrees of freedom. In this paper we have applied them to the discussion, from a computer science point of view, of the complexity of quantum spin chains with the shift dynamics.

In classical information theory, Brudno has proved a relation between the Kolmogorov-Sinai
dynamical entropy of ergodic time-evolutions and the algorithmic complexity per unit time step
of all almost trajectories. In quantum information theory there are different extensions of both the Kolmogrov-Solomonoff-Chatin algorithmic complexity and of the Kolmogorov-Sinai dynamical
entropy: their possible relations can be found in ~\cite{Benattib}.

The techniques developed in this thesis have been applied to quantum spin chains. They allowed
us to show that the Gacs algorithmic entropy per site of translation invariant states is equal to the von Neumann entropy rate. This could be done by removing an unnecessary condition in a previous proof of the same relations ~\cite{Benatti}.

One proposal to extend the Brudno's theorem is to  consider the classical version  of the concepts of the Gacs complexities based on semi-computable semi-measure functions using the classical Brudno's theorem. The essential obstacle to extend the Beoudno's theorem  based on associated symbolic dynamical systems is that we have no appropriate   meaning of trajectory in the associated symbolic dynamical systems. But, we have given a short proof of the Brudno's theorem in classical dynamical systems.

At the end, we have shown an extension of the Brudno's theorem using the quantum Shannon-Mac Millan theorem which is directly derived without using the classical Brudno's theorem, where "almost every for all trajectories" in the classical case is replaced by a sequence of high probabilities projections. Furthermore, it has shown that  entanglement and pure density matrices have the same role in the thermodynamic limit. Roughly speaking, rate of the log of the trace of the tensor product of universal semi-density matrices, associated to Hilbert spaces of subsystems, times density matrices pure or entangled, are equal to rate of von-Numann entropy of the state.

%%%%%%%%%%%% -----------------------------------------------------------------------
%%%%%%%%%%%% If do not have appendix then comment following 3 lines
%\appendix
%\include{peyvast1}
%%%%%%%%%%%%\include{peyvast2}
%%%%%%%%%%%% -------------------------------------------------------
%%%%%\singlespacing

\normalsize
\baselineskip=1cm
% مراجع خود را در این قسمت وارد کنید

\bibliographystyle{plain}
\bibliography{XBib}

%\addcontentsline{toc}{chapter}{Appendix}%

\end{document}